\documentclass{lmcs} 
\pdfoutput=1

\usepackage{lastpage}
\lmcsdoi{20}{4}{10}
\lmcsheading{}{\pageref{LastPage}}{}{}%
{Sep.~18,~2023}{Nov.~08,~2024}{}


\usepackage{hyperref,cleveref}
\usepackage{amsmath}
\usepackage{amssymb}
\usepackage{amsthm}
\usepackage{amscd}
\usepackage{amsfonts}
\usepackage{float}
\usepackage{mathtools}
\usepackage[utf8]{inputenc}
\usepackage{color, colortbl}
\usepackage{algorithm}
\usepackage[noend]{algpseudocode}
\usepackage{graphicx}
\usepackage{qtree}
\usepackage{multicol}
\usepackage{booktabs}
\usepackage{soul}
\usepackage{xspace}
\usepackage{multirow}

\usepackage{tikz}
\usetikzlibrary{calc}
\usetikzlibrary{automata, positioning, arrows, shapes.geometric}
\usepackage[font=small,labelfont=bf]{caption}
\usepackage{stmaryrd}

\theoremstyle{plain} 

\begin{document}

\title[Constructing Concise Characteristic Samples for Omega Acceptors]{Constructing Concise Characteristic Samples \\ for Acceptors of Omega Regular Languages}

\author[D.~Angluin]{Dana Angluin\lmcsorcid{0000-0002-6907-2999}}	
\address{Yale University, New Haven, CT, USA}	
\email{dana.angluin@yale.edu}  

\author[D.~Fisman]{Dana Fisman\lmcsorcid{0000-0002-6015-4170}}	
\address{Ben-Gurion University, Be'er-Sheva, Israel}	
\email{dana@cs.bgu.ac.il}  





\theoremstyle{plain} \numberwithin{equation}{section}
\newtheorem{theorem}[thm]{Theorem}
\newtheorem{proposition}[theorem]{Proposition}
\newtheorem{corollary}[theorem]{Corollary} 
\newtheorem{lemma}[theorem]{Lemma}

\theoremstyle{definition}
\newtheorem{example}[theorem]{Example}

\renewcommand{\sectionautorefname}{Section}
\renewcommand{\subsectionautorefname}{Section}
\renewcommand{\propositionautorefname}{Proposition}
\renewcommand{\corollaryautorefname}{Corollary}
\renewcommand{\lemmaautorefname}{Lemma}
\newcommand{\claimautorefname}{Claim}

\newcommand{\commentout}[1]{}

\newcommand{\commentbody}[3]{\colorbox{#2}{\textsc{#3:}}#1{\colorbox{#2}{.}}}
\newcommand{\commentft}[3]{\colorbox{#2}{...}\footnote{\colorbox{#2}{\textsc{#3:}}#1{\colorbox{#2}{.}}}}

\newcommand{\df}[1]{\commentft{#1}{yellow}{df}} 
\newcommand{\da}[1]{\commentft{#1}{blue!30!white}{da}}

\newcommand{\replace}[2]{\textcolor{blue}{\st{#1} #2}}
\newcommand{\remove}[1]{\textcolor{blue}{\st{#1}}}
\newcommand{\add}[1]{\textcolor{blue}{#1}}

\newcommand{\sema}[1]{{\llbracket}#1{\rrbracket}}

\newtheorem{mytheorem}[thm]{Theorem}

\newtheorem{myclaim}[theorem]{Claim}
\newtheorem{mycorollary}[theorem]{Corollary}
\newtheorem{myproposition}[theorem]{Proposition}

\newcommand{\forest}[1]{\ensuremath{\mathcal{#1}}}

\newcommand{\alg}[1]{\ensuremath{\mathbf{#1}}}
\newcommand{\aut}[1]{\mathcal{#1}}
\newcommand{\rep}[1]{\mathcal{#1}}

\newcommand{\class}[1]{\ensuremath{\mathbb{#1}}}
\newcommand{\famF}{\ensuremath{\mathcal{F}}}
\newcommand{\famS}{\ensuremath{\mathcal{S}}}
\renewcommand{\inf}{{\textsl{inf}}}
\newcommand{\occ}{{\textsl{occ}}}
\newcommand{\infss}[1]{\ensuremath{{\inf}_{#1}}}

\newcommand{\polarity}{{\textsl{polarity}}}

\newcommand{\width}{\textit{width}}

\newcommand{\SCCs}{\textit{SCCs}}
\newcommand{\maxSCCs}{\textit{maxSCCs}}
\newcommand{\minStates}{\textit{minStates}}
\newcommand{\children}{\textit{children}}
\newcommand{\size}{\textit{size}}
\newcommand{\maj}{\mathcal{C}_\textit{maj}}
\newcommand{\sample}{\textit{sample}}
\newcommand{\suffixes}{\textit{suffixes}}
\newcommand{\access}{\textit{access}}

\newcommand{\DM}{\class{DMA}}
\newcommand{\IB}{\class{IBA}}
\newcommand{\IC}{\class{ICA}}
\newcommand{\IP}{\class{IPA}}
\newcommand{\IR}{\class{IRA}}
\newcommand{\IS}{\class{ISA}}
\newcommand{\IM}{\class{IMA}}
\newcommand{\IX}{\class{IXA}}
\newcommand{\FDFA}{\class{FDFA}}

\newcommand{\ilptd}{\textsc{ilptd}\xspace}

\newcommand{\NB}{\class{NBA}}
\newcommand{\NBk}[1]{\ensuremath{\NB_{#1}}}
\newcommand{\DMk}[2]{\ensuremath{\DM_{#2}^{#1}}}

\algrenewcommand\algorithmicindent{2.0em}%
\algnewcommand\algorithmicswitch{\textbf{switch}}
\algnewcommand\algorithmiccase{\textbf{case}}
\algdef{SE}[SWITCH]{Switch}{EndSwitch}[1]{\algorithmicswitch\ #1\ \algorithmicdo}{\algorithmicend\ \algorithmicswitch}%
\algdef{SE}[CASE]{Case}{EndCase}[1]{\algorithmiccase\ #1}{\algorithmicend\ \algorithmiccase}%
\algtext*{EndSwitch}%
\algtext*{EndCase}%
\renewcommand{\algorithmicrequire}{\textbf{Input:}}
\renewcommand{\algorithmicensure}{\textbf{Output:}}

\newcommand{\buchi}{B\"uchi}
\newcommand{\cobuchi}{coB\"uchi}
\newcommand{\initstate}{\ensuremath{q_\iota}}
\newcommand{\A}{{\aut{A}}}
\newcommand{\M}{{\aut{M}}}
\newcommand{\R}{{\aut{R}}}
\newcommand{\la}{\langle}
\newcommand{\ra}{\rangle}
\newcommand{\naturals}{\mathbb{N}}

\newcommand{\fdfasubscript}[1]{\textsc{#1}}
\newcommand{\syntactic}{\fdfasubscript{s}}
\newcommand{\recurrent}{\fdfasubscript{r}}
\newcommand{\limit}{\fdfasubscript{l}}
\newcommand{\periodic}{\fdfasubscript{p}}

\newcommand{\ctlstar}{CTL$^*$}
\newcommand{\fairctl}{fair-CTL}

\tikzset{
	setofstates/.style={
		rectangle,
		rounded corners,
		draw=black,
		text centered,
		inner sep=7
	},
}
\tikzset{
	ministate/.style={align=center,
		state,inner sep=0pt, minimum size=0pt},
}
\tikzset{
	paradigm/.style={
		rectangle,
		rounded corners,
		draw=black,
		text centered,
		inner sep=7,
		minimum width=4.2cm,
		minimum height = 1.8cm,
		fill=gray!5
	},
}

\newcommand{\emptyst}{\ensuremath{\phantom{\circ}}}

\begin{abstract}
  \noindent 
  A characteristic sample for a language $L$ and a learning algorithm $\alg{L}$ is a finite sample of words $T_L$ labeled by their membership in $L$ such that  for any sample $T \supseteq T_L$ consistent with $L$, on input $T$ the learning algorithm $\alg{L}$ returns a hypothesis equivalent to $L$.
  Which omega automata have characteristic sets of polynomial size, and can these sets be constructed in polynomial time? 
  We address these questions here.

  In brief, non-deterministic omega automata of any of the common types, in particular B\"uchi,
  do not have characteristic samples of polynomial size. 
  %
  For deterministic omega automata that are isomorphic to their right congruence automata, the \emph{fully informative languages}, polynomial time algorithms for constructing characteristic samples and learning from them are given.
  
  The algorithms for constructing characteristic sets in polynomial time for the different omega automata (of types B\"uchi, coB\"uchi, parity, Rabin, Street, or Muller), require deterministic polynomial time algorithms for (1) equivalence of the respective omega automata, and (2) testing membership of the language of the automaton in the informative classes, which we provide.  
\end{abstract}

\maketitle

\tableofcontents

\section{Introduction}
With the growing success of machine learning in efficiently solving a wide spectrum of problems, we are witnessing an increased use of machine learning techniques in formal methods for system design. One thread in recent literature uses general purpose machine learning techniques for obtaining more efficient verification/synthesis algorithms.
Another thread, following the \emph{automata theoretic approach to verification}~\cite{Vardi95,KupfermanVW00} works on developing grammatical inference algorithms for verification and synthesis purposes. \emph{Grammatical inference} (aka \emph{automata learning}) refers to the problem of automatically inferring from examples a finite representation (e.g.\@  an automaton, a grammar, or a formula) for an unknown language~\cite{DelaHiguera2010}.  The term \emph{model learning}~\cite{Vaandrager17} was coined for the task of learning an automaton model for an unknown system. A large body of works has developed learning techniques for different automata types (e.g.\@  visibly-pushdown automata~\cite{KumarMV06},
I/O automata~\cite{Aarts2010}, register automata~\cite{HowarSJC12}, symbolic automata~\cite{DrewsD17}, program automata~\cite{ManevichS18}, probabilistic grammars~\cite{NitayFZ21}, lattice automata~\cite{FismanS22}) and has shown its usability in a diverse range of tasks.\footnote{E.g.\@ , tasks such as black-box checking~\cite{PeledVY99},
	specification mining~\cite{AmmonsBL02},
	assume-guarantee reasoning~\cite{CobleighGP03}, 
	regular model checking~\cite{HabermehlV05}, 
	learning verification fixed-points~\cite{VardhanSVA05}, 
	learning interfaces~\cite{NamA06},
	analyzing botnet protocols~\cite{ChocSS10} or smart card readers~\cite{ChaluparPPR14},
	finding security bugs~\cite{ChaluparPPR14},
	error localization~\cite{ChapmanCKKST15}, and
	code refactoring~\cite{MargariaNRS04,SchutsHV16}.}

In grammatical inference, the learning algorithm does not learn a \emph{language}, but rather a finite \emph{representation} of it. The complexity of learning algorithms may vary greatly by switching representations. For instance, if one wishes to learn regular languages, she may consider representations using deterministic finite automata (DFAs), non-deterministic finite automata (NFAs), regular expressions, linear grammars, etc. Since the translation results between two such formalisms are not necessarily polynomial, a polynomial learnability result for one representation does not necessarily imply  a polynomial learnability result for another representation. Let $\class{C}$ be a class of representations $\rep{C}$ with a size measure $\size(\rep{C})$ (e.g.\@  for DFAs the size measure can be the number of states in the minimal DFA). We extend $\size(\cdot)$ to the languages recognized by representations in $\class{C}$ by defining $\size(L)$ to be the minimum of $\size(\rep{C})$ over all $\rep{C}$ representing $L$.  In this paper we restrict attention to automata representations, namely, \emph{acceptors}.

There are various learning paradigms considered in the grammatical inference literature, roughly classified into \emph{passive} and \emph{active}. We mention here the  two central ones. 
In \emph{passive learning} the model of \emph{learning from finite data} refers to the following problem: given a finite sample $T\subseteq \Sigma^* \times \{0,1\}$ of labeled words, a learning algorithm $\alg{L}$ should return an acceptor $\aut{C}$ that agrees with the sample $T$. That is,  for every $(w,l)\in T$ the following holds: $w\in\sema{\aut{C}}$ iff $l=1$ (where $\sema{\aut{C}}$ is the language accepted by $\aut{C}$). 
The class $\class{C}$ is \emph{identifiable in the limit using polynomial time and data}  if and only if there exists a polynomial time algorithm $\alg{L}$ that takes as input a labeled sample $T$ and outputs an acceptor $\rep{C}\in \class{C}$ that is consistent with $T$, and $\alg{L}$ also satisfies the following condition.  If $L$ is any language recognized by an automaton from class $\class{C}$, then there exists a labeled sample $T_L$ consistent with $L$ of length bounded by a polynomial in $\size(L)$, and for any labeled sample $T$ consistent with $L$ such that $T_L \subseteq T$, on input $T$ the algorithm $\alg{L}$ produces an acceptor $\rep{C}$ that recognizes $L$.  In this case, $T_L$ is termed a \emph{characteristic sample} for the algorithm $\alg{L}$.  
The definition of identifiability in the limit using polynomial time and data relates to learning paradigms considering a \emph{teacher-learner} pair~\cite{GoldmanM96}.
While identification in the limit using polynomial time and data does
not require that the characteristic set can be computed in polynomial time, if one is interested in devising a teacher that can train a learner, having a teacher that can compute a characteristic set in polynomial time is also desired. 
The definition of \emph{efficiently teachable} adds this requirement. 
In \Cref{sec:notions-learnability-teachability} we define several notions related to efficient teachability and learnability, the stronger one is \emph{efficiently teachable}.
The question which representations of regular $\omega$-languages are efficiently teachable is the focus of this paper.

In \emph{active learning} the  model of \emph{query learning}~\cite{Angluin87} assumes the learner communicates with an \emph{oracle}  that can answer certain types of queries about the language.  The most common types of queries are \emph{membership queries} (is $w\in L$ where $L$ is the unknown language) and \emph{equivalence queries} (is $\sema{\aut{A}}=L$ where $\aut{A}$ is the current hypothesis for an acceptor recognizing $L$). Equivalence queries are typically assumed to return a counterexample, i.e. a word  in $\sema{\aut{A}}\setminus L$ or in $L\setminus\sema{\aut{A}}$.

With regard to $\omega$-automata (automata on infinite words) most of the works consider \emph{query learning} using \emph{membership queries} and \emph{equivalence queries}. The representations learned so far include:
$(L)_\$$~\cite{FarzanCCTW08old}, a non-polynomial reduction to finite words; families of DFAs (\class{FDFA})~\cite{AngluinF14,AngluinF16,AngluinBF16,LiCZL17}; strongly unambiguous B\"uchi automata (\class{SUBA})~\cite{AngluinAF20}; mod-2-multiplicity automata (\class{M}2\class{MA})~\cite{AngluinAFG22}; and deterministic weak parity automata (\class{DWPA})~\cite{MalerP95}. Among these only the latter two are known to be learnable in polynomial time using membership queries and proper equivalence queries.\footnote{
	Query learning with an additional type of query, \emph{loop-index queries}, was studied for deterministic B\"uchi automata~\cite{MichaliszynO20}.}
We show in \Cref{sec:positive-results} that the classes 
\class{M}2\class{MA}, \class{SUBA} and \class{DWPA} are efficiently teachable. 

One of the main obstacles in obtaining a  polynomial  learning algorithm for regular $\omega$-languages is that they do not in general have a Myhill-Nerode characterization; that is, there is no theorem correlating the states of a minimal automaton of some of the common automata types (B\"uchi, parity, Muller, etc.) to the equivalence classes of the right congruence of the language. The \emph{right congruence relation} for an $\omega$-language $L$ relates  two finite words $x$ and $y$ iff there is no infinite suffix $z$ differentiating them, that is $x\sim_L y$  (for $x,y\in\Sigma^*$) iff $ \forall z\in\Sigma^\omega.\ xz\in L \iff yz \in L$.
The quest for finding a polynomial query learning algorithm for a subclass of the regular $\omega$-languages, led to studying  subclasses of languages for which such a relation holds. These languages are termed \emph{fully informative}~\cite{AngluinF18}. 
We use $\IB,\IC,\IP,\IR,\IS,\IM$ to denote the classes of languages that are fully informative of type B\"uchi, coB\"uchi, parity, Rabin, Streett and Muller, respectively. 
A language $L$ is said to be fully informative of type $\class{X}$ for $\class{X}\in\{\class{B},\class{C},\class{P},\class{R}, \class{S}, \class{M}\}$ if there exists a deterministic automaton of type $\class{X}$ that recognizes $L$ and is isomorphic to the automaton derived from $\sim_L$. While many properties of these classes are now known, in particular that they span the entire  hierarchy of regular $\omega$-languages~\cite{Wagner75}, a polynomial learning algorithm for them is not known.

We show (in Sections~\ref{section:outline-positive}-\ref{sec:T-Acc-for-IRAs-ISAs}) that the classes $\IB,\IC,\IP,\IR,\IS,\IM$ can be identified in the limit using polynomial time and data. We further show (in \Cref{sec:poly-time-char-samples}) that  there is  a polynomial time algorithm to compute a characteristic sample given an acceptor $\rep{C}\in \class{IXA}$.
To show that these classes are also efficiently teachable we need polynomial time algorithms for inclusion and equivalence of automata of these types, that also return shortlex counterexamples in case of inequivalence.\footnote{The formal definition of shortlex is deferred to \Cref{sec:prelim}.} 
Such an algorithm is known to exist for the classes \class{NBA}, \class{NCA}, \class{NPA}, since these classes have inclusion algorithms in NL~\cite{Schewe10}.
For the other classes a polynomial-time algorithm can be obtained following a reduction to model checking a certain fragment of~\ctlstar\ formulas~\cite{ClarkeDK93}. However this reduction does not yield shortlex counterexamples. 
We provide such algorithms in Sections~\ref{sec:inclusion-algorithms}-\ref{sec:inclusion-DMAs}.

The last part of this paper (Sections~\ref{sec:computing-right-congruence-automaton}-\ref{sec:testing-membership-in-IX}) is devoted to the question of deciding whether a given automaton $\aut{A}$ of type \class{X} is isomorphic to its right congruence, or if this is not the case whether there exists an automaton $\aut{A}'$ of the same type that recognizes the same language and is isomorphic to its right congruence, namely whether the given automaton recognizes a language in the class $\class{IXA}$. Using this result we can show that a teacher can construct a characteristic sample not only given an acceptor which is isomorphic to the right congruence of the language, but also given an acceptor which is not, but is equivalent to such an acceptor.
We conclude in \Cref{sec:discussion} with a short discussion.

\section{Preliminaries}\label{sec:prelim}
\paragraph{Automata}
An \emph{automaton} is a tuple $\aut{M} = \la \Sigma, Q, q_\iota ,\delta \ra$ consisting of a finite alphabet $\Sigma$ of symbols, a finite set $Q$ of
states, an initial state $q_\iota\in Q$, and a transition function ${\delta: Q \times \Sigma \rightarrow 2^Q}$.
We extend $\delta$ to domain $Q \times \Sigma^*$ in the usual way: $\delta(q,\varepsilon) = q$ and $\delta(q,\sigma x) = \cup_{q' \in \delta(q,\sigma)} \delta(q',x)$ for all $q \in Q$ and $\sigma \in \Sigma$.

We define the \emph{size} of an automaton to be $|\Sigma| \cdot |Q|$.
A state $q \in Q$ is \emph{reachable} iff there exists $x \in \Sigma^*$ such that $q \in \delta(q_\iota,x)$.
For $q \in Q$, $\aut{M}^q$ is the automaton $\aut{M}$ with its initial state replaced by $q$.
We say that $\A$ is \emph{deterministic} if $|\delta(q,\sigma)| \leq 1$ and \emph{complete} if $|\delta(q,\sigma)| \geq 1$, for every $q\in Q$ and $\sigma \in \Sigma$.
For deterministic automata we abbreviate $\delta(q,\sigma)=\{q'\}$ as $\delta(q,\sigma)=q'$.
Two automata $\aut{M}$ and $\aut{M}'$ with the same alphabet $\Sigma$ are isomorphic if there exists a bijection $f$ from the states $Q$ of $\aut{M}$ to the states $Q'$ of $\aut{M}'$ such that $f(q_\iota) = q_\iota'$ and for every $q \in Q$ and $\sigma \in \Sigma$, $\{f(r) \mid r \in \delta(q,\sigma)\} = \delta'(f(q),\sigma)$.

We assume a fixed total ordering on $\Sigma$, which induces the \emph{shortlex} total ordering on $\Sigma^*$, defined as follows. For $x, y \in \Sigma^*$, $x$ precedes $y$ in the shortlex ordering if $|x| < |y|$ or $|x| = |y|$ and $x$ precedes $y$ in the lexicographic ordering induced by the ordering on $\Sigma$.

A \emph{run} of an automaton on a finite word ${v=a_1 a_2\ldots a_n}$ is a sequence of states ${q_0,q_1,\ldots,q_n}$ such that $q_0=q_\iota$, and for each $i\geq 1$, ${q_{i}\in\delta(q_{i-1},a_{i})}$.  A \emph{run} on an infinite word is defined similarly and consists of an infinite sequence of states.  For an infinite run $\rho = q_0, q_1, \ldots$, we define the set of states visited infinitely often, denoted $\infss{\aut{M}}(\rho)$, as the set of $q \in Q$ such that $q = q_i$ for infinitely many indices $i\in\naturals$.
This is abbreviated to $\inf(\rho)$ if $\aut{M}$ is understood.

\paragraph{The product of two automata.}
Let ${\M}_1$ and ${\M}_2$ be two deterministic complete automata with the same alphabet $\Sigma$,
where for $i = 1,2$,
${\M}_i = \la \Sigma, Q_i,(q_\iota)_i ,\delta_i \ra$.
Their product automaton, denoted ${\M}_1 \times {\M}_2$,
is the deterministic complete automaton
$\M = \la \Sigma, Q, q_\iota ,\delta \ra$
such that
$Q = Q_1 \times Q_2$ is
the set of ordered pairs of states of ${\M}_1$ and ${\M}_2$; the initial state
$q_\iota = ((q_\iota)_1,(q_\iota)_2)$ is
the pair of initial states of the two automata;
and for all $(q_1,q_2) \in Q$ and $\sigma \in \Sigma$,
$\delta((q_1,q_2), \sigma) = (\delta_1(q_1,\sigma),\delta_2(q_2, \sigma))$.
For $i = 1,2$, let $\pi_i$ be projection onto the $i$-th coordinate,
so that for a subset $S$ of $Q$, 
$\pi_1(S) = \{q_1 \in Q_1 \mid \exists q_2 \in Q_2.\ (q_1,q_2) \in S\}$,
and analogously for $\pi_2$.

\paragraph{Acceptors}

By augmenting an automaton $\aut{M} = \la \Sigma, Q, q_{\iota}, \delta \ra$ with an acceptance condition $\alpha$, obtaining a tuple $\aut{A} = \la \Sigma, Q, q_\iota,$ $\delta, \alpha \ra$, we get an \emph{acceptor}, a machine that accepts some words and rejects others. We may also denote $\aut{A}$ by $(\aut{M},\alpha)$.
An acceptor accepts a word if at least one of the runs on that word is accepting.  If the automaton is not complete, a given word $w$ may not have any run in the automaton, in which case $w$ is rejected.

For finite words the acceptance condition is a set $F \subseteq Q$ and a run on a word $v$ is accepting if it ends in an accepting state, i.e., if $\delta(q_\iota,v)$ contains an element of $F$. For infinite words, there are various acceptance conditions in the literature, and we consider six of them: \buchi, \cobuchi, parity, Rabin, Streett and Muller, all based on the set of states visited infinitely often in a given run.
For each model we define the related quantity of the size of the acceptor, taking into account the acceptance condition.

A \emph{\buchi} or \emph{\cobuchi} acceptance condition is a set of states $F \subseteq Q$. A run $\rho$ of a \buchi\ acceptor is accepting if it visits $F$ infinitely often, that is, $\inf(\rho) \cap F \neq \emptyset$. A run $\rho$ of a co\buchi\ acceptor is accepting if it visits $F$ only finitely many times, that is, $\inf(\rho) \cap F = \emptyset$.
The \emph{size} of a \buchi\ or \cobuchi\ acceptor is the size of its automaton.

A \emph{parity} acceptance condition is a map $\kappa:Q \rightarrow \naturals$ assigning to each state a natural number termed a color (or priority). A run of a parity acceptor is accepting if the \textbf{minimum} color visited infinitely often is \textbf{odd}. The \emph{size} of a parity acceptor is the size of its automaton.

A \emph{Rabin} or \emph{Streett} acceptance condition consists of a finite set of pairs of sets of states $\alpha = \{(G_1,B_1), \ldots, (G_k,B_k)\}$ for some $k\in\naturals$ and $G_i \subseteq Q$ and $B_i \subseteq Q$ for $i\in[1..k]$.
A run of a Rabin acceptor is accepting if there exists an $i\in[1..k]$ such that $G_i$ is visited infinitely often and $B_i$ is visited finitely often.
A run of a Streett acceptor is accepting if for all $i\in[1..k]$, $G_i$ is visited finitely often or $B_i$ is visited infinitely often.
The \emph{size} of a Rabin or Streett acceptor is the sum of the size of its automaton and $k-1$.

A \emph{Muller} acceptance condition is a set of sets of states $\alpha=\{F_1,F_2,\ldots,F_k\}$ for some $k\in\naturals$ and $F_i\subseteq Q$ for $i\in[1..k]$. A run of a Muller acceptor is accepting if the set $S$ of states visited infinitely often in the run is a member of $\alpha$.
The \emph{size} of a Muller acceptor is the sum of the size of its automaton and $k-1$.

The set of words accepted by an acceptor $\aut{A}$ is denoted by $\sema{\aut{A}}$.
$L_1 \oplus L_2$ is the symmetric difference of sets $L_1$ and $L_2$: $(L_1 \setminus L_2) \cup (L_2 \setminus L_1)$.
Two acceptors $\aut{A}$ and $\aut{B}$ are \emph{equivalent} if they accept the same language, that is, $\sema{\aut{A}}=\sema{\aut{B}}$.
For a state $q$, the acceptor $\aut{A}^q$ is the acceptor $\aut{A}$ with its automaton initial state replaced by $q$.  We say that the $\omega$-word $w$ is \emph{accepted from state $q$} iff $w \in \sema{\aut{A}^q}$.

We use three-letter acronyms for automata and classes concerning 
the common $\omega$-automata discussed above. The first letter is in $\{N,D,I\}$ and stands for \emph{non-deterministic}, \emph{deterministic} and \emph{isomorphic} (or \emph{fully informative}) which will be explained in the sequel. The second letter describes the acceptance condition, and the third letter $A$ stands for \emph{acceptor}. That is, we use NBA, NCA, NPA, NRA, NSA, NMA (resp., DBA, DCA, DPA, DRA, DSA, DMA) 
for non-deterministic (resp., deterministic) \buchi, co\buchi, parity, Rabin, Street and Muller acceptors. 
We use blackboard font for the respective classes of representations. That is, we use \class{NBA}, \class{NCA}, \class{NPA}, \class{NMA}, \class{NRA} and \class{NSA} (resp., \class{DBA}, \class{DCA}, \class{DPA}, \class{DRA}, \class{DSA} and \class{DMA}) for the corresponding class of representations. 
It is known that  $\class{NCA}$ 
and $\class{DCA}$ recognize the same languages and that the classes  $\class{DCA}$ and $\class{DBA}$ are distinct proper subclasses of the regular $\omega$-languages.  The other classes are the full class of regular $\omega$-languages.

\paragraph{Some relationships between the models}	
\label{ssec:model-relationships}

We observe the following known relationships~\cite[Chapter 1]{GradelTWbook}.
\begin{myclaim}
	\label{claim:make-automaton-complete}
	Let $\aut{A}$ be an acceptor of one of the types NBA, NCA, NPA, NRA, NSA, or NMA with $n$ states over the alphabet $\Sigma$.  There is an equivalent complete acceptor $\aut{A}'$ of the same type whose size is at most $|\Sigma|$ larger.  $\aut{A}'$ may be taken to be deterministic if $\aut{A}$ is deterministic.
\end{myclaim}

\begin{myclaim}\label{clm:basic-relations-between-omega-aut} 
\begin{enumerate}
\item \label{claim:NBA-to-NPA}
Let $\aut{B} = \la \Sigma,Q, q_\iota, \delta, F \ra$, where $\aut{B}$ is an NBA.  Define the NPA $\aut{P} = \la \Sigma,Q, q_\iota, \linebreak[5]\delta, \kappa \ra$ where $\kappa(q) = 1$ if $q \in F$ and $\kappa(q) = 2$ otherwise. Then $\aut{B}$ and $\aut{P}$ are equivalent and have the same size.
	$\aut{P}$ is deterministic if $\aut{B}$ is.
\item  
\label{claim:NBA-to-NRA}
	Let $\aut{B} = \la \Sigma,Q, q_\iota, \delta, F \ra$, where $\aut{B}$ is an NBA.  Define the NRA $\aut{R} = \la \Sigma,Q, q_\iota, \delta, \{(F,\emptyset)\} \ra$. Then $\aut{B}$ and $\aut{R}$ are equivalent and have the same size.
	$\aut{R}$ is deterministic if $\aut{B}$ is.
\item 
	\label{claim:NCA-to-NSA}
	Let $\aut{C} = \la \Sigma,Q, q_\iota, \delta, F \ra$, where $\aut{C}$ is an NCA.  Define the NSA $\aut{S} = \la \Sigma,Q, q_\iota, \delta, \{(F,\emptyset)\} \ra$. Then $\aut{S}$ and $\aut{C}$ are equivalent and have the same size.
	$\aut{S}$ is deterministic if $\aut{C}$ is.
\item 
 	\label{claim:DBA-DCA-complement}
	Let $\aut{B} = \aut{C} = \la \Sigma, Q, q_\iota, \delta, F \ra$, where
	$\aut{B}$ is a complete DBA and $\aut{C}$ is a complete DCA.
	Then $\aut{B}$ and $\aut{C}$ are the same size and the languages they recognize are complements 
	of each other, that is,
	$\sema{\aut{B}} = \Sigma^{\omega} \setminus \sema{\aut{C}}$.
\item  
\label{claim:DRA-DSA-complement}
	Let $\aut{R} = \aut{S} = \la \Sigma, Q, q_\iota, \delta, \{(G_1,B_1),\ldots,(G_k,B_k)\} \ra$, where
	$\aut{R}$ is a complete DRA and $\aut{S}$ is a complete DSA.
	Then $\aut{R}$ and $\aut{S}$ have the same size and the languages they recognize are complements of each other, that is, $\sema{\aut{R}} = \Sigma^{\omega} \setminus \sema{\aut{S}}$.
\end{enumerate}
\end{myclaim}

\paragraph{Right congruence}

An equivalence relation $\sim$ on $\Sigma^*$ is a \emph{right congruence} if $x \sim y$ implies $x\sigma \sim y\sigma$ for all $x,y \in \Sigma^*$ and $\sigma \in \Sigma$. The \emph{index} of $\sim$, denoted ${|\!\sim\!|}$ is the number of equivalence classes of $\sim$.  For a word $x \in \Sigma^*$ the notation $[x]_{\sim}$ denotes the equivalence class of $\sim$ that contains $x$.

With a right congruence $\sim$ of finite index one can naturally associate a complete deterministic automaton $\aut{M}_\sim=\la \Sigma, Q, \initstate, \delta \ra$ as follows:  the set of states $Q$ consists of the equivalence classes of $\sim$. The initial state $\initstate$ is the equivalence class $[\varepsilon]_\sim$. The transition function $\delta$ is defined by $\delta([u]_\sim,\sigma)=[u\sigma]_\sim$ for all $\sigma \in \Sigma$.
Also, given a complete deterministic automaton $\aut{M} = \la \Sigma, Q, \initstate, \delta \ra$, we can naturally associate with it a right congruence as follows: $x \sim_\aut{M} y$ iff $\aut{M}$ reaches the same state of $\aut{M}$ when reading $x$ or $y$, that is, $\delta(q_\iota,x) = \delta(q_\iota,y)$.

Given a language $L\subseteq \Sigma^*$ its \emph{canonical right congruence} $\sim_L$ is defined as follows: 
$x \sim_L y$ iff ${\forall z \in \Sigma^*}.\ {xz\in L} \iff {yz \in L}$. 
The Myhill-Nerode theorem states that a language $L \subseteq \Sigma^*$ is regular iff $\sim_L$ is of finite index. Moreover, if $L$ is accepted by a complete DFA $\aut{A}$, then $\sim_\aut{M}$ refines $\sim_L$, where $\aut{M}$ is the automaton of $\aut{A}$. Finally, any complete DFA of minimum size that accepts $L$ has an automaton that is isomorphic to $\aut{M}_{\sim_L}$.

For an $\omega$-language $L\subseteq \Sigma^\omega$, its \emph{canonical right congruence} $\sim_L$ is defined similarly, by quantifying over $\omega$-words. That is, $x \sim_L y$ iff ${\forall z \in \Sigma^\omega}.\ {xz\in L} \iff {yz \in L}$.
If $L$ is a regular $\omega$-language then $\sim_L$ is of finite index, and for any complete DBA (resp., DCA, DPA, DRA, DSA, DMA) $\aut{A}$ that accepts $L$, $\sim_{\aut{M}}$ refines $\sim_L$, where $\aut{M}$ is the automaton of the acceptor.

However, for regular $\omega$-languages, the relation $\sim_L$  does not suffice to obtain a ``Myhill-Nerode'' characterization.  In particular, for a regular $\omega$-language $L$ there may be no way to define an acceptance condition for $\aut{M}_{\sim_L}$ that yields a DBA (resp., DCA, DPA, DRA, DSA, DMA) that accepts $L$.  As an example consider the language $L = (a+b)^*(bba)^\omega$. Then $\sim_{L}$ consists of just one equivalence class, 
because for any $x \in \Sigma^*$ and $w \in \Sigma^\omega$ we have that $xw \in L$ iff $w$ has $(bba)^\omega$ as a suffix.
But a DBA (resp., DCA, DPA, DRA, DSA, DMA) that accepts $L$ clearly needs more than a single state.

\paragraph{The fully informative classes}

In light of the lack of a Myhill-Nerode result for regular $\omega$-languages, we define a restricted type of deterministic \buchi\ 
(resp., \cobuchi, parity, Rabin, Streett, Muller) acceptors.  
  Following the three-letter acronym notation introduced in the \emph{Acceptors} subsection, for $X\in\{B,C,P,R,S,M\}$ we say that 
   a DXA 
  $\aut{A}$
  recognizing a language $L$ is \emph{fully informative}
if it is complete and its automaton is isomorphic to $\aut{M}_{\sim_L}$ (the automaton that corresponds to the canonical right congruence $\sim_L$ for the language, as defined above).   
A DXA is an IXA
if it is fully informative. 
A DXA $\aut{B}$ is in $\IX$ if there exists an IXA  $\aut{A}$ such that $\sema{\aut{B}}=\sema{\aut{A}}$.
We note that every state of an 
IXA is reachable because every state of $\aut{M}_{\sim_L}$ is reachable. Since every state of a minimal automaton for a language $L$ in $\IX$ for $\class{X}\in\{\mathbb{B},\mathbb{C},\mathbb{P},\mathbb{R},\mathbb{S},\mathbb{M}\}$ corresponds to an equivalence class of $\sim_L$ we refer to the $\IX$ classes as the \emph{fully informative} classes.

Despite the fact that each of these classes is a proper subset of its corresponding deterministic class (e.g., \class{IBA} is a proper subset of \class{DBA}), these classes are more expressive than one might first conjecture.
It was shown in~\cite{AngluinF18} that in every class of the infinite Wagner hierarchy~\cite{Wagner75} there are languages in $\IM$ and $\IP$. 
Moreover, in a small experiment reported in~\cite{AngluinF18}, among randomly generated  Muller automata, the vast majority turned out to be in $\IM$.

\section{Notions of learnability and teachability}\label{sec:notions-learnability-teachability}
In this section we define and compare general notions of learnability and teachability, and some computable and polynomial time variants of them.  A summary of the definitions is provided in~\Cref{tbl:defs-summary}.

\begin{table}\caption{Summary of the definitions for teachability and learnability}\label{tbl:defs-summary}

{\small{
\begin{tabular}{|ccc||l|}
\hline
Samples & Learner & Teacher & Definition\\ 
\hline \hline
arb & arb & arb & characteristic samples\\
poly & arb & arb & concise characteristic samples \\
arb & poly & arb & efficiently learnable\\
arb & arb & poly & efficiently teachable (samples will be poly)\\
poly & poly & arb & identifiable in the limit using polynomial time and data \\
arb & poly & poly & efficiently teachable/learnable (samples will be poly)\\
\hline
\end{tabular}}}
\end{table}

\subsection{Teachers, learners, and characteristic samples}

We are concerned with examples and concepts that can be represented by finite binary strings as follows.
$\mathcal{X} = \{0,1\}^*$ is the domain of representations of examples.  (In the next section we describe how we use finite strings as representations of $\omega$-words.)
A \emph{concept} is any subset of $\mathcal{X}$.
A \emph{class of concepts} $\class{C}$ consists of the set $\{0,1\}^*$ of representations of concepts, together with a mapping $\sema{\cdot}$ from $\{0,1\}^*$ to concepts such that $\sema{\aut{C}}$ is the subset of $\mathcal{X}$ that $\aut{C}$ represents.
Thus every finite binary string represents an example and also a concept.
The \emph{length} of the representation of an example $x$ or concept $\aut{C}$ is its length as a string, that is, $|x|$ or $|\aut{C}|$.
The \emph{size} of a concept $\aut{C}$, denoted $\size(\aut{C})$,  is the minimum $|\aut{C}'|$ of any representation $\aut{C}'$ such that $\sema{\aut{C}'} = \sema{\aut{C}}$.\footnote{
The notions of size defined for acceptors in~\Cref{sec:prelim} polynomially relate to the notion of size defined here. }

A \emph{sample} $S$ is a finite set of elements $(x,b)$ where $x \in \mathcal{X}$ and $b \in \{0,1\}$.  The \emph{length} of $S$ is the sum of the lengths of the examples $x$ that appear in it.  A sample $S$ is \emph{consistent} with a concept $\aut{C}$ iff for every $(x,b) \in S$ we have $b = 1$ iff $x \in \sema{\aut{C}}$.

A \emph{learner} for $\class{C}$ is a function $\alg{L}$ that maps a sample $T$ to the representation of a concept $\alg{L}(T)$ in $\class{C}$ with the property that if $T$ is consistent with at least one element of $\class{C}$, then $\alg{L}(T)$ and $T$ are consistent.
A \emph{teacher} for $\class{C}$ is a function $\alg{T}$ that maps the representation of a concept $\aut{C}$ in $\class{C}$ to a sample $\alg{T}(\aut{C})$ such that $\alg{T}(\aut{C})$ and $\aut{C}$ are consistent.
Note that learners and teachers need not be computable.

A sample $T$ is a \emph{characteristic sample} for $\aut{C}$ and a learner $\alg{L}$
if $T$ is consistent with $\aut{C}$ and for every
sample $T'\supseteq T$ consistent with $\aut{C}$ we have
$\sema{\alg{L}(T')} = \sema{\aut{C}}$.
The intuition is that additional information consistent with $\aut{C}$ beyond $T$ will not cause the learner to change its mind about the correct concept.

A class $\class{C}$ \emph{has characteristic samples for a learner} $\alg{L}$ if there exists a teacher $\alg{T}$ such that for every $\aut{C}$ in $\class{C}$, $\alg{T}(C)$ is a characteristic sample for $\aut{C}$ and $\alg{L}$.
A class $\class{C}$ \emph{has characteristic samples} if it has characteristic samples for some learner.
A well known property of characteristic samples is the following.

\begin{lemma}[Key Property of Characteristic Samples]
    \label{lemma:key-property-char-samples}
    Assume $\class{C}$ is a class of concepts, $\alg{T}$ is a teacher, and $\alg{L}$ is a learner such that $\alg{T}$ gives a characteristic sample for $\aut{C}$ and $\alg{L}$ for every $\aut{C}$ in $\class{C}$.  If $\aut{C}_i$ and $\aut{C}_j$ are any concepts from $\class{C}$ such that $\sema{\aut{C}_i} \neq \sema{\aut{C}_j}$ then there exists some $(x,b) \in \alg{T}(\aut{C}_i) \cup \alg{T}(\aut{C}_j)$ such that $x \in \sema{\aut{C}_i} \oplus \sema{\aut{C}_j}$.
\end{lemma}

\begin{proof}
    Assume to the contrary.  Let $T = \alg{T}(\aut{C}_i) \cup \alg{T}(\aut{C}_j)$ and consider $\alg{L}$ with input $T$.  Because the characteristic sample for $\aut{C}_i$ is contained in $T$ and $T$ is consistent with $\aut{C}_i$, $\alg{L}$ must output a concept denoting $\sema{\aut{C}_i}$.  The same is true of $\aut{C}_j$, but because $\sema{\aut{C}_i} \neq \sema{\aut{C}_j}$ it is impossible for $\alg{L}$ to output a concept whose denotation is equal to both of them.
\end{proof}

\subsection{Computable teachers and learners}

To discuss computability of teaching and learning, we consider the following three possible properties of a class of concepts.

\begin{enumerate}[align=left]
    \item [(C1)] There is an algorithm to decide membership of $x$ in $\sema{\aut{C}}$, given example representation $x$ and concept representation $\aut{C}$.
    \item [(C2)] There is an algorithm to decide whether there exists a concept representation $\aut{C}$ consistent with $T$, given a sample $T$.
    \item [(C3)] There is an algorithm to decide whether $\sema{\aut{C}_i} = \sema{\aut{C}_j}$, given two concept representations $\aut{C}_i$ and $\aut{C}_j$.
\end{enumerate}

The assumptions (C1) and (C2) are sufficient to guarantee the existence of characteristic samples for $\class{C}$ with a computable learner and a possibly non-computable teacher. The learner uses the algorithm of identification by enumeration~\cite{Gold67}.

\begin{theorem}[Identification by Enumeration]
    \label{theorem:id_by_enum}
    Under the assumptions (C1) and (C2), there is a computable learner $\alg{L}$ such that class $\class{C}$ has characteristic samples for $\alg{L}$.
\end{theorem}
\begin{proof}
    The learner $\alg{L}$ enumerates the finite binary strings representing concepts in shortlex order as $\aut{C}_1, \aut{C}_2, \ldots$.  Given a sample $T$, the learner first uses (C2) to check whether there is any concept in $\class{C}$ consistent with $T$.  If not, it outputs an arbitrary concept.  Otherwise, the output of $\alg{L}$ is $\aut{C}_n$ for the least $n$ such that $\aut{C}_n$ is consistent with $T$.
    This is possible because by (C1), whether $\aut{C}$ is consistent with $T$ is decidable.

    The teacher $\alg{T}$ with input $\aut{C}$ finds the least $n$ such that $\sema{\aut{C}_n} = \sema{\aut{C}}$, and for each $m < n$ determines an example $x_m$ that distinguishes $\sema{\aut{C}_m}$ from $\sema{\aut{C}_n}$.  The sample output by $\alg{T}$ consists of all $x_m$ with $m < n$, labeled to be consistent with $\aut{C}$.
\end{proof}

We cannot necessarily take the teacher $\alg{T}$ in this construction to be computable; in particular, it must test the equivalence of two concepts.
In fact, Fisman et al.~\cite{FismanFZ23} show 
that there exists a class $\class{C}$ that satisfies assumptions (C1) and (C2) and has characteristic samples such that 
there is no computable function $\alg{T}$ to construct a characteristic sample for every language in the class.\footnote{The characteristic samples for the class are of cardinality $2$, but have no computable bound on their length.}

However, under the further assumption (C3) that equivalence of concepts is computable, the teacher $\alg{T}$ in the proof may be made computable, since if two concepts $\aut{C}_i$ and $\aut{C}_j$ are determined to be inequivalent, it is safe to do an unbounded search for an example that distinguishes them.
Thus we have the following.
\begin{corollary}
    \label{cor:id_by_enum_computable_teacher}
    Under the assumptions (C1), (C2) and (C3), there is are a computable teacher $\alg{T}$ and a computable learner $\alg{L}$ such that $\alg{T}$ computes characteristic samples for $\class{C}$ and $\alg{L}$.
\end{corollary}

For a concept $\aut{C}$, the cardinality of the characteristic sample constructed for $\aut{C}$ in the proof of \Cref{theorem:id_by_enum} may be as large as the number of distinct concepts preceding it in the enumeration $\aut{C}_1, \aut{C}_2, \ldots$, which can be exponential in $|\aut{C}|$.
Barzdin and Freivalds~\cite{BarzdinFreivalds1972} proved that the majority vote algorithm could be used in a prediction setting to bound the number of mistakes of prediction linearly in the size of the target concept, which implies a corresponding bound for characteristic samples. (See also the surveys~\cite{Freivaldsetal1991} and~\cite{ZEUGMANN20084}.)
For completeness, we describe their construction in the context of characteristic samples.

First we prove a lemma on using the majority vote algorithm with a finite set of concepts to construct a sample.
If $\class{D}$ is a finite nonempty set of concepts, we define the majority vote concept of $\class{D}$, denoted $\maj(\class{D})$, to consist of all $x \in \mathcal{X}$ such that the cardinality of the set of concepts in $\class{D}$ containing $x$ is at least as large as the cardinality of the set of concepts in $\class{D}$ not containing $x$.
The concept $\maj(\class{D})$ may or may not belong to $\class{D}$ or $\class{C}$.
Given a finite nonempty set of concepts $\class{D}$ and another concept $\aut{C}$, we define a sample, denoted $\sample(\class{D},\aut{C})$ by executing the halving algorithm with $\class{D}$ as the concept set and $\aut{C}$ as the target concept.
In detail, initialize $\class{D}_0 = \class{D}$, $i = 0$ and $S_0 = \emptyset$.
While $\class{D}_i$ is nonempty, compare $\maj(\class{D}_i)$ and $\aut{C}$.
If these are the same concept, then $S_i$ is output as the value of $\sample(\class{D},\aut{C})$.
Otherwise, let $x$ be the least element of $\mathcal{X}$ on which they differ, and add the pair $(x,b)$ to $S_i$ to get $S_{i+1}$, where $b = 1$ if $x$ is in $\sema{\aut{C}}$ and $b = 0$ otherwise.  $\class{D}_{i+1}$ is set to those concepts in $\class{D}_i$ that are consistent with $(x,b)$, $i$ is set to $i+1$, and the while loop continues.
If $\class{D}_i$ is empty, then $S_i$ is the value of $\sample(\class{D},\aut{C})$.

\begin{lemma}
    \label{lemma:maj_vote_to_sample}
    If $\class{D}$ is a finite nonempty set of concepts and $\aut{C}$ any concept, $\sample(\class{D},\aut{C})$ has at most $1+\log_2(|\class{D}|)$ elements.
    For any concept $\aut{C}'$ in $\class{D}$, we have that $\aut{C}$ and $\aut{C}'$ are both consistent with $\sample(\class{D},\aut{C}) \cup \sample(\class{D},\aut{C}')$ iff $\sema{\aut{C}} = \sema{\aut{C}'}$.
\end{lemma}

\begin{proof}
Because each $\class{D}_{i+1}$ contains at most half as many concepts as $\class{D}_i$, termination must occur with $i \le \log_2(|\class{D}|)$, and therefore $|\sample(\class{D},\aut{C})| \le 1+\log_2(|\class{D}|)$.

If $\sema{\aut{C}} = \sema{\aut{C}'}$, the same sample will be returned for each concept, and they will both be consistent with that sample.  Conversely, if both are consistent with the union of their final samples, the sets $\class{D}_i$ and $S_i$ will be the same for each $i$.  Because $\aut{C}'$ is in $\class{D}$, termination cannot occur because $\class{D}_i$ becomes empty.  Thus, termination must occur because both are equal to $\maj(\class{D}_i)$ for some $i$, which means that they are the same concept.
\end{proof}

Note in particular that if $\aut{C}$ is not equivalent to any concept in $\class{D}$, then for all $\aut{C}'$ in $\class{D}$ the fact that $\sema{\aut{C}} \neq \sema{\aut{C}'}$ will be witnessed by some example $(x,b)$ in $\sample(\class{D},\aut{C}) \cup \sample(\class{D},\aut{C}')$.

\begin{theorem}
    \label{theorem:linear-cardinality-char-samples}
    Let $\class{C}$ be a concept class.  There exist a teacher $\alg{T}$ and a learner $\alg{L}$ such that for every concept representation $\aut{C}$ in $\class{C}$, $\alg{T}(\aut{C})$ is a characteristic sample for $\aut{C}$ and $\alg{L}$, and the cardinality of $\alg{T}(\aut{C})$ is bounded by $O(|\aut{C}|)$.
\end{theorem}

\begin{proof}
The approach is to use the method of \Cref{lemma:maj_vote_to_sample} on successive finite blocks of concepts of double exponential cardinality.
Let $\class{B}_0$ contain concepts $\aut{C}_1$ and $\aut{C}_2$.  For $i > 0$, let $\class{B}_i$ contain those concepts $\aut{C}_k$ such that $2^{2^{i-1}} < k \le 2^{2^i}$.

Then the teacher $\alg{T}$ on input $\aut{C}$ finds the least $i$ such that for some $k$ with $\aut{C}_k$ in $\class{B}_i$, $\sema{\aut{C}} = \sema{\aut{C}_k}$.
The sample output for $\aut{C}$ is the union of $\sample(\class{B}_j,\aut{C})$ for $j = 0,1,\ldots, i$.

The cardinality of the sample output is at most the sum of the cardinalities of the samples $\sample(\class{B}_j,\aut{C})$.  We have $|\class{B}_0| = 2$ and for $i > 0$
\[|\class{B}_i| = 2^{2^i}-2^{2^{i-1}} < 2^{2^i}.\]
For all $j$, the sample $\sample(\class{B}_j,\aut{C})$ thus has at most $1+ \log_2(\class{B}_j) = 1+2^j$ elements.
The sum of these bounds through $i$ is
\begin{equation}\label{eq}
  (1+2^0) + (1 + 2^1) + \ldots + (1 + 2^i) = i + 2^{i+1}  
\end{equation}
The least $i$ such that $\class{B}_i$ contains $\aut{C}_k$ is at most $\log_2\log_2(k)$ for $k > 1$. So for $k > 1$ by instantiating $i$ with $\log_2\log_2(k)$ in~\Cref{eq} we get that the cardinality of $\alg{T}(\aut{C})$ is at most $2\log_2(k) + \log\log(k)$.
In the shortlex ordering of finite binary strings, the string $\aut{C}_k$ has length at least $\log_2 k$, so the cardinality of $\alg{T}(\aut{C})$ is bounded by $O(|\aut{C}|)$.

The learner $\alg{L}$ on input $T$ finds the least $\aut{C}$, if any, such that $\aut{C}$ is consistent with $T$ and the sample $\alg{T}(\aut{C})$ is a subset of $T$.
If such a $\aut{C}$ is found, it is output.
Otherwise, if there is some concept $\aut{C}$ in $\class{C}$ consistent with $T$, $\alg{L}$ outputs the least such, and otherwise outputs an arbitrary element of $\class{C}$.

To see that $\alg{T}(\aut{C})$ is a characteristic sample for $\aut{C}$ and $\alg{L}$,
suppose $T$ contains $\alg{T}(\aut{C})$ and is consistent with $\aut{C}$.  Then $\alg{L}$ outputs $\aut{C}$ unless there is some earlier concept $\aut{C}'$ that occurs in an example string of $T$ and is consistent with $T$ and is such that $\alg{T}(\aut{C}')$ is contained in $T$.
Consider the computation of the sample for $\aut{C}$, and the block $\class{B}_j$ that contains the least concept equivalent to $\aut{C}'$.
Because $\aut{C}$ and $\aut{C}'$ are consistent with $\alg{T}(\aut{C})$, they are both consistent with $\sample(\class{B}_j,\aut{C})$ and because they are consistent with $\alg{T}(\aut{C}')$, they are consistent with $\sample(\class{B}_j,\aut{C}')$, so by \Cref{lemma:maj_vote_to_sample} applied to the block $\class{B}_j$, they denote the same concept.
\end{proof}

We turn to the question of the computability of the teacher and the learner.
\begin{corollary}
    \label{cor:linear-cardinality-char-samples-computable}
    Let $\class{C}$ be a concept class satisfying (C1), (C2) and (C3). There exist a computable teacher $\alg{T}$ and a computable learner $\alg{L}$ such that for every concept representation $\aut{C}$ in $\class{C}$, $\alg{T}(\aut{C})$ is a characteristic sample for $\aut{C}$ and $\alg{L}$, and the cardinality of $\alg{T}(\aut{C})$ is bounded by $O(|\aut{C}|)$.
\end{corollary}

\begin{proof}
    To allow the construction of $\sample(\class{B}_i,\aut{C})$ to be computable, using (C3), the teacher may test the equivalence of pairs of concepts in $\class{B}_i$, and pairs consisting of a concept from $\class{B}_i$ and $\aut{C}$, and for each inequivalent pair, may search for the least example $x$ on which they disagree.  Then the same computation as in \Cref{lemma:maj_vote_to_sample} is done, using just these finitely many disagreement strings as the domain.

    On input $\aut{C}$, to allow the learner to limit its search to a finite set of candidates, the teacher also adds to the sample some example $(x,b)$, where  $|x| \ge |\aut{C}|$ and $b = 1$ iff $x$ is in $\sema{\aut{C}}$.\footnote{See the discussion following \Cref{theorem:poly-teachable-implies-poly-teachable-learnable}.} On input $T$, the learner searches among $\aut{C}$ such that $|\aut{C}|$ is bounded by the length of the longest string in $T$ to find the least $\aut{C}$ consistent with $T$ such that $\alg{T}(\aut{C})$ is a subset of $T$.  If none is found, then (C2) allows the learner to output some concept consistent with $T$, if there is one.
\end{proof}

\subsection{Polynomial time learners and teachers}

Moving from this general setting, we consider polynomial bounds on the length of characteristic samples and the running times of the teacher and learner.
We consider the following possible properties of $\class{C}$.
\begin{enumerate}[align=left]
    \item[(P1)] \emph{Polynomial time example membership.} There is a polynomial time algorithm to decide whether $x \in \sema{\aut{C}}$ given $x$ and $\aut{C}$.
    \item[(P2)] \emph{Polynomial time default hypothesis construction.} There is a polynomial time algorithm that returns $\aut{C}_T$ in $\class{C}$ consistent with a given sample $T$, or determines that no concept in $\class{C}$ is consistent with $T$.
    \item[(P3)] \emph{Polynomial time equivalence with least counterexamples.} There is a polynomial time algorithm that determines whether $\sema{\aut{C}_i} = \sema{\aut{C}_j}$ given concept representations $\aut{C}_i$ and $\aut{C}_j$.  In the case the concepts are not equal, the algorithm also returns the shortlex least string $x \in \sema{\aut{C}_i} \oplus \sema{\aut{C}_j}$.
\end{enumerate}

A class $\class{C}$ is \emph{concisely distinguishable} if there exists a polynomial $p(n)$ such that for every pair $\aut{C}_i$ and $\aut{C}_j$ such that $\sema{\aut{C}_i} \neq \sema{\aut{C}_j}$, there exists a string $x 
 \in \sema{\aut{C}_i} \oplus \sema{\aut{C}_j}$ such that 
$|x| \le p(\size(\aut{C}_1)+\size(\aut{C}_2))$.

\paragraph{Concise characteristic samples}
A class $\class{C}$ \emph{has concise characteristic samples} if there exist a polynomial $p(n)$, a teacher $\alg{T}$, and a learner $\alg{L}$ such that for every $\aut{C} \in \class{C}$, $\alg{T}(\aut{C})$ is a characteristic sample for $\aut{C}$ and $\alg{L}$, and  $|\alg{T}(\aut{C})| \le p(\size(\aut{C}))$.
Recall that the length of a sample 
is the sum of length of the strings in it. In \Cref{cor:linear-cardinality-char-samples-computable} we considered just the cardinality of  $\alg{T}(\aut{C})$ and not the length of the words in the sample, but in 
$|\alg{T}(\aut{C})|$ the length of the words in the sample matters.
Note also that the polynomial bound is in terms of the size of the smallest representation of $\sema{\aut{C}}$.
If $\class{C}$ has concise characteristic samples then $\class{C}$ is concisely distinguishable, by \Cref{lemma:key-property-char-samples}.
In the converse direction, we have the following consequence of  \Cref{cor:linear-cardinality-char-samples-computable}.

\begin{corollary}
    \label{cor:concise-distinguishable-concise-char-samples}
    Assume the class $\class{C}$ satisfies (C1), (C2) and (C3) and is concisely distinguishable.  Then it has concise characteristic samples for a computable teacher $\alg{T}$ and a computable learner $\alg{L}$.
\end{corollary}

\begin{proof}
    Because $\class{C}$ is concisely distinguishable, two concepts $\aut{C}_i$ and $\aut{C}_j$ can be tested for equivalence by checking agreement on all strings of length at most a fixed polynomial in $\size(\aut{C}_i) + \size(\aut{C}_j)$.  If they are inequivalent, this process will yield a distinguishing example of at most that length.  If this method is used in the algorithm for constructing characteristic samples in the proof of \Cref{cor:linear-cardinality-char-samples-computable}, the resulting characteristic samples will be concise.
\end{proof}

\paragraph{Efficient teachability}
We say that $\class{C}$ is \emph{efficiently teachable} if
there exist a polynomial time teacher $\alg{T}$ and a learner $\alg{L}$
such that 
$\alg{T}(\aut{C})$ is a
characteristic sample for $\aut{C}$ and $\alg{L}$
for every $\aut{C}\in\class{C}$.

\begin{lemma}
    If class $\class{C}$ is efficiently teachable, then $\class{C}$ has concise characteristic samples.
\end{lemma}
\begin{proof}
    Let $\alg{T}$ be a polynomial time teacher and $\alg{L}$ a learner witnessing the fact that $\class{C}$ is efficiently teachable.  To see that $\class{C}$ has concise characteristic samples, we define the (not necessarily polynomial time) teacher $\alg{T}'$ as follows.  On input $\aut{C}$, let $\aut{C}'$ minimize $\size(\aut{C}')$ subject to $\sema{\aut{C'}} = \sema{\aut{C}}$.  Then $\alg{T}'$ outputs $\alg{T}(\aut{C}')$, which is a characteristic sample for $\aut{C}$ and $\alg{L}$ and is of size bounded by a polynomial in $\size(\aut{C})$.
\end{proof}

\paragraph{Efficient learnability}
We say that $\class{C}$ is \emph{efficiently learnable} if
there exist a teacher $\alg{T}$ and a polynomial time learner $\alg{L}$
such that 
$\alg{T}(\aut{C})$ is a
characteristic sample for $\aut{C}$ and $\alg{L}$
for every $\aut{C}\in\class{C}$.

As has been observed before~\cite{Pitt1989}, if (P1) and (P2) are satisfied, this definition is superfluous because the learning algorithm of identification by enumeration used in the proof of \Cref{theorem:id_by_enum} can be modified to run in polynomial time and the teacher adjusted appropriately to create a (potentially ridiculously large) characteristic sample.
This observation motivates the introduction of the criterion of identification in the limit with polynomial time and data (described in~\Cref{sec:id-in-the-limit-with-poly}).
The argument for this observation follows.
\begin{lemma}
\label{lemma:efficiently_learnable}
    If class $\class{C}$ satisfies (P1) and (P2) then $\class{C}$ is efficiently learnable.
\end{lemma}
\begin{proof}
    The learner $\alg{L}$ on input $T$ of length $n$ simulates the algorithm of identification by enumeration for $n^2$ steps.  It checks whether the last hypothesis $\aut{C}$ output by the simulation is consistent with $T$ and outputs $\aut{C}$ if so, using (P1). Otherwise, it returns the default hypothesis consistent with $T$, if any, using (P2). Given $\aut{C}$ in $\class{C}$, the teacher $\alg{T}$ takes the characteristic sample $T_{\aut{C}}$ from \Cref{theorem:id_by_enum} and adds enough additional examples of $\aut{C}$ that there is time for $\alg{L}$'s simulation of identification by enumeration to converge to its final answer. Note that since the domain $\mathcal{X}$ is infinite it is always possible to add more labeled words to the sample. 
\end{proof}

\paragraph{Requiring both teacher and learner to be efficient}
We say that $\class{C}$ is \emph{efficiently teachable/learnable} if
there exist a polynomial time teacher $\alg{T}$ and a polynomial time learner $\alg{L}$
such that 
$\alg{T}(\aut{C})$ is a
characteristic sample for $\aut{C}$ and $\alg{L}$
for every $\aut{C}\in\class{C}$.
This corresponds to the definition of \emph{polynomially T/L teachable} of Goldman and Mathias~\cite{GoldmanM96}.

It turns out that if (P1) and (P2) are satisfied, efficiently teachable/learnable is no stronger than requiring a polynomial time teacher.

\begin{theorem}
    \label{theorem:poly-teachable-implies-poly-teachable-learnable}
    Assume that class $\class{C}$ satisfies (P1) and (P2) and is efficiently teachable.  Then $\class{C}$ is efficiently teachable/learnable.
\end{theorem}

\begin{proof}
    Let $\alg{T}$ be polynomial time teacher witnessing the efficient teachability of $\class{C}$.  We define another polynomial time teacher $\alg{T}'$ that on input $\aut{C}$ outputs $\alg{T}(\aut{C})$ together with one pair $(x,b)$ such that $x$ is the string $\aut{C}$ and $(x,b)$ is consistent with $\aut{C}$.
    Note that the teacher directly provides the $\emph{text}$ of the target concept.

    We define a learner $\alg{L}$ as follows.  On input $T$, $\alg{L}$ finds the least $x$ (if any) such that $(x,b)$ is in $T$, the concept representation $\aut{C} = x$ is such that $\aut{C}$ is consistent with $T$ (using (P1)), and the sample $\alg{T}(\aut{C})$ is a subset of $T$ (using the polynomial time algorithm $\alg{T}$).  If such a $\aut{C}$ is found, $\alg{L}$ outputs it.  Otherwise (using (P2)) it either outputs a default concept $\aut{C}$ consistent with $T$ or an arbitrary concept.  Then $\alg{L}$ runs in time polynomial in $|T|$.

    The argument that for every $\aut{C}$ in $\class{C}$, $\alg{T}'(\aut{C})$ is a characteristic sample for $\aut{C}$ and $\alg{L}$ is the same as that in the proof of \Cref{cor:linear-cardinality-char-samples-computable}.
\end{proof}

In this proof, the teacher provides the text of the target concept itself in the sample and the learner depends on this fact and considers as potential candidates only those concepts whose strings are in the sample.
This may seem like a kind of unacceptable ``collusion'' between the teacher and the learner.
In the proof of \Cref{cor:linear-cardinality-char-samples-computable}, the teacher provides the learner with less specific syntactic information about the target concept in the form of a bound on its length.
Gold's founding positive result in the area of characteristic samples may be stated as follows.
\begin{thmC}[{\cite[Theorem~4]{Gold78}}]
DFAs are efficiently teachable.
\end{thmC}
In Gold's original proof, we see a high degree of coordination between the teacher and the learner to enable the teacher to provide all the examples that the learner will consult in constructing its hypothesis, even when additional correct examples may be present.

The issue of what kinds of coordination and communication should be permitted in a model of teaching and learning is complex, and has been considered by a number of researchers.
One paradigm requires the teacher to be able to teach any learner (or more precisely, any \emph{consistent learner}, i.e., any learner that never hypothesizes a concept inconsistent with the information it has). Goldman and Mathias, in their seminal paper~\cite{GoldmanM96}, show that this requirement is too strong in the sense that it makes even rather simple classes of concepts very hard to teach. 
Moreover, some applications call for a learning paradigm in which a teacher is required to teach a particular learner rather than an arbitrary one. 

The paradigm introduced by Goldman and Mathias~\cite{GoldmanM96}, which is the one we follow here, addresses the problem of collusion by allowing an adversary to add an arbitrary set of correctly labeled examples to the sample generated by the teacher before it is given to the learner.
For a class $\class{C}$ and a deterministic teacher $\alg{T}$ and a deterministic learner $\alg{L}$, they define $\alg{T}$ and $\alg{L}$ to be a \emph{colluding pair} if there exist $\aut{C}_i$ and $\aut{C}_j$ such that $\sema{\aut{C}_i}=\sema{\aut{C}_j}$ and for all samples $T_i \supseteq \alg{T}(\aut{C}_i)$ and $T_j \supseteq \alg{T}(\aut{C}_j)$ consistent with $\aut{C}_i$ and $\aut{C}_j$, we have $\alg{L}(T_i) \neq \alg{L}(T_j)$.
That is, in a colluding pair, the teacher is able to communicate some distinguishing information about which of the two concepts, $\aut{C}_i$ or $\aut{C}_j$, was its input.
They prove that under their paradigm there is no colluding pair.
Of course, one could consider other notions of collusion.

To see that the possibility of including the text of the target concept among the examples does not trivialize the problem of creating a characteristic sample, note
that the learner cannot a priori know which examples in the sample correspond to information from the teacher and which were added by the adversary. Hence, to determine the true concept the learner must rule out possibly spurious examples that the adversary has added. The only information available to the learner, other than the strings in the sample, are their associated labels, indicating their membership in the target concept. Hence, the learner must use semantic information about the target concept to eliminate spurious candidates, which is what we would like from such a learning paradigm. 

The efficient teachability results we prove starting in \Cref{section:outline-positive} in fact include descriptions of algorithms for both the teacher and the learner, and do \emph{not} include the text of the target concept in the sample.
Like Gold's founding result, they do exhibit a high degree of coordination between the teacher and learner.

\subsection{Identification in the limit with polynomial time and data.}
\label{sec:id-in-the-limit-with-poly}
We relate the definitions above to the definition of identification in the limit using polynomial time and data introduced by Gold~\cite{Gold78} and refined by
de la Higuera~\cite{Higuera97}, who also showed that it is closely related to the definition of \emph{semi-poly T/L teachable} introduced by Goldman and Mathias~\cite{GoldmanM96}.

Using the terminology of this paper,
a class $\class{C}$ is \emph{identifiable in the limit using polynomial time and data}
if there exist a polynomial $p(n)$, a teacher $\alg{T}$, and a polynomial time learner $\alg{L}$ such that for every $\aut{C} \in \class{C}$, the sample $\alg{T}(\aut{C})$ is a characteristic sample for $\aut{C}$ and $\alg{L}$, and $|\alg{T}(\aut{C})| \le p(\size(\sema{\aut{C}}))$.
This definition can be viewed as correcting the deficiency of the concept of efficiently learnable (indicated by \Cref{lemma:efficiently_learnable}) by requiring a polynomial bound on the size of the characteristic sample.
In this terminology, ``polynomial time'' refers to the polynomial running time of $\alg{L}$, and ``polynomial data'' refers to the polynomial bound on the size of the sample $\alg{T}(\aut{C})$. Of course, the latter is not a worst-case measure; there could be arbitrarily large finite samples for which $\alg{L}$ outputs an incorrect hypothesis.  

Because the definition of $\class{C}$ being identifiable in the limit using polynomial time and data simply adds the requirement that the learner be computable in polynomial time to the definition of $\class{C}$ having concise characteristic samples, we immediately have the following.

\begin{lemma}
    If $\class{C}$ is identifiable in the limit using polynomial time and data then $\class{C}$ has concise characteristic samples.
\end{lemma}

Comparing with the concept of being efficiently teachable, we note the following differences.
Identifiability in the limit using polynomial time and data only requires the \emph{existence} of a characteristic sample, and the bound on the length of the characteristic sample is in terms of the size of the \emph{smallest} representation of $\sema{\aut{C}}$.
For the definition of efficiently teachable, the characteristic sample must not only exist, but be computable in polynomial time, and the bound on the length of the characteristic sample is in terms of the length of $\aut{C}$.

\begin{lemma}
    If $\class{C}$ satisfies (P1) and (P2) and is efficiently teachable, then $\class{C}$ is identifiable in the limit using polynomial time and data. 
\end{lemma}
\begin{proof}
    By \Cref{theorem:poly-teachable-implies-poly-teachable-learnable}, there are a polynomial time teacher $\alg{T}$ and a polynomial time learner $\alg{L}$ witnessing that $\class{C}$ is efficiently teachable/learnable.  Let $\aut{C} \in \class{C}$ be given.  Let $\aut{C}' \in \class{C}$ minimize $|\aut{C}'|$ subject to $\sema{\aut{C}'} = \sema{\aut{C}}$.  Then $\alg{T}(\aut{C}')$ is a characteristic sample for $\aut{C}$ and $\alg{L}$ of length polynomial in $\size(\aut{C})$.  Thus $\class{C}$ is identifiable in the limit using polynomial time and data.
\end{proof}

To see that identifiability in the limit using polynomial time and data may not imply efficient teachability, we consider the following example.

\begin{theorem}
\label{theorem:poly-learnable-not-poly-teachable}
Assume that there is no polynomial time algorithm for integer factorization.
Then there exists a concept class that satisfies (P1) and (P2) and is identifiable in the limit using polynomial time and data but is not efficiently teachable.
\end{theorem}

\begin{proof}
We assume a standard decimal representation of positive integers.
For any positive integer $n$, let $F(n)$ be the list of primes in its prime factorization in non-decreasing order, for example, $F(54) = (2,3,3,3)$.
Let $\mathbb{N}_+$ denote the set of positive integers, and let $F(\mathbb{N}_+) = \{F(n) \mid n \in \mathbb{N}_+\}$, the set of all prime factorization lists.
If $\ell$ is a prime factorization list $(p_1,p_2,\ldots,p_k)$ then let $M(\ell)$ denote the product of the primes in $\ell$, that is, $M(\ell) = p_1 \cdot p_2 \cdots p_k$.
$M$ is computable in polynomial time.

We define a class of concepts as follows.
$\class{C}_{\textit{fact}}$ has domain $F(\mathbb{N}_+)$ and consists of all finite subsets of $F(\mathbb{N}_+)$.  
We represent a finite set $\{\ell_1,\ell_2,\ldots,\ell_m\} \subset F(\mathbb{N}_+)$ as a finite list $(n_1,n_2,\ldots,n_m)$ such that $n_i = M(\ell_i)$.
Thus, $\{(2,3,3,3)\}$ is represented by the string $(54)$, and $\{(3),(2,3),(3,5,5)\}$ is represented by the string $(3,6,75)$.

Because primality can be decided in polynomial time~\cite{agrawal2004primes}, there is a polynomial time algorithm to test whether $F(n) = (p_1,p_2,\ldots,p_k)$ given $n$ and $(p_1,p_2,\ldots,p_k)$ as inputs. 
Therefore (P1) holds for $\class{C}_{\textit{fact}}$.

Consider the learner $\alg{L}$ that on input a finite sample $T$ first checks that every example is a prime factorization list and that no list is given two different labels.
(If this check fails, the output is arbitrary.)
Let $T_1$ denote the set of examples in $T$ with label $1$.
$\alg{L}$ outputs a representation of the concept $T_1 = \{\ell_1,\ell_2,\ldots,\ell_m\}$ as $(M(\ell_1),M(\ell_2),\ldots,M(\ell_m))$.
Thus property (P2) holds and $\class{C}_{\textit{fact}}$ is learnable in polynomial time and has characteristic samples of polynomial length, consisting of exactly the finitely many positive examples of each concept.
Thus, $\class{C}_\textit{fact}$ is identifiable in the limit using polynomial time and data.

Now suppose for the sake of contradiction that there is a polynomial time teacher $\alg{T}$ for $\class{C}_{\textit{fact}}$.  Consider the sample produced by $\alg{T}$ for the concept $()$, which denotes the empty set.  This consists of a finite number of negative examples, say $\ell_1,\ell_2,\ldots,\ell_m$, where each $\ell_i$ is a prime factorization list.
Let $S = \{M(\ell_i) \mid 1 \le i \le m\}$.
Consider any positive integer $n \not\in S$, and let $F(n) = (p_1,p_2,\ldots,p_k)$.
If the sample $\alg{T}((n))$ does not contain $F(n)$ as a positive example, it must consist exclusively of negative examples, and both the empty set and $\{F(n)\}$ are consistent with the union of the samples produced by $\alg{T}$ for the two concepts, contradicting \Cref{lemma:key-property-char-samples}.
Thus, for all but finitely many positive integers $n$, $\alg{T}$ must produce the prime factorization of $n$. 
This implies that there is a polynomial time algorithm for prime factorization, contradicting our assumption.
\end{proof}

\subsection{Relation to learning with equivalence and membership queries}\label{sec:relation-to-learning-with-mq-and-eq}

In the paradigm of \emph{learning with membership and equivalence queries}, a learning algorithm can access an oracle that truthfully answers two types of queries about the target concept $\aut{C}$, and its goal is to halt and output a representation of $\sema{\aut{C}}$.
In a \emph{membership query}, or MQ, the learning algorithm provides a string $x$ and the answer is $1$ or $0$ depending on whether $x \in \sema{\aut{C}}$ or not.
In an \emph{equivalence query}, or EQ, the learning algorithm provides a representation $\aut{C}' \in \class{C}$, and the answer from the oracle is either ``yes'', if $\sema{\aut{C}'} = \sema{\aut{C}}$, and otherwise is an arbitrarily chosen element of  $\sema{\aut{C}'} \oplus \sema{\aut{C}}$ (a \emph{counterexample} to the conjecture that $\aut{C}'$ is correct).
If the learning algorithm successfully learns every $\aut{C} \in \class{C}$ and at every point its running time is bounded by a polynomial in $\size(\sema{\aut{C}})$ and the length of the longest counterexample seen to that point, we say that $\class{C}$ is polynomially learnable using membership and equivalence queries.

\begin{figure}
    \begin{center}
    \scalebox{.8}{
    \begin{tikzpicture}[->,>=stealth',shorten >=1pt,node distance=6cm,semithick,initial text=,initial where=above]
    \node[paradigm] (mqeq) {\begin{tabular}{c}
         Polynomial \\ MQ \& EQ \\ algorithm
    \end{tabular}};
    \node[paradigm] (iiptd) [below left of=mqeq]   {\begin{tabular}{c}Identification in the \\ limit using polynomial\\ time and data\end{tabular}};
    \node[paradigm] (effteach) [below right of=mqeq]   {\begin{tabular}{c} Efficiently \\ Teachable
    \end{tabular}};
    \node[paradigm] (effteachlear) [below of=mqeq, node distance=8cm]   {\begin{tabular}{c} Efficiently \\ Teachable/ \\Learnable
    \end{tabular}};

    \path (mqeq) edge  node [sloped,above] {\small{(P1),(P2)}} (iiptd);
    \path (mqeq) edge  node [sloped,below] {\small{\cite{GoldmanM96}}} (iiptd);
    \path (mqeq) edge  node [sloped,above] {\small{(P1),(P2),(P3)}} (effteach);
    \path (mqeq) edge  node [sloped,below] {\small{\Cref{cor:MQ_EQ_to_efficiently_teachable}}} (effteach);

    \path (effteach) edge  node [sloped,above] {\small{(P1),(P2)}} (effteachlear);
    \path (effteach) edge  node [sloped,below] {\small{\Cref{theorem:poly-teachable-implies-poly-teachable-learnable}}} (effteachlear);    
    \path (effteachlear) edge  node [sloped,above] {} (iiptd);
    \path (effteachlear) edge  node [sloped,below] {\small{by definition\phantom{--.}}} (iiptd);
  
    \path (iiptd) edge  node [sloped,above,yshift=0.3em] {\small{\begin{tabular}{c}(P1),(P2) \\ Factorization is hard\end{tabular}}} (effteach);
    \path (iiptd) edge  node [sloped,midway] {$\large{\boldmath{\diagup}}$} (effteach);      
    \path (iiptd) edge  node [sloped,below,yshift=-0.3em] {\small{\Cref{theorem:poly-learnable-not-poly-teachable}}} (effteach);    

    \end{tikzpicture}}
    \caption{Summary of main general results about efficient teachability}
    \label{fig:teachability-results}
    \end{center}
\end{figure}

Goldman and Mathias~\cite{GoldmanM96} prove that any class that
can be learned by a deterministic polynomial time algorithm
using any of a large set of example-based queries is teachable by
a computationally unbounded teacher and a polynomial time learner.
A corollary of their Theorem 2 is the following. 
\begin{theorem}\label{thm:GM96-ilptd}
    Suppose the class $\class{C}$ satisfies properties (P1) and (P2) and is concisely distinguishable.  If $\class{C}$ is polynomially learnable using membership and equivalence queries, it is also identifiable in the limit using polynomial time and data.
\end{theorem}
(This result is also proved by Bohn and L\"{o}ding~\cite{BohnL21}.)
Can this result be strengthened to conclude that $\class{C}$ is polynomially teachable?
To answer this question, we examine the proof in more detail.
Let $\alg{A}$ be a learning algorithm using membership and equivalence queries that learns $\class{C}$ in polynomial time.
Given $\aut{C}$, $\alg{T}$ constructs a sample $T_{\aut{C}}$ by simulating $\alg{A}$ and answering its queries according to $\aut{C}$ as follows.
A membership query with $x$ is answered by determining whether $x \in \sema{\aut{C}}$.
For an equivalence query with $\aut{C}'$, if $\sema{\aut{C}'} \neq \sema{\aut{C}}$, then $x$ is chosen to be the shortlex least element of
$\sema{\aut{C}'} \oplus \sema{\aut{C}}$ and returned as the counterexample to the simulation of $\alg{A}$.

The counterexample $x$ is of length polynomial in the sum of the sizes of $\aut{C}'$ and $\aut{C}$ by the assumption of concise distinguishability.
If instead $\sema{\aut{C}'} = \sema{\aut{C}}$, then the sample $T_{\aut{C}}$ is constructed of all the strings $x$ that appeared in membership queries or as counterexamples returned to equivalence queries during the simulation, labeled to be consistent with $\aut{C}$.
Because of the polynomial running time of $\alg{A}$ and the choice of shortest counterexamples, the length of $T_{\aut{C}}$ is bounded by a polynomial in $\size(\sema{\aut{C}})$.

The corresponding learning algorithm $\alg{L}$ takes a sample $T$ as input and simulates the learning algorithm $\alg{A}$, attempting to answer its queries using $T$ as follows.
For a membership query with $x$, if $x$ is an example in $T$, the answer is its label in $T$.
If $x$ is not an example in $T$, then $\alg{L}$ outputs a default $\aut{C}_T$ consistent with $T$ and halts (using property (P2)).
For an equivalence query with $\aut{C}$, $\alg{L}$ checks whether $\aut{C}$ is consistent with $T$ (using property (P1)).
If it is consistent, then it outputs $\aut{C}$ and halts.
If it is not consistent, it finds the shortlex least $x$ that is an example in $T$ whose label is not consistent with $\aut{C}$ and returns $x$ as the counterexample to $\alg{A}$'s equivalence query.
The running time of $\alg{L}$ is polynomial in the length of $T$.
It is because of the choice of the shortlex counterexample by both $\alg{T}$ and $\alg{L}$ that $\alg{T}$ can anticipate exactly the queries that will be made in the simulation of $\alg{A}$ by $\alg{L}$, even when the sample $T$ is a superset of the characteristic sample $T_{\aut{C}}$.

What must we assume in order that the function $\alg{T}$ in this construction can be computed in polynomial time?
Property (P1) allows $\alg{T}$ to answer membership queries in polynomial time, and property (P3) ensures that $\class{C}$ is concisely distinguishable, and allows $\alg{T}$ to find shortlex least counterexamples in polynomial time, so we have the following.

\begin{corollary}
\label{cor:MQ_EQ_to_efficiently_teachable}
    Suppose the class $\class{C}$ satisfies properties (P1), (P2) and (P3).
    If $\class{C}$ is polynomially learnable using membership and equivalence queries, then it is also efficiently teachable.
\end{corollary}

\Cref{fig:teachability-results}  provides a summary of the main general results about efficient teachability.

\section{Focusing on regular \texorpdfstring{$\omega$}{omega}-languages}\label{sec:which}
The rest of the paper concerns efficient teachability of regular $\omega$-languages. This section starts by describing 
examples and samples for $\omega$-languages, and continues by describing some known or immediate results on the subject.

\subsection{Examples and samples for \texorpdfstring{$\omega$}{omega}-languages}\label{subsec:examples-and-samples-for-omega-langs}
Because we require finite representations of examples,  $\omega$-words in our case, we work with ultimately periodic words, that is, words of the form $u(v)^\omega$ where $u\in\Sigma^*$ and $v\in\Sigma^+$. It is known that two regular $\omega$-languages are equivalent iff they agree on the set of ultimately periodic words~\cite{Buchi62,CalbrixNP93old}, so this choice is not limiting.

The example $u(v)^{\omega}$ is concretely represented by the pair $(u,v)$ of finite strings, and its length is $|u|+|v|$.  A \emph{labeled example} is a pair $(u(v)^{\omega},l)$, where the label $l$ is either $0$ or $1$.  A \emph{sample} is a finite set of labeled examples such that no example is assigned two different labels.  The \emph{length} of a sample is the sum of the lengths of the examples that appear in it.  A sample $T$ and a language $L$ are \emph{consistent} with each other if and only if for every labeled example $(u(v)^{\omega},l) \in T$, $l = 1$ iff $u(v)^{\omega} \in L$.  A sample $T$ and an acceptor $\aut{A}$ are \emph{consistent} with each other if and only if $T$ is consistent with $\sema{\aut{A}}$.
The following results give two useful procedures on examples that are computable in polynomial time.

\begin{proposition}
	\label{prop:inequality-bound}
	Let $u_1, u_2 \in \Sigma^*$ and $v_1, v_2 \in \Sigma^+$. If $u_1(v_1)^{\omega} \neq u_2(v_2)^{\omega}$ then they differ in at least one of the first $\ell$ symbols, for $\ell = \max(|u_1|,|u_2|)+ |v_1|\cdot|v_2|$.
\end{proposition}

Let $\suffixes(u(v)^{\omega})$ denote the set of all $\omega$-words that are suffixes of $u(v)^{\omega}$.

\begin{proposition}
	\label{prop:computing-suffixes}
	The set $\suffixes(u(v)^{\omega})$ consists of at most $|u|+|v|$ different examples: one of the form $u'(v)^{\omega}$ for every nonempty suffix $u'$ of  $u$, and one of the form $(v_2 v_1)^{\omega}$ for every 
    division of $v = v_1 v_2$ into a non-empty prefix $v_1$ and suffix $v_2$.
\end{proposition}

\subsection{Negative results for nondeterministic classes}
The classes \class{NBA}, \class{NPA}, \class{NMA}, \class{NCA} do not have concise characteristic sets~\cite{AngluinFS20}. The proof is by constructing a family of languages $\{L_n\}_{n\in\mathbb{N}}$ with an acceptor of size quadratic in $n$ for which at least one word of length at least exponential in $n$ must be included in any characteristic sample for $L_n$.\footnote{A negative result regarding query learning of $\class{NBA}$, $\class{NPA}$ and $\class{NMA}$ was obtained by Angluin et al.~\cite{AngluinAF20}. That result makes a plausible assumption of cryptographic hardness, which is not required here.} Since an NBA (resp. NCA) is a special case of NRA (resp. NSA) the same is true for \class{NRA} and \class{NSA}.
\begin{theorem}
The classes \class{NBA}, \class{NPA}, \class{NMA}, \class{NCA}, \class{NRA} and \class{NSA} do not have concise characteristic sets, and therefore are neither identifiable in the limit using polynomial time and data nor efficiently teachable.
\end{theorem}

\subsection{Consequences of membership and equivalence algorithms}
\label{sec:positive-results}

In the domain of $\omega$-automata, researchers have so far found very few polynomial time learning algorithms using membership and equivalence queries.
We consider the cases of Mod 2 multiplicity automata, strongly unambiguous B\"{u}chi automata, and deterministic weak parity automata.

A B\"uchi automaton is \emph{unambiguous} if no word has more than one run starting in the initial state and visiting an accepting state infinitely often. It is \emph{strongly unambiguous} if no word has more than one run visiting an accepting state infinitely often, whether it starts at the initial state or not.
For instance, consider the deterministic B\"uchi automaton $\aut{B}=(\{a,b\},\{q_a,q_b\},q_a,\delta,\{q_b\})$ with $\delta(q_a,a)=\delta(q_b,a)=q_a$
and $\delta(q_a,b)=\delta(q_b,b)=q_b$. Since $\aut{B}$ is deterministic there is a unique run on every word, and thus it is also unambiguous. However, since the word $(b)^\omega$ is accepted from both $q_a$ and $q_b$, $\aut{B}$ is \underline{not} strongly unambiguous. 

Angluin et al.~\cite{AngluinAF20} give a polynomial time mapping $r$ of a strongly unambiguous B\"uchi automata (SUBA) $\aut{C}$ to a representation $r(\aut{C})$ as a Modulo-2 multiplicity automaton (M2MA), and observe that there is a polynomial time algorithm for learning M2MAs using membership and equivalence queries~\cite{BeimelBBKV:2000}.
(Please see \cite{AngluinAF20} for precise definitions.)
We note that a shortlex least counterexample can be returned in case of inequivalence~\cite{Sakarovitch2009,SakarovitchBook2009}.
It follows from \Cref{cor:MQ_EQ_to_efficiently_teachable} that M2MAs are efficiently teachable.
We then also have the following.

\begin{corollary}
    The class $\class{SUBA}$ is efficiently teachable. 
\end{corollary}
\begin{proof}
    Let $\alg{T}$ be a teacher witnessing the efficient teachability of M2MAs.
    The teacher $\alg{T}'$ with input a SUBA $\aut{C}$ generates the characteristic sample $\alg{T}(r(\aut{C}))$, to which it adds one example $(x,b)$ such that $x = \aut{C}$ and $b = 1$ iff $x \in \sema{\aut{C}}$.  The learner $\alg{L}$ with input $T$ searches for the least $x$ (if any) such that $(x,b) \in T$, $r(x)$ is consistent with $T$, and $\alg{T}(r(x)) \subseteq T$.
    If such an $x$ is found, it is output; otherwise, $\alg{L}$ constructs and outputs a SUBA consistent with $T$. A procedure constructing a default SUBA acceptor that agrees with a given sample $T$ is given in the proof of \Cref{prop:default-acceptor}.   
\end{proof}
Maler and Pnueli~\cite{MalerP95} give an algorithm that learns the class \class{DWPA}  of deterministic weak parity automata in polynomial time using membership and equivalence queries.
The weak parity condition is obtained from the parity condition using $\occ(\rho)$ instead of $\inf(\rho)$ where $\occ(\rho)$ is the set of states visited somewhere during the run $\rho$. It is known that $\class{DWPA}=\class{DBA}\cap\class{DCA}$.
Membership and equivalence of $\class{DWPA}$ are decidable in polynomial time, thus by \Cref{thm:GM96-ilptd}, $\class{DWPA}$ is identifiable in the limit using
polynomial time and data. It is not known whether there is a polynomial time algorithm for equivalence with shortlex counterexamples, so we are unable to apply \Cref{cor:MQ_EQ_to_efficiently_teachable} to deduce they are also efficiently teachable.  However, DWPAs are a special case of DPAs and the following is a corollary of \Cref{theorem:informative-classes-efficiently-teachable-learnable} for DPAs.
\begin{corollary}
    The class \class{DWPA} is efficiently teachable.
\end{corollary}

The learning algorithm of Maler and Pnueli for \class{DWPA} exploits the fact that  $\class{DWPA}=\class{DBA}\cap\class{DCA}=\class{IBA}\cap\class{ICA}$, that is, this class is fully informative and does have a one-to-one relationship between states of the minimal DWPA for a language $L$ and the equivalence classes of the right congruence $\sim_L$. 
The class  \class{DWPA} is a small sub-class of the class of the fully informative regular $\omega$-languages --- there are fully informative languages in every level of the Wagner hierarchy~\cite{AngluinF18}, whereas \class{DWPA} is one of the lowest levels in the hierarchy. The rest of the paper is dedicated to showing that fully informative languages of any of the considered $\omega$-automata types (B\"uchi, coB\"uchi, parity, Muller, Rabin and Streett) are efficiently teachable.

\section{The informative classes are efficiently teachable}
\label{section:outline-positive}

This section covers some preliminary issues and gives an overview of the milestones needed to prove that the informative classes are efficiently teachable.

\subsection{Duality}

There are reductions of the problem of efficient teachability between $\class{IBA}$ and $\class{ICA}$ and between
$\class{IRA}$ and $\class{ISA}$, using the duality between these types of acceptors.
Consequently we focus on the classes $\class{IBA}$, $\class{IPA}$, $\class{IRA}$ and $\class{IMA}$ in what follows.\footnote{The results regarding the classes $\class{IBA}$ (and
$\class{ICA}$),
$\class{IPA}$ and
$\class{IMA}$ were obtained in~\cite{AngluinFS20}; here we extend them to the classes $\class{IRA}$ (and $\class{ISA}$). Results for identifiability in the limit using polynomial time and data (but not efficient teachability) of the classes $\class{IRA}$ (and $\class{ISA}$) have also been provided in~\cite{BohnL21} using a different algorithm.}

\begin{proposition}
	\label{prop:ib-ic-ir-is-duality}
	 \class{IBA} (resp., \class{IRA}) is efficiently teachable if and only if \class{ICA} (resp., \class{ISA}) is.
\end{proposition}
\begin{proof}
	Let $\aut{A}$ be an ICA.  Because $\aut{A}$ is deterministic and complete, if we let $\aut{A}'$ denote the IBA with the same components as $\aut{A}$, then $\aut{A}'$ accepts the complement of the language $\aut{A}$, by Claim~\ref{clm:basic-relations-between-omega-aut}~(\ref{claim:DBA-DCA-complement}).
	
	We modify the characteristic sample for $\aut{A}'$ by complementing all its labels to get a characteristic sample for $\aut{A}$.  The algorithm to learn an ICA from a sample $T$ is obtained by complementing all the labels in the sample $T$ and calling the algorithm to learn an IBA from a sample.  The resulting IBA, now considered to be an ICA, is returned as the answer.
	
	The same conversion may be done with acceptors of types IRA and ISA, by \linebreak[5]Claim~\ref{clm:basic-relations-between-omega-aut}~(\ref{claim:DRA-DSA-complement}).
\end{proof}

\subsection{The default acceptor}

One condition of the definition of being efficiently teachable is that the learning algorithm must run in polynomial time and return an acceptor of the required type that is consistent with the input sample $T$, even if the sample $T$ does not subsume a characteristic sample.
To meet this condition, we use the strategy of Gold's construction, that is, the learning algorithm optimistically assumes that the sample includes a characteristic sample, and if that assumption fails to produce an acceptor consistent with the sample, the algorithm instead produces a \emph{default acceptor} to ensure that its hypothesis is consistent with the sample. 
Alternatively, one can use Bohn and L\"{o}ding's generalization of the RPNI algorithm to learning $\omega$-automata, which has a more complex default strategy~\cite{BohnL21}.

\begin{figure}
    \begin{center}
    \begin{tikzpicture}[->,>=stealth',shorten >=1pt,auto,node distance=1.2cm,semithick,initial text=,initial where=above]
        \node[ministate,initial] (start) {\emptyst};
        \node[ministate] (a) [below left of=start]   {\emptyst};
        \node[ministate] (b) [below right of=start]   {\emptyst};
        \node[ministate,accepting] (ba) [below right of=b]   {\emptyst};
        \node[ministate,accepting] (bb) [right of=b]   {\emptyst};

        \node[ministate] (aa) [below  left of=a] {\emptyst};
        \node[ministate] (aaa) [ left of=aa] {\emptyst}; 
        \node[ministate,accepting] (aab) [below left of=aa] {\emptyst};      
        \node[ministate] (ab) [below right of=a] {\emptyst};
        \node[ministate] (aba) [below left of=ab] {\emptyst};
        \node[ministate] (abb) [below right of=ab] {\emptyst};
        \node[ministate,accepting] (abba) [below left of=abb] {\emptyst}; 
        \node [ministate,accepting] (abbb) [below right of=abb] {\emptyst};

        \path (start) edge node [above] {$a$} (a);
        \path (start) edge node [above] {$b$} (b);
        \path (a) edge node [above] {$a$} (aa);       
        \path (a) edge node [above] {$b$} (ab);
        \path (b) edge node [above] {$a$} (ba);       
        \path (b) edge node [above] {$b$} (bb);              
        \path (aa) edge node [above] {$b$} (aab);       
        \path (aa) edge node [above] {$a$} (aaa);        
        \path (ab) edge   node [above] {$a$} (aba);
        \path (ab) edge node {$b$} (abb);
        \path (abb) edge   node [above] {$a$} (abba);
        \path (abb) edge   node [above] {$b$} (abbb);
        \path (ba) edge [loop right] node [above] {$a,b$} (ba);
        \path (bb) edge [loop right] node [above] {$a,b$} (bb); 
        \path (aaa) edge [loop left] node [above] {$a,b$} (aaa);
        \path (aab) edge [loop left] node [above] {$a,b$} (aab);
        \path (aba) edge [loop left] node [above] {$a,b$} (aba);
        \path (abba) edge [loop left] node [above] {$a,b$} (abba);
        \path (abbb) edge [loop right] node [above] {$a,b$} (abbb);

        \node[label] (words) [right of=ab, node distance=5.05cm]   
        {
        $\begin{array}{r@{~~}l}
              (a)^\omega & \texttt{aaa}|\texttt{aaaa\ldots}  \\
              a(b)^\omega & \texttt{abbb}|\texttt{bbb\ldots}  \\
              (ab)^\omega & \texttt{aba}|\texttt{baba\ldots} \\
              ab(baa)^\omega & \texttt{abba}|\texttt{aba\ldots} \\
              bab(aab)^\omega & \texttt{ba}|\texttt{baaba\ldots} \\
              (aab)^\omega & \texttt{aab}|\texttt{aaba\ldots} \\    
              bb(aba)^\omega & 
              \texttt{bb}|\texttt{abaab\ldots}
        \end{array}$
        };

        \node[label] (sample) [above of=words, node distance=3cm] 
        {$T=\{((a)^\omega,0), (a(b)^\omega, 1), ((ab)^\omega,0), (ab(baa)^\omega,1),(bab(aab)^\omega,1),\allowbreak((aab)^\omega,1), (bb(aba)^\omega,1)\}\phantom{--------}$};
        \node[label] (prefixes) [below of=words, node distance=3.0cm]   {$\begin{array}{l@{\,}l@{\,}l} U=\{&aaa,abbb,aba,\\
        & abba,ba,aab,bb&\}\end{array}$};

        \node[ministate,initial] (1) [right of=start, node distance=8.0cm]   {\emptyst};
        \node[ministate,accepting] (1lb) [below of=1]   {\emptyst};        

        \node[ministate,initial] (2) [right of=1, node distance=1cm]   {\emptyst};
        \node[ministate] (2a) [below of=2]   {\emptyst};        
        \node[ministate,accepting] (2lbaa) [below of=2a]   {\emptyst};        
        \node[ministate,initial] (2laab) [above right of=2lbaa]   {\emptyst};        
        \node[ministate] (2laba) [below right  of=2laab]   {\emptyst}; 
        \node[ministate,initial,initial where=right] (2b) [below  of=2laba]   {\emptyst};  
        \node[ministate,initial,initial where=right] (2bb) [below  of=2b]   {\emptyst};  
        
        \node[label] (periods) [left of=2bb, node distance=1.7cm]   {\phantom{--}$\begin{array}{l}
        V=\{b,baa\}\\
        X_{b}=\{a\}\\
        X_{baa}=\{ab\}\\
        X_{aab}=\{\epsilon\}\\
        X_{aba}=\{b,bb\} \\ \phantom{--}\\
        \end{array}$};

        \path (1) edge node [left] {$a$} (1lb);       
        \path (1lb) edge [loop below] node   {$b$} (1lb);       

        \path (2) edge  node  [left] {$a$} (2a);       
        \path (2a) edge node [left] {$b$} (2lbaa);       
        \path (2lbaa) edge node [above] {$b$} (2laab);       
        \path (2laab) edge node [right] {$a$} (2laba);     
        \path (2laba) edge node [below] {$a$} (2lbaa);       
        \path (2b) edge node [left] {$b$} (2laba);     
        \path (2bb) edge node [left] {$b$} (2b);

    \end{tikzpicture} 
    \caption{Top: a sample $T$.
    Middle: The distinguishing prefixes $U$ of words in the sample $T$.
    Left: Default acceptor of type DBA for sample $T$ (the dead state and the transitions to it are omitted).
    Right: Default acceptor of type SUBA for sample $T$.}
	\label{fig:example-default-acceptor}
    \end{center}
\end{figure}

The construction of the default acceptor is given in the proof of the following proposition, and is accompanied by an example illustrated
in \Cref{fig:example-default-acceptor}.

\begin{proposition}
	\label{prop:default-acceptor}
	There is a polynomial time algorithm that takes a sample $T$ and returns a DBA (resp., DCA, DPA, DRA, DSA, DMA, SUBA) consistent with $T$.
\end{proposition}

\begin{proof}
    Given a sample $T$, we initialize $U$ to be the empty set, and for every word $w_i$ in the sample, we find the shortest prefix $u_i$ of  $w_i$  that distinguishes it from all other examples in $T$ and add it to $U$ (see \Cref{fig:example-default-acceptor}, middle).    
    We arrange the finite words in $U$ in a trie in the usual manner.
    We add self-loops on each $\sigma\in\Sigma$ to the leaves of the trie 
    (see \Cref{fig:example-default-acceptor}, left).
    This deterministic automaton is termed the prefix-tree automaton~\cite{OncinaG92}. If the automaton is incomplete, we add a new dead state with self-transitions on each $\sigma\in\Sigma$, and define all undefined transitions to go to the dead state.
    Recall that by \Cref{prop:inequality-bound} the length of a prefix distinguishing two examples $u_1(v_1)^\omega$ and $u_2(v_2)^\omega$ is polynomially bounded by the length of the examples. It follows that the prefix-tree automaton can be constructed in time polynomial in the length of the sample $T$.  
	
    For a DBA, the acceptance condition $F$ consists of all the trie-leaf states that are prefixes of positive examples in $T$ (see \Cref{fig:example-default-acceptor}, left).
    For a DMA, the acceptance condition consists of $\{\{q\} \mid q \in F\}$. For a DCA, the acceptance condition consists of the dead state (if one was added) and all the trie-leaf states that are not prefixes of positive examples in $T$.  The DBA thus constructed may be transformed to a DPA or a DRA using Claim \ref{clm:basic-relations-between-omega-aut}~(\ref{claim:NBA-to-NPA}) or Claim~\ref{clm:basic-relations-between-omega-aut}~(\ref{claim:NBA-to-NRA}), respectively, and the DCA may be transformed to a DSA using Claim~\ref{clm:basic-relations-between-omega-aut}~(\ref{claim:NCA-to-NSA}).

    For a SUBA the construction is different. We first construct a set $V$ of shortest periods of positively labeled examples, by going iteratively over the examples $w_1,w_2,\ldots$ and proceeding as follows. For $w_i=u(v)^\omega$ if $v_i$ is a shortest period of $w_i$, and none of the rotations of $v_i$ is in $V$ we add $v_i$ to $V$. 
    Then for every period $v=\sigma_1\sigma_2\ldots\sigma_k$ in $V$ we construct an automaton with $k$ states $s_0,s_1,\ldots,s_{k-1}$ and transitions $(s_i,\sigma_{i+1},s_{i'})$ for $0\leq i < k$ where
    $i'=(i+1)\mod k$. We designate $s_0$ as an accepting state. 
    Finally, for each positively labeled word $w$ we look for the shortest prefix $x$ of $w$ such that $w=x(y)^\omega$ for some rotation $y=\sigma_i\ldots\sigma_k\sigma_1\ldots\sigma_{i-1}$ of a period
    $v=\sigma_1\sigma_2\ldots\sigma_k\in V$.
    For a rotation $y$ of $v\in V$, let 
    $X_y$ be the set of such shortest prefixes.
    The prefixes in $X_y$ are arranged in a suffix-sharing trie and the trie is connected to the automata constructed for the periods by landing in the state reading $\sigma_i$ in the automaton of $v$. Each node of this trie corresponding to a start of a prefix in $X_y$ is added to the set of initial states (see \Cref{fig:example-default-acceptor}, right).
    It is easy to see that the constructed NBA is consistent with the sample. To see that it is a SUBA, consider a word $w\in\Sigma^\omega$ and assume $w=xy^\omega$ for $y$ a shortest period of $w$ and $x$ the shortest prefix reaching such $y$. Note that there is only one such representation of $w$. 
    By the construction of the SUBA, $w$ can only be accepted via a cycle reading $y$ and there is only one such cycle, moreover the cycle can be entered only at the position after reading $x$, and by the suffix-sharing trie there is only one state from which reading $x$ gets to this position.
\end{proof}

\subsection{Strongly connected components}
\label{ssec:SCCs}

The acceptance conditions that we consider are all based on the set of states visited infinitely often in a run of the automaton on an input $w \in \Sigma^{\omega}$.  We consider only acceptors whose automata are deterministic and complete, so for any $w \in \Sigma^{\omega}$ there is exactly one run, which we denote $\rho(w)$, of the automaton on input $w$.
Thus we may define $\inf(w) = \inf(\rho(w))$, the set of states visited infinitely often in this unique run.  
In the run $\rho(w)$, there is some point after which none of the states visited finitely often is visited.
Because each state in $\inf(w)$ is visited infinitely often, for any states $q_1, q_2 \in \inf(w)$, there exists a non-empty word $x \in \Sigma^*$ such that $\delta(q_1,x) = q_2$ and for each prefix $x'$ of $x$, $\delta(q_1,x') \in \inf(w)$, that is, the path from $q_1$ to $q_2$ on $x$ does not visit any state outside the set $\inf(w)$.

These properties motivate the following definition. Given an automaton $\aut{M}$, a \emph{strongly connected component} (SCC) of $\aut{M}$ is a nonempty set of states $C$ such that for every $q_1, q_2 \in C$, there exists a nonempty string $x \in \Sigma^*$ such that $\delta(q_1,x) = q_2$ and for any prefix $x'$ of $x$, $\delta(q_1,x')\in C$.  

Note that an SCC need not be maximal, and that a singleton state set $\{q\}$ is an SCC if and only if the state $q$ has a self-loop, that is, $\delta(q,\sigma) = q$ for some $\sigma \in \Sigma$.
There is a close relationship between SCCs and the set of states visited infinitely often in a run.

\begin{proposition}
	\label{prop:inf-to-SCC}
	Let $\aut{M}$ be a complete deterministic automaton and $w \in \Sigma^\omega$.   
	Then $\inf(w)$ is an SCC of $\aut{M}$.  If $w$ is the ultimately periodic word $u(v)^{\omega}$, then $\inf(w)$ may be computed in time polynomial in the size of $\aut{M}$ and the length of $u(v)^{\omega}$.
\end{proposition}

\begin{proposition}
	\label{prop:SCC-to-inf}
	For any deterministic automaton $\aut{M} = \la \Sigma, Q, q_{\iota}, \delta \ra$ and any reachable SCC $C$ of $\aut{M}$,  there exists an ultimately periodic word $w = u(v)^\omega$ of length at most $|Q|+|C|^2$ such that $C = \inf(w)$.
	Such a word may be found in time polynomial in $|Q|$ and $|\Sigma|$.
\end{proposition}

\begin{proof}
	Because $C$ is reachable, a word $u \in \Sigma^*$ of minimum length such that $\delta(q_{\iota},u) \in C$ may be found by breadth first search.  The length of $u$ is at most $|Q|$. If $C = \{q\}$, then there is at least one symbol $\sigma \in \Sigma$ such that $\delta(q,\sigma) = q$.  Then the $\omega$-word $w = u(\sigma)^\omega$ is such that $C = \inf(w)$.  The length of this ultimately periodic word is at most $|Q|+1$.
	
	If $C$ contains at least two states, let $q_1, \ldots, q_k$ be the states in $C$ that are not $q$.  Then for each $i$, there exist two nonempty finite words $x_i$ and $y_i$ each of length at most $n$ such that $\delta(q,x_i) = q_i$ and $\delta(q_i,y_i) = q$, and the path on $x_i$ from $q$ to $q_i$ and the path on $y_i$ from $q_i$ to $q$ do not visit any states outside of $C$. The words $x_i$ and $y_i$ may be found in polynomial time by breadth-first search.
	Then the word $w = u(x_1y_1 \cdots x_ky_k)^\omega$ is such that $\inf(w) = C$.  The length of this ultimately periodic word is at most $|Q|+|C|^2$.
\end{proof}

We let $\alg{Witness}(C,\aut{M})$ denote the ultimately periodic word $u(v)^{\omega}$ returned by the algorithm described in the proof above for the reachable SCC $C$ of automaton $\aut{M}$.

\begin{proposition}
	\label{prop:unionofSCCs}
	If $C_1$ and $C_2$ are SCCs of automaton $\aut{M}$ and $C_1 \cap C_2 \neq \emptyset$, then $C_1 \cup C_2$ is also an SCC of $\aut{M}$.
\end{proposition}

If $\aut{M}$ is an automaton and $S$ is any set of its states, define $\SCCs(S)$ to be the set of all $C$ such that $C \subseteq S$ and $C$ is an SCC of $\aut{M}$.  Also define $maxSCCs(S)$ to be the maximal elements of $\SCCs(S)$ with respect to the subset ordering.
The following is a consequence of \Cref{prop:unionofSCCs}.

\begin{proposition}
	\label{prop:maxSCCs_disjoint}
	If $\aut{M}$ is an automaton and $S$ is any set of its states, then the elements of $\maxSCCs(S)$ are pairwise disjoint, and every set $C \in SCCs(S)$ is a subset of exactly one element of $\maxSCCs(S)$.
\end{proposition}

There are some differences in the terminology related to strong connectivity between graph theory and omega automata, which we resolve as follows.
In graph theory,
a \emph{path} of length $k$ from $u$ to $v$
in a directed graph $(V,E)$ is a finite
sequence of vertices $v_0, v_1, \ldots, v_k$
such that
$u = v_0$, $v = v_k$ and
for each
$i$ with $i \in [1..k]$, $(v_{i-1},v_i) \in E$.
Thus, for every vertex $v$, there is a path of length $0$
from $v$ to $v$.
A set of vertices $S$ is \emph{strongly connected}
if and only if for all $u, v \in S$, there is a path of
some nonnegative length from $u$ to $v$
and all the vertices in the path are elements of $S$.
Thus, for every vertex $v$, the singleton set $\{v\}$ is
a strongly connected set of vertices.
A \emph{strongly connected component} of a directed graph is
a maximal strongly connected set of vertices.
There is a linear time algorithm to find the set of
strong components of a directed graph~\cite{Tarjan72}.

In this paper, we use the terminology SCC and maximal SCC
to refer to the definitions from the theory of omega automata,
and the terminology 
\emph{graph theoretic strongly connected components} 
to refer to the definitions from graph theory.
We use the term \emph{trivial strong component} to refer
to a graph theoretic strongly connected component that is a
singleton vertex $\{v\}$ such that there is no edge $(v,v)$.

If $\aut{M}$ is an automaton, we may define a related directed graph
$G(\aut{M})$ whose vertices are the states of $\aut{M}$ and whose edges $(q_1, q_2)$ are the pairs of states such that $q_2 \in \delta(q_1,\sigma)$ for some $\sigma \in \Sigma$.
Then for any set $S$ of states of $\aut{M}$, the maximal SCCs in $S$, $\maxSCCs(S)$, are the graph theoretic strongly connected components of the subgraph of $G(\aut{M})$ induced by $S$, with
any trivial strong components removed.

\begin{proposition}
	
	\label{prop:maxSCCs-in-linear-time}
	For automaton $\aut{M}$ and any subset $S$ of its states, $\maxSCCs(S)$ can be computed in time linear in the size of $\aut{M}$.
\end{proposition}

\subsection{Proving efficient teachability of the informative classes --- overview}
\label{ssec:overview-of-positive}
We can show that a class is efficiently teachable by first showing that it is identifiable in the limit using polynomial time and data, and then giving a polynomial time teacher to construct the required characteristic samples.
To show that a class is identifiable in the limit using polynomial time and data there are two parts: (i) defining a sample $T_L$ of size polynomial in the size of the given acceptor $\aut{A}$ for the language $L$ at hand, and (ii) providing a polynomial time learning algorithm $\alg{L}$ that for every given sample $T$ returns an acceptor consistent with $T$, and, moreover, for any sample $T$ consistent with $L$ that subsumes $T_L$, returns an acceptor that accepts $L$.

The definition of an acceptor has two parts: (a) the definition of the automaton and (b) the definition of the acceptance condition.
Correspondingly, we view the characteristic sample as a union of two parts: $T_{Aut}$ (to specify the automaton) and $T_{Acc}$ (to specify the acceptance condition).
In \Cref{sec:char-sample-aut} we discuss the construction of $T_{Aut}$, which is common to all the classes we consider, as they all are isomorphic to $\aut{M}_{\sim L}$ where $L = \sema{\aut{A}}$ for the target automaton $\aut{A}$.  We also describe a polynomial time algorithm to construct $\aut{M}_{\sim L}$ using the sample $T_{Aut}$.

Because the acceptance conditions differ, $T_{Acc}$ is different for each type of acceptor we consider.
In \Cref{sec:T-Acc-for-IMAs-IBAs-ICAs} we describe the construction of $T_{Acc}$ for acceptors of types IMA and IBA and learning algorithms for acceptors of these types, showing that $\class{IMA}$, $\class{IBA}$ and $\class{ICA}$ are identifiable in the limit using polynomial time and data.
In \Cref{sec:T-Acc-for-DPAs} we describe the construction of $T_{Acc}$ for acceptors of type IPA and a learning algorithm for acceptors of this type, showing that $\class{IPA}$ is identifiable in the limit using polynomial time and data.
In \Cref{sec:T-Acc-for-IRAs-ISAs} we describe the construction of $T_{Acc}$ for acceptors of type IRA and a learning algorithm for acceptors of this type, showing that $\class{IRA}$ and $\class{ISA}$ are identifiable in the limit using polynomial time and data.

In \Cref{sec:poly-time-char-samples} we show that 
the characteristic samples we have defined can be computed in polynomial time in the size of the acceptor.  These results rely on polynomial time algorithms for the inclusion and equivalence problems for the acceptors.  These are described in Sections~\ref{sec:inclusion-algorithms}, \ref{sec:inclusion-DPAs}, \ref{sec:inclusion-DRAs}, and~\ref{sec:inclusion-DMAs}.

This does not yet entail that the class $\class{IXA}$ from $\{\IB,\IC,\IP,\IM,\IR,\IS\}$
is efficiently teachable. This is because $\IB$ (for instance) includes also DBAs that are not IBAs but have equivalent IBAs. 
In \Cref{sec:computing-right-congruence-automaton} we show the right congruence automaton $\aut{M}_{\sim_L}$ can be computed in polynomial time, given a DBA, DCA, DPA, DRA, DSA, or DMA accepting $L$, which yields a polynomial time algorithm to test whether a DBA is an IBA, and similarly for the other acceptor types.  In \Cref{sec:testing-membership-in-IX}, we consider the harder problem of deciding whether a DBA accepts a language in $\class{IBA}$, and give polynomial time algorithms for DBAs, DCAs, DPAs, DRAs, DSAs and DMAs.  With these results we can finally claim that the $\class{IXA}$ classes are efficiently teachable.

\section{The sample \texorpdfstring{$T_{Aut}$}{T\_Aut} for the automaton}
\label{sec:char-sample-aut}
In this section we describe the construction of the $T_{Aut}$ part of the sample.  We first show that if two states of the automaton are distinguishable, they are distinguishable by words of length polynomial in the number of states of the automaton. 

\subsection{Existence of short distinguishing words}
\label{ssec:short-dist-words}

Let $A$ be an acceptor of one of the types DBA, DCA, DPA, DRA, DSA, or DMA over alphabet $\Sigma$. We say that states $q_1$ and $q_2$ of $\aut{M}$ are \emph{distinguishable} if there exists a word $w \in \Sigma^\omega$ that is accepted from one state but not the other, that is, $w \in \sema{\aut{A}^{q_1}} \setminus \sema{\aut{A}^{q_2}}$ or $w \in \sema{\aut{A}^{q_2}} \setminus \sema{\aut{A}^{q_1}}$. In this case we say that $w$ is a \emph{distinguishing word}.  

\begin{proposition}
	\label{prop:poly_distinguishing_experiments}
	If two states of a complete DBA, DCA, DPA, DRA, DSA, or DMA of $n$ states are distinguishable, then they are distinguishable by an ultimately periodic $\omega$-word of length bounded by $O(n^4)$.
\end{proposition}

\begin{proof}
	We prove the result for a DMA.  Because any DBA, DCA, DPA, DRA, or DSA is equivalent to a DMA with the same automaton, this result holds for these types of acceptors as well. 
	Let $\aut{A}$ be a complete DMA of $n$ states such that the states $q_1$ and $q_2$ are distinguishable. Then there exists an $\omega$-word $w$ that is accepted from exactly one of the two states, that is, $w$ is accepted by exactly one of $\aut{A}^{q_1}$ and $\aut{A}^{q_2}$.  
	
	Let $\aut{M}_i$ denote the automaton of $\aut{A}$ with its initial state replaced by $q_i$ for $i = 1,2$. Let $\aut{M}$ denote the product automaton $\aut{M}_1 \times \aut{M}_2$. The number of states of $\aut{M}$ is $n^2$. By \Cref{prop:inf-to-SCC}, $\infss{M}(w)$ is a reachable SCC $C$ of $\aut{M}$, and by \Cref{prop:SCC-to-inf} there exists an ultimately periodic word $u(v)^{\omega}$ of length bounded by $O(n^4)$ such that $\infss{\aut{M}}(u(v)^{\omega}) = C$.
	Then for $i = 1,2$, $\infss{\aut{M}_i}(u(v)^{\omega}) = \pi_i(C) = \infss{\aut{M}_i}(w)$, so $u(v)^{\omega}$ is also accepted by exactly one of $\aut{A}^{q_1}$ and $\aut{A}^{q_2}$, and $u(v)^{\omega}$ distinguishes $q_1$ and $q_2$. 
\end{proof}

\subsection{Defining the sample \texorpdfstring{$T_{Aut}$}{T\_Aut} for the automaton}
\label{ssec:T-Aut-definition}

We now define the $T_{Aut}$ part of the characteristic sample, given an acceptor $\aut{A} = \la \Sigma, Q, q_\iota, \delta, \alpha \ra$ that is an IBA, ICA, IPA, IRA, ISA, or IMA. This construction is analogous to that of the corresponding part of a characteristic sample for a DFA, with distinguishing experiments that are ultimately periodic $\omega$-words instead of finite strings.

Let $\aut{M}$ be the automaton of $\aut{A}$ and let $n$ be the number of states of $\aut{M}$.
Because $\aut{A}$ is an IBA, ICA, IPA, IRA, ISA, or IMA, every state is reachable and every pair of states is distinguishable.  
We define a distinguished set of $n$ \emph{access strings} for the states of $\aut{M}$ as follows.
For each state $q$, $\access(q)$ is the least string $x$ in the shortlex ordering such that $\delta(q_\iota,x) = q$.
Given $\aut{A}$, the access strings may be computed in polynomial time by breadth first search.

Because every pair of states is distinguishable,
by \Cref{prop:poly_distinguishing_experiments}, there exists a set $E$ of at most $n$ distinguishing experiments, each of length at most $n^2 + n^4$, that distinguish every pair of states.
The issue of computing $E$ is addressed in \Cref{sec:poly-time-char-samples}.
The sample $T_{Aut}$ consists of all the examples in $(S \cdot E )\cup (S \cdot \Sigma \cdot E)$, labeled to be consistent with $\aut{A}$.  There are at most $(1 + |\Sigma|)n^2$ labeled examples in $T_{Aut}$, each of length bounded by a polynomial in $n$.  A learner using $T_{Aut}$ is described next.

\subsection{Learning the automaton from \texorpdfstring{$T_{Aut}$}{T\_Aut}}
\label{ssec:learn-aut}

We now describe a learning algorithm $\alg{L}_{Aut}$ and prove the following.
\begin{theorem}
	\label{theorem:L-Aut-works}
	The algorithm $\alg{L}_{Aut}$ with a sample $T$ as input runs in polynomial time and returns a deterministic complete automaton $\aut{M}$.
	Let $\aut{A}$ be an  acceptor of type IBA, ICA, IPA, IRA, ISA, or IMA.
	If $T$ is consistent with $\aut{A}$ and subsumes $T_{Aut}$ then the returned automaton $\aut{M}$ is isomorphic to the automaton of $\aut{A}$.
\end{theorem}

Algorithm $\alg{L}_{Aut}$ on input $T$ constructs a set $E$ of words that serve as experiments used to distinguish candidate states. For each $(u(v)^{\omega}, l)$ in $T$, all of the elements of $\suffixes(u(v)^{\omega})$ are placed in $E$.
Two strings $x,y \in \Sigma^*$ are \emph{consistent with respect to $T$} if and only if there does not exist any $u(v)^{\omega} \in E$ such that the examples $xu(v)^{\omega}$ and $yu(v)^{\omega}$ are oppositely labeled in $T$.

Starting with the empty string $\varepsilon$, the algorithm builds up a prefix-closed set $S$ of finite strings as follows. Initially, $S_1 = \{\varepsilon\}$.  After $S_k$ has been constructed, the algorithm considers each $s \in S_k$ in shortlex order, and each symbol $\sigma \in \Sigma$ in the ordering defined on $\Sigma$. If there exists no $s' \in S_k$ such that $s\sigma$ is consistent with $s'$ with respect to $T$, then $S_{k+1}$ is set to $S_k \cup \{s\sigma\}$ and $k$ is set to $k+1$.  If no such pair $s$ and $\sigma$ is found, then the final set $S$ is $S_k$.

In the second phase, the algorithm uses the strings in $S$ as names for states and constructs a transition function $\delta$ using $S$ and $E$.  For each $s \in S$ and $\sigma \in \Sigma$, there is at least one $s' \in S$ such that $s\sigma$ and $s'$ are consistent with respect to $T$. The algorithm selects any such $s'$ and defines $\delta(s,\sigma) = s'$.
Once $S$ and $\delta$ are defined, the algorithm returns the automaton $\aut{M} = \la \Sigma, S, \varepsilon, \delta \ra$.

\begin{proof}[Proof of \Cref{theorem:L-Aut-works}]
	$E$ may be computed in time polynomial in the length of $T$, by \Cref{prop:computing-suffixes}. 
	Because the default acceptor for $T$ has a polynomial number of states and is consistent with $T$, the number of distinguishable states, and the number of strings added to $S$, is bounded by a polynomial in the length of $T$.
	The returned automaton $\aut{M}$ is deterministic and complete by construction.
	
	Assume the sample $T$ is consistent with $\aut{A}$ and subsumes $T_{Aut}$.
	For any pair of states of $\aut{A}$, the set $E$ includes an experiment to distinguish them.
	Also, if $x$ and $y$ reach the same state of $\aut{A}$, there is no experiment in $E$ that distinguishes them.
	Then the set $S$ is precisely the access strings of $\aut{A}$.
	The choice of $s'$ for $\delta(s,\sigma)$ is unique in each case, and the returned automaton $\aut{M}$ is isomorphic to the automaton of $\aut{A}$. 
\end{proof}

Although the processes of constructing $T_{Aut}$ and learning an automaton from it are the same for acceptors of types IBA, ICA, IPA, IRA, ISA, or IMA, different types of acceptance condition require different kinds of characteristic samples and learning algorithms.

In the following sections we describe for each type of acceptor the corresponding sample $T_{Acc}$ and learning algorithm.
Each learning algorithm takes as input an automaton $\aut{M}$ and a sample $T$ and returns in polynomial time an acceptor of the appropriate type consistent with $T$.
We show that for each type of acceptor $\aut{A}$, if the input automaton $\aut{M}$ is isomorphic to the automaton of $\aut{A}$ and the sample $T$ is consistent with $\aut{A}$ and subsumes the $T_{Acc}$ for $\aut{A}$, then the learning algorithm returns an acceptor that is equivalent to $\aut{A}$.
This learning algorithm is then combined with $\alg{L}_{Aut}$ to prove the relevant class of languages are identifiable in the limit using polynomial time and data.

\section{The samples \texorpdfstring{$T_{Acc}$}{T\_Acc} and learning algorithms for IMA and IBA}	
\label{sec:T-Acc-for-IMAs-IBAs-ICAs}

The straightforward cases of Muller, \buchi\ and \cobuchi\ acceptance conditions are covered in this section.
Subsequent sections cover the cases of parity, Rabin and Street acceptance conditions, which are somewhat more involved.

\subsection{Muller acceptors}
\label{ssec:T-Acc-IMA}

Let $\aut{A}$ be an IMA with acceptance condition ${\alpha} = \{F_1,\ldots,F_k\}$.  By \Cref{prop:inf-to-SCC}, we may assume that each $F_i$ is a reachable SCC of $\aut{A}$.
The sample $T_{Acc}^{\mbox{\scriptsize{IMA}}}$ consists of $k$ positive examples, one for each set $F_i$.
The example for $F_i$ is $(u(v)^{\omega},1)$ where $\inf(u(v)^{\omega}) = F_i$.
These examples may be found in polynomial time in the size of $\aut{A}$ by \Cref{prop:SCC-to-inf}.

The learning algorithm $\alg{L}_{Acc}^{{\mbox{\scriptsize{IMA}}}}$ takes as input a deterministic complete automaton $\aut{M}$ and a sample $T$.
It constructs an acceptance condition $\alpha'$ as follows. For each positive labeled example $(u(v)^{\omega},1) \in T$, it computes the set $C = \infss{\aut{M}}(u(v)^{\omega})$ and makes $C$ a member of $\alpha'$.  Once the set $\alpha'$ is complete, the algorithm checks whether the DMA $(\aut{M},\alpha')$ is consistent with $T$.  If so, it returns $(\aut{M},\alpha')$; if not, it returns the default acceptor of type DMA for $T$.

\begin{theorem}
	\label{theorem:L-Acc-IMA-works}
	Algorithm $\alg{L}_{Acc}^{{\mbox{\scriptsize{IMA}}}}$ runs in time polynomial in the sizes of the inputs $\aut{M}$ and $T$.  Let $\aut{A}$ be an IMA. If the input automaton $\aut{M}$ is isomorphic to the automaton of $\aut{A}$, and the sample $T$ is consistent with $\aut{A}$ and subsumes $T_{Acc}^{\mbox{\scriptsize{IMA}}}$, then algorithm $\alg{L}_{Acc}^{{\mbox{\scriptsize{IMA}}}}$ returns an IMA $(\aut{M},\alpha')$ equivalent to $\aut{A}$.
\end{theorem}

\begin{proof}
	The construction of $\alpha'$ can be done in time polynomial in the sizes of $\aut{M}$ and $T$ by \Cref{prop:inf-to-SCC}.  The returned acceptor is consistent with $T$ by construction.
	
	Assume $\aut{M}$ is isomorphic to the automaton of $\aut{A}$ and that $T$ is consistent with $\aut{A}$.  For ease of notation, assume the isomorphism is the identity. Then for each positive example $(u(v)^{\omega},1)$ in $T$, the set $F = \inf(u(v)^{\omega})$ must be in $\alpha$, so $\alpha'$ is a subset of $\alpha$.
	
	If $T$ subsumes $T_{Acc}^{\mbox{\scriptsize{IMA}}}$, then for every set $F \in \alpha$ there is a positive example $(u(v)^{\omega},1)$ in $T$ with $F = \inf(u(v)^{\omega})$.  Thus the set $F$ is added to $\alpha'$, and $\alpha$ is a subset of $\alpha'$.  Thus, $(\aut{M},\alpha')$ is equivalent to $\aut{A}$, and because $T$ is consistent with $\aut{A}$, the IMA $(\aut{M},\alpha')$ is returned by $\alg{L}_{Acc}^{{\mbox{\scriptsize{IMA}}}}$.
\end{proof}

\begin{theorem}\label{thm:im:itptd}
	The class $\IM$ is identifiable in the limit using polynomial time and data.
\end{theorem}

\begin{proof}
	Let $\aut{A}$ be an IMA accepting a language $L$.
	The characteristic sample $T_L = T_{Aut} \cup T_{Acc}^{\mbox{\scriptsize{IMA}}}$ is of size polynomial in the size of $\aut{A}$.
	
	The combined learner $\alg{L}^{{\mbox{\scriptsize{IMA}}}}$ takes a sample $T$ as input and runs $\alg{L}_{Aut}$ on $T$ to produce an automaton $\aut{M}$ and then runs $\alg{L}_{Acc}^{{\mbox{\scriptsize{IMA}}}}$ on $\aut{M}$ and $T$ and returns the resulting acceptor.  It runs in polynomial time in the size of $T$ because it is the composition of two polynomial time algorithms, and the acceptor it returns is guaranteed to be consistent with $T$.
	
	If the sample $T$ is consistent with $\aut{A}$ and subsumes $T_L$, then by \Cref{theorem:L-Aut-works} the automaton $\aut{M}$ returned by $\alg{L}_{Aut}$ is isomorphic to the automaton of $\aut{A}$.  Then by \Cref{theorem:L-Acc-IMA-works} the acceptor returned by $\alg{L}_{Acc}^{{\mbox{\scriptsize{IMA}}}}$ with inputs $\aut{M}$ and $T$ is an IMA equivalent to $\aut{A}$.
\end{proof}

\subsection{\buchi\  acceptors}
\label{ssec:T-Acc-IBA}

The case of \buchi\ acceptors is nearly as straightforward as that of Muller acceptors.
Let $\aut{A}$ be an IBA with $n$ states and acceptance condition $F$.
For every state $q$ of $\aut{A}$, if there is an $\omega$-word $w$ such that $\aut{A}$ rejects $w$ and $q \in \inf(w)$, then there is an example $u(v)^{\omega}$ of length $O(n^2)$ such that $\aut{A}$ rejects $u(v)^{\omega}$ and $q \in \inf(u(v)^{\omega})$, by \Cref{prop:SCC-to-inf}. The negative labeled example $(u(v)^{\omega},0)$ is included in $T_{Acc}^{\mbox{\scriptsize{IBA}}}$.

The learning algorithm $\alg{L}_{Acc}^{\mbox{\scriptsize{IBA}}}$ takes as input a deterministic complete automaton $\aut{M}$ and a sample $T$.
The acceptance condition $F'$ consists of all the states $q$ of $\aut{M}$ such that for no negative example $(u(v)^{\omega},0)$ in $T$ do we have $q \in \inf_{\aut{M}}(u(v)^{\omega})$.
Once $F'$ has been computed, the algorithm checks whether the DBA $(\aut{M},F')$ is consistent with the sample $T$.  If so, it returns $(\aut{M},F')$; if not, it returns the default acceptor of type DBA for $T$.

\begin{theorem}
	\label{theorem:L-Acc-IBA-works}
	Algorithm $\alg{L}_{Acc}^{\mbox{\scriptsize{IBA}}}$ runs in time polynomial in the sizes of the inputs $\aut{M}$ and $T$.  Let $\aut{A}$ be an IBA. If the input automaton $\aut{M}$ is isomorphic to the automaton of $\aut{A}$, and the sample $T$ is consistent with $\aut{A}$ and subsumes $T_{Acc}^{\mbox{\scriptsize{IBA}}}$, then algorithm $\alg{L}_{Acc}^{\mbox{\scriptsize{IBA}}}$ returns an IBA $(\aut{M},F')$ equivalent to $\aut{A}$.
\end{theorem}

\begin{proof}
	The construction of $F'$ can be done in time polynomial in the sizes of $\aut{M}$ and $T$ by \Cref{prop:inf-to-SCC}.  The returned acceptor is consistent with $T$ by construction.
	
	Assume the input $\aut{M}$ is isomorphic to the automaton of $\aut{A}$, and that $T$ is consistent with $\aut{A}$ and subsumes $T_{Acc}^{\mbox{\scriptsize{IBA}}}$.  For ease of notation, assume the isomorphism is the identity. We show that the DBA $(\aut{M},F')$ is equivalent to $\aut{A}$.
	
	If $\aut{A}$ rejects the word $u(v)^{\omega}$ then let $C = \infss{\aut{M}}(u(v)^{\omega})$. Because $T$ subsumes $T_{Acc}^{\mbox{\scriptsize{IBA}}}$, for each $q \in C$, there is a negative example $(u'(v')^{\omega},0)$ in $T$ such that $q \in \infss{\aut{M}}(u'(v')^{\omega})$.  Thus no $q \in C$ is in $F'$ and $(\aut{M},F')$ also rejects $u(v)^{\omega}$.
	
	Conversely, if $\aut{A}$ accepts the word $u(v)^{\omega}$, then there is at least one state $q \in F$ such that $q \in \inf_{\aut{M}}(u(v)^{\omega})$.  Because $T$ is consistent with $\aut{A}$, there is no negative example $(u(v)^{\omega},0)$ in $T$ such that $q \in \infss{\aut{M}}(u(v)^{\omega})$, so $q \in F'$ and $(\aut{M},F')$ also accepts $u(v)^{\omega}$.  Thus $(\aut{M},F')$ is equivalent to $\aut{A}$. Because $T$ is consistent with $\aut{A}$, the IBA $(\aut{M},F')$ is returned by $\alg{L}_{Acc}^{\mbox{\scriptsize{IBA}}}$.
\end{proof}

\begin{theorem}
	\label{theorem:limit-id-of-DBA-DCA}
	The classes $\IB$ and $\IC$ are identifiable in the limit using polynomial time and data.
\end{theorem}

\begin{proof}
	The result for $\IC$ follows from that for $\IB$ by \Cref{prop:ib-ic-ir-is-duality}.  
	Let $\aut{A}$ be an IBA accepting language $L$. The characteristic sample $T_L = T_{Aut} \cup T_{Acc}^{\mbox{\scriptsize{IBA}}}$ is of size polynomial in  $\size(\aut{A})$.  
	
	The combined learning algorithm $\alg{L}^{\mbox{\scriptsize{IBA}}}$ takes a sample $T$ as input and runs $\alg{L}_{Aut}$ to get a deterministic complete automaton $\aut{M}$.  It then runs $\alg{L}_{Acc}^{\mbox{\scriptsize{IBA}}}$ on inputs $\aut{M}$ and $T$, and returns the resulting acceptor.  $\alg{L}^{\mbox{\scriptsize{IBA}}}$ runs in polynomial time in the length of $T$ and returns a DBA consistent with $T$.
	
	If the sample $T$ is consistent with $\aut{A}$ and subsumes $T_L$ then $\alg{L}_{Aut}$ returns an automaton $\aut{M}$ isomorphic to the automaton of $\aut{A}$ by \Cref{theorem:L-Aut-works}.  Then the acceptor returned by $\alg{L}_{Acc}^{\mbox{\scriptsize{IBA}}}$ on inputs $\aut{M}$ and $T$ is an IBA equivalent to $\aut{A}$ by \Cref{theorem:L-Acc-IBA-works}.
\end{proof}

\section{The sample \texorpdfstring{$T_{Acc}$}{T\_Acc} and the learning algorithm for IPA}
\label{sec:T-Acc-for-DPAs}

The construction of $T_{Acc}^{\mbox{\scriptsize{IPA}}}$ for an IPA $\aut{P}$ builds on the construction of the canonical forest of SCCs for $\aut{P}$, whose construction and properties are described next. 
 Roughly speaking, the purpose of the canonical forest for a given parity automaton $\aut{P}$ is to expose a set of words that if placed in the sample will lead a smart learner to correctly determine a coloring function for the constructed automaton. It is thus not surprising, that while developed for a different motivation, it has similarities with Carton and Maceiras's algorithm to compute the minimal number of colors for a given parity automaton~\cite{CartonM99}.
 We begin with the definition and properties of a decreasing forest of SCCs of an $\omega$-automaton.

\subsection{A decreasing forest of SCCs of an automaton}
\label{ssec:decreasing-forest}

Let $\aut{M}$ be a deterministic automaton, and let $S$ be a subset of its states.
A \emph{decreasing forest of SCCs of $\aut{M}$ rooted in $S$} is a finite rooted forest $\forest{F}$ in which every node $C$ is an SCC of $\aut{M}$ that is contained in $S$, and the following properties are satisfied. 
\begin{enumerate}
	\item The roots of $\forest{F}$ are the elements of $\maxSCCs(S)$.
	\item Whenever $D_1, \ldots, D_k$ are the children of node $C$, we have $D_1 \cup \ldots \cup D_k \subsetneq C$. Also, letting $\Delta(C) = C \setminus (D_1 \cup \ldots \cup D_k)$, the children $D_1, \ldots, D_k$ are exactly the elements of $\maxSCCs(C \setminus \Delta(C))$.
\end{enumerate}

\begin{proposition}
	\label{prop:monotone-forest-properties}
	Let $\aut{M}$ be a deterministic automaton, $S$ a subset of its states, and $\forest{F}$ a decreasing forest of SCCs of $\aut{M}$ rooted in $S$.  Then the following are true.
	\begin{enumerate}
		\item The roots of $\forest{F}$ are pairwise disjoint.
		\item The children of any node are pairwise disjoint.
		\item $\forest{F}$ has at most $|S|$ nodes.
		\item For any $D \subseteq S$ that is an SCC of $\aut{M}$, there is a unique node $C$ in $\forest{F}$ such that $D \subseteq C$ and $D$ is not a subset of any of the children of $C$.
	\end{enumerate}
\end{proposition}

\begin{proof}
	The roots of $\forest{F}$ are the elements of $\maxSCCs(S)$, which are pairwise disjoint.  The children of a node $C$ are the elements of $\maxSCCs(C \setminus \Delta(C))$, which are pairwise disjoint.  The sets $\Delta(C)$ for nodes $C$ in $\forest{F}$ are contained in $S$, nonempty, and pairwise disjoint, so the number of nodes is at most $|S|$.  If $D \subseteq S$ is an SCC of $\aut{M}$, then $D$ is a subset of exactly one of the roots of $\forest{F}$, say $C_1$.  If $D \cap \Delta(C_1) \neq \emptyset$, then $D$ is not a subset of any of the children of $C_1$.  Otherwise, $D$ must be a subset of exactly one of the children of $C_1$, say $C_2$.  If $D \cap \Delta(C_2) \neq \emptyset$, then $D$ is not a subset of any of the children of $C_2$.  Continuing in this way, we eventually arrive at the required node $C$.
\end{proof}

Given a decreasing forest $\forest{F}$ of SCCs of automaton $\aut{M}$ rooted in $S$, and an SCC $D \subseteq S$, we denote by $\alg{Node}(D,\forest{F})$ the unique node $C$ of $\forest{F}$ such that $D \subseteq C$ and $D$ is not a subset of any of the children of $C$.  We note that if $C = \alg{Node}(D,\forest{F})$ then $D \cap \Delta(C) \neq \emptyset$.
If $D$ is a child of some $C$ in $\forest{F}$, we define \emph{merging $D$ into $C$} as the operation of removing $D$ from $\forest{F}$ and making the children of $D$ (if any) direct children of $C$.

\begin{proposition}
	\label{prop:merges-in-a-decreasing-forest}
	Let $\aut{M}$ be a deterministic automaton and $S$ a subset of its states.  Let $\forest{F}$ be a decreasing forest of SCCs of $\aut{M}$ rooted in $S$.  Let $D$ be a child of $C$ in $\forest{F}$ and let \forest{F}' be obtained from $\forest{F}$ by merging $D$ into $C$.  Then \forest{F}' is also a decreasing forest of SCCS of $\aut{M}$ rooted in $S$.
\end{proposition}

\begin{proof}
	After the merge, the roots of \forest{F} remain the elements of $\maxSCCs(S)$.
	Let $D_1, \ldots, D_k$ be the children of $C$ in $F$, where $D = D_k$, and let $E_1, \ldots, E_{\ell}$ be the children of $D$ in $F$. Because the union of $E_1, \ldots, E_{\ell}$ is a proper subset of $D = D_k$ and the union of $D_1,\ldots, D_k$ is a proper subset of $C$, the union of $D_1, \ldots, D_{k-1}$ with the union of $E_1, \ldots, E_{\ell}$ is a proper subset of $C$, therefore the union of the children of $C$ in \forest{F} is a proper subset of $C$.
	Also, the children of $C$ in \forest{F} are the maximum SCCs of $C \setminus \Delta_{\forest{F}'}(C)$, and no other nodes are affected, so \forest{F} is a decreasing forest of SCCs of $\aut{M}$ rooted in $S$.
\end{proof}

\subsection{Constructing the Canonical Forest and Coloring of a DPA}

Let $\aut{P} =\la \Sigma, Q, q_\iota, \delta, \linebreak[5]\kappa \ra$ be a complete DPA.  We extend the coloring function $\kappa$ to nonempty sets of states by $\kappa(S) = \min\{k(q) \mid q \in S\}$, the minimum color of any state in $S$.  We define the \emph{$\kappa$-parity} of $S$ to be $1$ if $\kappa(S)$ is odd, and $0$ if $\kappa(S)$ is even. A word $w \in \Sigma^{\omega}$ is accepted by $\aut{P}$ iff the $\kappa$-parity of $\inf(w)$ is $1$. Note that the union of two sets of $\kappa$-parity $b$ is also of $\kappa$-parity $b$.  
For any nonempty $S \subseteq Q$, we define $\minStates(S) = \{q \in S \mid \kappa(q) = \kappa(S)\}$, the states of $S$ that are assigned the minimum color among all states of $S$.

\subsubsection{The $\minStates$-Forest.}  We describe an algorithm to construct the \emph{$\minStates$-forest of $\aut{P}$}.   
The roots of the $\minStates$-forest are the elements of $\maxSCCs(Q)$, each marked as unprocessed.
The following procedure is repeated until all elements are marked as processed.
If $C$ is unprocessed, then ${\mathcal D} = \maxSCCs(C \setminus \minStates(C))$ is computed.
If ${\mathcal  D}$ is empty, $C$ becomes a leaf in the forest, and is marked as processed.
Otherwise, $C$ is marked as processed and the elements of ${\mathcal  D}$ are made the children of $C$ and are marked as unprocessed.

\begin{proposition}
	\label{prop:minStates-forest}
	Let $\aut{P} = \la \Sigma, Q, q_\iota, \delta, \kappa \ra$ be a complete DPA with automaton $\aut{M}$. Let $\forest{F}$ be the $\minStates$-forest of $\aut{P}$. Then $\forest{F}$ is a decreasing forest of SCCs of $\aut{M}$ rooted in $Q$, and can be computed in polynomial time. For any SCC $D$, $\kappa(D) = \kappa(\alg{Node}(D,\forest{F}))$.
\end{proposition}

\begin{proof}
	Referring to the construction of the $\minStates$-forest $\forest{F}$, its roots are the elements of $\maxSCCs(Q)$.  When a node $C$ is processed, the nonempty set $\minStates(C)$ is removed and the maximum SCCs (if any) of the result become the children of $C$, so the union of the children of $C$ is a proper subset of $C$, and the children are the maximum SCCs of $C \setminus \Delta(C)$.
	Let $D$ be an SCC.  Then $D \subseteq Q$ and for the node $C = \alg{Node}(D,\forest{F})$, we have that $D \subseteq C$ and $D$ is not a subset of any child of $C$.  Thus $D \cap \minStates(C) \neq \emptyset$, because otherwise $D$ would be a subset of some child of $C$. This implies that $\kappa(D) = \kappa(C)$.
	The $\minStates$-forest of $\aut{P}$ can be computed in polynomial time because it has at most $|Q|$ nodes, and each set $\maxSCCs(S)$ can be computed in polynomial time by \Cref{prop:maxSCCs-in-linear-time}.
\end{proof}

\subsubsection{The Canonical Forest and Coloring.}
The canonical forest of $\aut{P}$ is constructed as follows, starting with the $\minStates$-forest of $\aut{P}$.
While there exist in the forest a node $D$ and its parent $C$ of the same $\kappa$-parity, one such pair $D$ and $C$ is selected, and the child node $D$ is merged into the parent node $C$.
When no such pair remains, the result is the \emph{canonical forest of $\aut{P}$}, denoted $\forest{F}^*(\aut{P})$.
The canonical forest $\forest{F}^*(\aut{P})$ can be computed from $\aut{P}$ in polynomial time.

From the canonical forest $\forest{F}^*(\aut{P})$, we define the \emph{canonical coloring} $\kappa^*$. The states in $(Q \setminus \bigcup\maxSCCs(Q))$ are not contained in any SCC of $\aut{P}$ and do not affect the acceptance or rejection of any $\omega$-word.  For definiteness, we assign them $\kappa^*(q) = 0$. For a root node $C$ of $\kappa$-parity $b$, we define $\kappa^*(q) = b$ for all $q \in \Delta(C)$.  Let $C$ be an arbitrary node of $\forest{F}^*(\aut{P})$.  If the states of $\Delta(C)$ have been assigned color $k$ by $\kappa^*$ and $D$ is a child of $C$, then the states of $\Delta(D)$ are assigned color $k+1$ by $\kappa^*$.
Clearly $\kappa^*$ can be computed from $\aut{P}$ in polynomial time.

\begin{example}
    \Cref{fig:example-DPA}~(a) shows the graph $G(\aut{P})$ of a DPA $\aut{P}$ with states $a$ through $m$, labeled by the colors assigned by $\kappa$.  \Cref{fig:example-minStates-forest}~(b) shows the $\minStates$-forest of $\aut{P}$, with the nodes labeled by their $\kappa$-parities. \Cref{fig:example-canonical-forest}~(c) shows the canonical forest $\forest{F}^*(\aut{P})$ of $\aut{P}$, with the nodes labeled by their $\kappa$-parities.   \Cref{fig:example-canonical-coloring}~(d) shows the graph $G(\aut{P})$ re-colored using the canonical coloring $\kappa^*$.
\end{example}

\begin{figure}[t]
{\small{
	\begin{center}
		\scalebox{0.65}{
			\begin{tikzpicture}[->,>=stealth',shorten >=1pt,auto,node distance=2.2cm,semithick,initial text=,initial where=left]
			\node[state] (node-a)  {$a : 1$};
			\node[state] (node-b)  [left  of=node-a]{$b : 3$};
			\node[state] (node-h)  [below of=node-a]{$h : 4$};
			\node[state] (node-i)  [below right of=node-a]{$i : 3$};
			\node[state] (node-c)  [below left of=node-b]{$c : 1$}; 
			\node[state] (node-d)  [below of=node-h]{$d : 5$};
			\node[state] (node-e) [below right of=node-c]{$e : 4$};
			\node[state] (node-f) [below of=node-c]{$f : 4$};
			\node[state] (node-g) [below right of=node-f] {$g : 5 $};
			\node[state] (node-j) [below of=node-i] {$j : 4 $};
			\node[state] (node-k) [below right of=node-i] {$k : 5$};
			\node[state] (node-l) [below of=node-j] {$l : 2$};
			\node[state] (node-m) [below right of=node-j] {$m : 4$};
			\node[label] (label-a) [below of=node-d, scale=1.5, node distance=2.1cm] {(a)};				
			
			\path (node-a) edge [bend right]  node [left] {} (node-b); 
			\path (node-a) edge [bend right]  node {} (node-h);
			\path (node-a) edge [bend left] node{} (node-i);
			\path (node-b) edge node [right] {} (node-a); 
			\path (node-b) edge node {} (node-c);
			\path (node-b) edge node {} (node-e);
			\path (node-h) edge node [right]{} (node-a);
			\path (node-d) edge node {} (node-b);
			\path (node-e) edge node {} (node-d);
			\path (node-d) edge [bend left] node {} (node-e);
			\path (node-c) edge node {} (node-e);
			\path (node-c) edge [bend right] node {} (node-f);
			\path (node-f) edge [bend right] node {} (node-c);
			\path (node-f) edge [bend right] node {} (node-g);
			\path (node-g) edge [bend right] node {} (node-f);
			\path (node-i) edge node {} (node-j);
			\path (node-j) edge [bend right] node {} (node-k);
			\path (node-k) edge [bend right] node {} (node-j);
			\path (node-k) edge node {} (node-i);
			\path (node-j) edge node {} (node-l);
			\path (node-l) edge [bend right] node {} (node-m);
			\path (node-m) edge [bend right] node {} (node-l);
			\path (node-m) edge [loop right] (node-m);
			
			\node[setofstates] (C0) [right of=node-a, node distance=9cm]{$a,b,c,d,e,f,g,h :1 $};
			\node[setofstates] (C01) [below of=C0,node distance=2.75cm] {$f,g : 0$};
			\node[setofstates] (C00) [left of=C01,node distance=2.75cm] {$b,d,e : 1$};
			\node[setofstates] (C000) [below of=C00,node distance=2.75cm] {$d,e : 0$};
			\node[setofstates] (C1) [right of=C0,node distance=2.95cm] {$i,j,k : 1$};
			\node[setofstates] (C10) [below of=C1,node distance=2.75cm] {$j,k : 0$};
			\node[setofstates] (C2) [right of=C1,node distance=2cm] {$l,m : 0$};
			\node[setofstates] (C20) [below of=C2,node distance=2.75cm] {$m : 0$};
			\node[label] (label-b) [right of=label-a, scale=1.5, node distance=10.5cm] {(b)};								
			
			\path (C0) edge (C00);
			\path (C0) edge (C01);
			\path (C00) edge (C000);
			\path (C1) edge (C10);
			\path (C2) edge (C20);

			\node[setofstates] (D0) [below of=node-b, node distance=8.5cm] {$a,b,c,d,e,f,g,h : 1$};
			\node[setofstates] (D01) [below of=D0,node distance=2.75cm] {$f,g : 0$};
			\node[setofstates] (D000) [left of=D01,node distance=2.75cm] {$d,e : 0$};
			\node[setofstates] (D1) [right of=D0,node distance=2.95cm] {$i,j,k : 1$};
			\node[setofstates] (D10) [below of=D1,node distance=2.75cm] {$j,k : 0$};
			\node[setofstates] (D2) [right of=D1,node distance=2.0cm] {$l,m : 0$};
			
			\path (D0) edge (D000);
			\path (D0) edge (D01);
			\path (D1) edge (D10);

			\node[state] (node--a) [right of=D0,node distance=12.5cm]  {$a : 1$};
			\node[state] (node--b)  [left  of=node--a]{$b : 1$};
			\node[state] (node--h)  [below of=node--a]{$h : 1$};
			\node[state] (node--i)  [below right of=node--a]{$i : 1$};
			\node[state] (node--c)  [below left of=node--b]{$c : 1$}; 
			\node[state] (node--d)  [below of=node--h]{$d : 2$};
			\node[state] (node--e) [below right of=node--c]{$e : 2$};
			\node[state] (node--f) [below of=node--c]{$f : 2$};
			\node[state] (node--g) [below right of=node--f] {$g : 2$};
			\node[state] (node--j) [below of=node--i] {$j : 2$};
			\node[state] (node--k) [below right of=node--i] {$k : 2$};
			\node[state] (node--l) [below of=node--j] {$l : 0$};
			\node[state] (node--m) [below right of=node--j] {$m : 0$};
			\node[label] (label-d) [below of=node--d, scale=1.5, node distance=2.1cm] {(d)};		
			\node[label] (label-c) [left of=label-d, scale=1.5, node distance=10.5cm] {(c)};																					
			
			\path (node--a) edge [bend right]  node [left] {} (node--b); 
			\path (node--a) edge [bend right]  node {} (node--h);
			\path (node--a) edge [bend left] node{} (node--i);
			\path (node--b) edge node [right] {} (node--a); 
			\path (node--b) edge node {} (node--c);
			\path (node--b) edge node {} (node--e);
			\path (node--h) edge node [right]{} (node--a);
			\path (node--d) edge node {} (node--b);
			\path (node--e) edge node {} (node--d);
			\path (node--d) edge [bend left] node {} (node--e);
			\path (node--c) edge node {} (node--e);
			\path (node--c) edge [bend right] node {} (node--f);
			\path (node--f) edge [bend right] node {} (node--c);
			\path (node--f) edge [bend right] node {} (node--g);
			\path (node--g) edge [bend right] node {} (node--f);
			\path (node--i) edge node {} (node--j);
			\path (node--j) edge [bend right] node {} (node--k);
			\path (node--k) edge [bend right] node {} (node--j);
			\path (node--k) edge node {} (node--i);
			\path (node--j) edge node {} (node--l);
			\path (node--l) edge [bend right] node {} (node--m);
			\path (node--m) edge [bend right] node {} (node--l);	
			\path (node--m) edge [loop right] (node--m);					
			
			\end{tikzpicture}
		}
		\caption{
			(a) Graph $G(\aut{P})$ with states colored by $\kappa$. (b) The $\minStates$-forest of $\aut{P}$, with $\kappa$-parities of nodes. 
			(c) Canonical forest $\forest{F}^*(\aut{P})$, with $\kappa$-parities of nodes.  (d) Graph $G(\aut{P})$ with the canonical coloring $\kappa^*$.}
		\label{fig:example-DPA}  
		\label{fig:example-minStates-forest}
		\label{fig:example-canonical-forest}
		\label{fig:example-canonical-coloring}  
	\end{center}
 }}
\end{figure}	

\begin{theorem}
	\label{theorem:canonical-forest-properties}
	Let $\aut{P} =\la \Sigma, Q, q_\iota, \delta, \kappa \ra$ be a complete DPA with automaton $\aut{M}$.  The canonical forest $\forest{F}^*(\aut{P})$ is a decreasing forest of SCCs of $\aut{M}$ rooted in $Q$ and has the following properties.
	\begin{enumerate}
		\item For any SCC $D$, both $D$ and $\alg{Node}(D,\forest{F}^*(\aut{P}))$ have the same $\kappa$-parity.
		\item For every node $C$ of $\forest{F}^*(\aut{P})$, the $\kappa$-parity of $C$ is the same as the $\kappa^*$-parity of $C$.
		\item The children in $\forest{F}^*(\aut{P})$ of a node $C$ of $\kappa$-parity $b$ are the maximal SCCs $D \subseteq C$ of $\kappa$-parity $1-b$.
	\end{enumerate}
\end{theorem}

\begin{proof}
	Because $\forest{F}^*(\aut{P})$ is obtained from the $\minStates$-forest of $\aut{P}$ by a sequence of merges, $\forest{F}^*(\aut{P})$ is a decreasing forest of SCCs of $\aut{M}$ rooted in $Q$ by  \Cref{prop:merges-in-a-decreasing-forest}.
	Let $\forest{F}_0$ denote the $\minStates$-forest of $\aut{P}$, and let $\forest{F}_i$ denote the forest after $i$ merges have been performed in the computation to produce the canonical acceptor $\forest{F}^*(\aut{P})$.
	
	By \Cref{prop:minStates-forest}, for any SCC $D$, $\kappa(D) = \kappa(\alg{Node}(D,\forest{F}_0))$, so property (1) holds for $\forest{F}_0$.  We show by induction that it holds for each $\forest{F}_i$ and therefore for $\forest{F}^*(\aut{P})$. Assume that property (1) holds of $\forest{F}_i$, and $\forest{F}_{i+1}$ is obtained from $\forest{F}_i$ by merging child node $D$ into parent node $C$.  Let $D'$ be any SCC. If $\alg{Node}(D',\forest{F}_i) = \alg{Node}(D',\forest{F}_{i+1})$, then because property (1) holds in $\forest{F}_i$, we have that the $\kappa$-parity of $D'$ is the same as the $\kappa$-parity of $\alg{Node}(D',\forest{F}_{i+1})$.
	Otherwise, it must be that $D = \alg{Node}(D',\forest{F}_i)$ and $C = \alg{Node}(D',\forest{F}_{i+1})$.  Because $D$ is only merged to $C$ if they are of the same $\kappa$-parity, this implies that the $\kappa$-parity of $D'$ is the same as the $\kappa$-parity of $C$.  Thus, property (1) holds also in $\forest{F}_{i+1}$.
	
	For property (2), we note that the $\kappa$-parity and the $\kappa^*$-parity of each root node of $\forest{F}^*(\aut{P})$ is the same.  Suppose $C$ is a node of $\forest{F}^*(\aut{P})$ whose $\kappa$-parity and $\kappa^*$-parity are equal to $b$, and $D$ is a child of $C$.  Then by the construction of $\forest{F}^*(\aut{P})$, the $\kappa$-parity of $D$ is $1-b$.  And by the definition of $\kappa^*$, the $\kappa^*$ parity of the elements of $\Delta(D)$ is the opposite of the $\kappa^*$-parity of the elements of $\Delta(C)$. Because the $\kappa^*$-parity of $C$ is $b$, the $\kappa^*$-parity of $D$ is also $1-b$.
	
	For property (3), let $D$ be any maximal SCC of $\aut{P}$ that is contained in $C$ and has $\kappa$-parity $1-b$. Then $\alg{Node}(D,\forest{F}_0)$ is contained in $C$ and has the same $\kappa$-parity as $D$.  Since $D$ is maximal, we must have $D = \alg{Node}(D,\forest{F}_0)$, and any nodes on the path between $C$ and $D$ must have $\kappa$-parity $b$ and must be merged into $C$ to form $\forest{F}^*(\aut{P})$.  Thus, $D$ is a child of $C$ in $\forest{F}^*(\aut{P})$. 
	
	Conversely, if $D$ is a child of $C$ in $\forest{F}^*(\aut{P})$ then $D \subseteq C$ and the $\kappa$-parity of $D$ is $1-b$.  Assume $D'$ is an SCC such that $D' \subseteq C$, the $\kappa$-parity of $D'$ is $1-b$ and $D \subsetneq D'$.  Then $D'' = \alg{Node}(D',\forest{F}_0)$ is a descendant of $C$ in $\forest{F}_0$ that has $\kappa$-parity $1-b$, and $D \subsetneq D''$, so because there is another node of parity $1-b$ on the path between $D$ and  $C$ in $\forest{F}_0$, $D$ cannot be a child of $C$ in $\forest{F}^*(\aut{P})$, a contradiction.
\end{proof}	

Replacing the coloring function of $\aut{P}$ by the canonical coloring does not change the $\omega$-language accepted.

\begin{theorem}
	\label{theorem:canonical-coloring-works}
	Let $\aut{P} = \la \Sigma, Q, q_\iota, \delta, \kappa \ra$ be a complete DPA, and $\aut{P}^*$ be $\aut{P}$ with the canonical coloring $\kappa^*$ for $\aut{P}$ in place of $\kappa$.  Then $\aut{P}$ and $\aut{P}^*$ recognize the same $\omega$-language.
\end{theorem}
\begin{proof}
	Let $w$ be an $\omega$-word and let $D = \inf(w)$. This is an SCC of the (common) automaton of $\aut{P}$ and $\aut{P}^*$.  Let $C = \alg{Node}(D,\forest{F}^*(\aut{P}))$.  Then $D \cap \Delta(C) \neq \emptyset$, and $\kappa^*(D) = \kappa^*(C)$, by the definition of $\kappa^*$,
	The $\kappa^*$ parity of $C$ is the same as the $\kappa$-parity of $C$, by property (2) of \Cref{theorem:canonical-forest-properties}.
	The $\kappa$-parity of $C$ is the same as the $\kappa$-parity of $D$, by property (1) of \Cref{theorem:canonical-forest-properties}.
	Thus, the $\kappa^*$-parity of $D$ is the same as the $\kappa$-parity of $D$, and $w \in \sema{\aut{P}}$ iff $w \in \sema{\aut{P}^*}$.
\end{proof}

\subsection{Constructing \texorpdfstring{$T_{Acc}^{\mbox{\scriptsize{IPA}}}$}{T\_Acc\textasciicircum IPA}}
\label{ssec:T-Acc-IPA}

We now describe the construction of $T_{Acc}^{\mbox{\scriptsize{IPA}}}$, the second part of the characteristic sample for an IPA $\aut{P}$ with the automaton $\aut{M}$ of $n$ states. The sample $T_{Acc}^{\mbox{\scriptsize{IPA}}}$ consists of one example $u(v)^{\omega} = \alg{Witness}(C,\aut{M})$ of length $O(n^2)$ for each reachable SCC $C$ in the canonical forest $\forest{F}^*(\aut{P})$. The example $u(v)^{\omega}$ is labeled $1$ if it is accepted by $\aut{P}$ and $0$ otherwise.  Thus $T_{Acc}^{\mbox{\scriptsize{IPA}}}$ contains at most $n$ labeled examples, each of length $O(n^2)$.

\subsection{The learning algorithm \texorpdfstring{$\alg{L}_{Acc}^{\mbox{\scriptsize{IPA}}}$}{L\_Acc\textasciicircum IPA}}
\label{ssec:learn-alg-dpa}

Given a complete deterministic automaton $\aut{M} = \la \Sigma, Q, q_{\iota}, \delta \ra$ and a sample $T$ as input, the learning algorithm $\alg{L}_{Acc}^{\mbox{\scriptsize{IPA}}}$ attempts to construct a coloring of the states of $\aut{M}$ consistent with $T$.

The algorithm first constructs the set $Z$ of all $C \subseteq Q$ such that for some labeled example $(u(v)^{\omega},l)$ in $T$ we have $C = \infss{\aut{M}}(u(v)^{\omega})$. If two examples with different labels are found to yield the same set $C$, this is evidence that the automaton $\aut{M}$ is not correct, and the learning algorithm returns the default acceptor of type DPA for $T$.  

Otherwise, each set $C$ in $Z$ is associated with the label of the one or more examples that yield $C$.  The set $Z$ is partially ordered by the subset relation.  The learning algorithm then attempts to construct a rooted forest $\forest{F}'$ with nodes that are elements of $Z$, corresponding to the canonical forest of the target acceptor. Initially, $\forest{F}'$ contains as roots all the maximal elements of $Z$.  If these are not pairwise disjoint, it returns the default acceptor of type DPA for $T$.  Otherwise, the root nodes are all marked as unprocessed.

For each unprocessed node $C$ in $\forest{F}'$, it computes the set of all $D \in Z$ such that $D \subseteq C$, $D$ has the opposite label to $C$, and $D$ is maximal with these properties, and makes $D$ a child of $C$ and marks $D$ as unprocessed.  When all the children of a node $C$ have been determined, the algorithm checks two conditions: (1) that the children of $C$ are pairwise disjoint, and (2) there is at least one $q \in C$ that is not in any child of $C$.  If either of these conditions fail, then it returns the default acceptor of type DPA for $T$.  If both conditions are satisfied, then the node $C$ is marked as processed.  When there are no more unprocessed nodes, the construction of $\forest{F}'$ is complete. Note that $\forest{F}'$ has at most $|Q|$ nodes.

When the construction of $\forest{F}'$ is complete, for each node $C$ in $\forest{F}'$ let $\Delta(C)$ denote the elements of $C$ that do not appear in any of its children.  Then the learning algorithm assigns colors to the elements of $Q$ starting from the roots of $\forest{F}'$, as follows.  If $C$ is a root with label $l$, then $\kappa'(q) = l$ for all $q \in \Delta(C)$.  If the elements of $\Delta(C)$ have been assigned color $k$ and $D$ is a child of $C$, then $\kappa'(q) = k+1$ for all $s \in \Delta(D)$.  When this process is complete, any uncolored states $q$ are assigned $\kappa'(q) = 0$.  

If the resulting DPA $(\aut{M},\kappa')$ is consistent with the sample $T$, the algorithm $\alg{L}_{Acc}^{\mbox{\scriptsize{IPA}}}$ returns $(\aut{M},\kappa')$. If not, it returns the default acceptor of type DPA for $T$.

\begin{theorem}
	\label{theorem:L-Acc-IPA-works}
	Algorithm $\alg{L}_{Acc}^{\mbox{\scriptsize{IPA}}}$ runs in time polynomial in the sizes of the inputs $\aut{M}$ and $T$.  Let $\aut{P}$ be an IPA. If the input automaton $\aut{M}$ is isomorphic to the automaton of $\aut{P}$, and the sample $T$ is consistent with $\aut{P}$ and subsumes $T_{Acc}^{\mbox{\scriptsize{IPA}}}$, then algorithm $\alg{L}_{Acc}^{\mbox{\scriptsize{IPA}}}$ returns an IPA $(\aut{M},\kappa')$ equivalent to $\aut{P}$.
\end{theorem}
\begin{proof}
	The construction of $\kappa'$ can be done in time polynomial in the sizes of $\aut{M}$ and $T$.  The returned acceptor is consistent with $T$ by construction.
	
	Assume the input $\aut{M} = \la \Sigma, Q, q_{\iota}, \delta \ra$ is isomorphic to the automaton of $\aut{P}$, and that $T$ is consistent with $\aut{P}$ and subsumes $T_{Acc}^{\mbox{\scriptsize{IPA}}}$.  For ease of notation, assume the isomorphism is the identity. We show that the forest $\forest{F}'$ constructed by the learning algorithm is equal to the canonical forest of $\forest{F}^*(\aut{P})$, the coloring $\kappa'$ is equal to the canonical coloring $\kappa^*$, and therefore the acceptor $(\aut{M},\kappa')$ is equivalent to $\aut{P}$.
	
	The roots of $\forest{F}^*(\aut{P})$ are the maximal SCCs contained in $Q$, and for each such root $C$, $T_{Acc}^{\mbox{\scriptsize{IPA}}}$ contains an example $(u(v)^{\omega},l)$ such that $C = \inf(u(v)^{\omega})$.  Thus, the set of maximal elements of $Z$ is equal to the set of roots of $\forest{F}^*(\aut{P})$.
	
	Let $C$ be any node of $\forest{F}^*(\aut{P})$, and let $D$ be a child of $C$ in $\forest{F}^*(\aut{P})$. Then $D$ is an SCC, $D \subseteq C$, the parity of $D$ is opposite to the parity of $C$, and $D$ is maximal in the subset ordering with these properties, by property (3) of \Cref{theorem:canonical-forest-properties}. In the sample $T_{Acc}^{\mbox{\scriptsize{IPA}}}$ there is an example $(u(v)^{\omega},l)$ with $D = \inf(u(v)^{\omega})$, so $D$ is an element of $Z$, and will be made a child of $C$ in $\forest{F}'$ because $D \subseteq C$, the label $l$ is the opposite of the label of $C$, and $D$ is maximal in $Z$ with these properties.
	Conversely, if $D$ is made a child of $C$ in $\forest{F}'$, then $D \subseteq C$, the label of $D$ is opposite to the label of $C$ (that is, they are of opposite $\kappa$-parity), and $D$ is maximal in $Z$ with these properties. This implies $D$ is a child of $C$ in $\forest{F}^*(\aut{P})$, by property (3) of \Cref{theorem:canonical-forest-properties}.
	
	By induction, $\forest{F}'$ is equal to $\forest{F}^*(\aut{P})$, and therefore $\kappa'$ is equal to the canonical coloring $\kappa^*$.  Then the IPA $(\aut{M},\kappa')$ is equivalent to $\aut{P}$, by \Cref{theorem:canonical-coloring-works}.  Because $T$ is consistent with $\aut{P}$, the IPA $(\aut{M},\kappa')$ is returned by $\alg{L}_{Acc}^{\mbox{\scriptsize{IPA}}}$.
\end{proof}

\begin{theorem}
	\label{theorem:ip-poly-id-lim}
	The class $\IP$ is identifiable in the limit using polynomial time and data.  
\end{theorem}

\begin{proof}
	Let $\aut{P}$ be an IPA accepting the language $L$.  The characteristic sample $T_L = T_{Aut} \cup T_{Acc}^{\mbox{\scriptsize{IPA}}}$ for $\aut{P}$ is of size polynomial in the size of $\aut{P}$.
	
	The combined learning algorithm $\alg{L}^{\mbox{\scriptsize{IPA}}}$ with a sample $T$ as input first runs $\alg{L}_{Aut}$ on $T$ to get a complete deterministic automaton $\aut{M}$ and then runs $\alg{L}_{Acc}^{\mbox{\scriptsize{IPA}}}$ on inputs $\aut{M}$ and $T$ and returns the resulting acceptor.
	The running time of $\alg{L}^{\mbox{\scriptsize{IPA}}}$ is polynomial in the length of $T$ and the returned acceptor is consistent with $T$.
	
	Assume  the sample $T$ is consistent with $\aut{P}$ and subsumes $T_L$. By \Cref{theorem:L-Aut-works}, the automaton $\aut{M}$ is isomorphic to the automaton of $\aut{P}$.
	By \Cref{theorem:L-Acc-IPA-works}, the acceptor $(\aut{M},\kappa')$ is equivalent to $\aut{P}$. Because it is consistent with $T$, the IPA $(\aut{M},\kappa')$ is the acceptor returned by $\alg{L}^{\mbox{\scriptsize{IPA}}}$.
\end{proof}

\section{The sample \texorpdfstring{$T_{Acc}$}{T\_Acc} and learning algorithm for IRA}
\label{sec:T-Acc-for-IRAs-ISAs}

In this section we introduce some terminology, establish a normal form for Rabin acceptors, define an ordering on sets of states of an automaton, and then describe the learning algorithm $\alg{L}_{Acc}^{\mbox{\scriptsize{IRA}}}$ and sample $T_{Acc}^{\mbox{\scriptsize{IRA}}}$ for IRAs.  We then prove that the classes $\class{IRA}$ and $\class{ISA}$ are identifiable in the limit using polynomial time and data.

Let $\aut{R} = \la \Sigma, Q, q_i, \delta, \alpha \ra$ be a Rabin acceptor, where the acceptance condition $\alpha = \{(G_1,B_1),\ldots,(G_k,B_k)\}$ is a set of ordered pairs of states.
We say that an $\omega$-word $w$ \emph{satisfies a pair of state sets} $(G,B)$ iff $\inf(w) \cap G \neq \emptyset$ and $\inf(w) \cap B = \emptyset$.
Also, $w$ \emph{satisfies the acceptance condition} $\alpha$ iff there exists $i \in [1..k]$ such that $w$ satisfies $(G_i,B_i)$.
Then an $\omega$-word $w$ is accepted by $\aut{R}$ iff $w$ satisfies the acceptance condition $\alpha$ of $\aut{R}$.

\subsection{Singleton normal form for a Rabin acceptor}
\label{ssec:singleton-normal-form}

We say that a Rabin acceptor $\aut{R}$ is in \emph{singleton normal form} iff for every pair $(G_i,B_i)$ in its acceptance condition we have $|G_i| = 1$, that is, every $G_i$ is a singleton set.
To avoid extra braces, we abbreviate the pair $(\{q\},B)$ by $(q,B)$.
Every Rabin acceptor may be put into singleton normal form by a polynomial time algorithm.

\begin{proposition}
	\label{prop:singleton-normal-form}
	Let $\aut{R} = \la \Sigma, Q, q_\iota, \delta, \alpha \ra$ be a Rabin acceptor where $\alpha = \{(G_1,B_1),\ldots, \allowbreak (G_k,B_k)\}$.
	Define the acceptance condition $\alpha'$ to contain $(q,B_i)$ for every $i \in [1..k]$ and $q \in G_i$, and let $\aut{R}'$ be $\aut{R}$ with $\alpha$ replaced by $\alpha'$.
	Then $\aut{R}'$ is in singleton normal form and accepts the same language as $\aut{R}$.
	Also, $\aut{R}'$ is of size at most $|Q|$ times the size of $\aut{R}$.
\end{proposition}

\begin{proof}
	If an $\omega$-word $w$ satisfies a pair $(q,B)$ of $\alpha'$, then there exists a pair $(G_i,B_i)$ of $\alpha$ with $q \in G_i$ and $B_i = B$, so $w$ also satisfies the pair $(G_i,B_i)$ in $\alpha$.
	Conversely, if $w$ satisfies a pair $(G_i,B_i)$ in $\alpha$, then there exists $q \in G_i$ such that $q \in \inf(w)$, so $w$ also satisfies $(q, B_i)$ in $\alpha'$.
	Each pair $(G_i,B_i)$ in $\alpha$ is replaced by at most $|Q|$ pairs in $\alpha'$.
\end{proof}

\subsection{An ordering on sets of states}
\label{ssec:state-set-ordering}

Given an automaton, we define an ordering $\preceq$ on sets of its states that is used to coordinate between the characteristic sample and the learning algorithm for an IRA.

Let $\aut{M} = \la \Sigma, Q, q_\iota, \delta \ra$ be a deterministic complete automaton in which every state is reachable from the initial state $q_\iota$.
Recall from \Cref{ssec:T-Aut-definition}
that $\access(q)$ is the shortlex least string $s \in \Sigma^*$ such that $\delta(q_\iota,s) = q$.

For a set of states $S$, we define $\access(S)$ to be the sequence of values $\access(q)$ for $q \in S$, sorted into increasing shortlex order.
We define the total ordering $\preceq$ on sets of states of $\aut{M}$ as follows.  If $S_1, S_2 \subseteq Q$, then $S_1 \preceq S_2$ iff either $|S_1| < |S_2|$ or $|S_1| = |S_2|$ and the sequence $\access(S_1)$ is less than or equal to the sequence $\access(S_2)$ in the lexicographic ordering,
using the shortlex ordering on $\Sigma^*$ to compare the component entries.
For example, if 
$\access(S_1)=\langle \epsilon,a,baa\rangle$,
$\access(S_2)=\langle a,ba \rangle$, and
$\access(S_3)=\langle \epsilon,ab,ba\rangle $
then $S_2 \preceq S_1 \preceq S_3$.
\subsection{The learning algorithm \texorpdfstring{$\alg{L}_{Acc}^{\mbox{\scriptsize{IRA}}}$}{L\_Acc\textasciicircum IRA}}
\label{ssec:ira-learning-algorithm}

We begin with the description of the learning algorithm $\alg{L}_{Acc}^{\mbox{\scriptsize{IRA}}}$, which is used in the definition of the sample $T_{Acc}^{\mbox{\scriptsize{IRA}}}$ in the next section. The inputs to $\alg{L}_{Acc}^{\mbox{\scriptsize{IRA}}}$ are a deterministic complete automaton $\aut{M}$ and a sample $T$. The algorithm attempts to construct a singleton normal form Rabin acceptance condition $\beta$ to produce an acceptor $(\aut{M}, \beta)$ consistent with $T$.

The processing of an example $w$ from $T$ depends only on $\infss{\aut{M}}(w)$. The algorithm first computes the set $\{\infss{\aut{M}}(w) \mid (w,1) \in T\}$, sorts its elements into decreasing order $C_1, C_2, \ldots, C_{\ell}$ using $\preceq$, and for each $i \in [1..\ell]$ chooses a positive example $(z_i,1) \in T$ with $C_i = \infss{\aut{M}}(z_i)$.

At each stage $k$ of the learning algorithm, $\beta_k$ is a singleton normal form acceptance condition.
Initially, $\beta_0 = \emptyset$  (which is satisfied by no words) and $k=0$.

The main loop processes the positive examples $z_i$ for $i = 1,2,\ldots,\ell$.
If $z_i$ is accepted by $(\aut{M},\beta_k)$, then the algorithm goes on to the next positive example.
Otherwise, we say that the example $z_i$ \emph{causes the update of} $\beta_k$.
Let $G = \infss{\aut{M}}(z_i)$ and $B = Q \setminus G$ and let $S_k$ be the set of pairs $(q,B)$ such that $q \in G$ and there is no negative example $(w,0)$ in $T$ such that $w$ satisfies $(q,B)$.
Then $\beta_{k+1}$ is set to $\beta_k \cup S_k$ and $k$ is set to $k+1$.

When the positive examples $z_1, z_2, \ldots, z_{\ell}$ have been processed, let $\beta = \beta_k$ for the final value of $k$. If the Rabin acceptor $(\aut{M},\beta)$ is consistent with $T$, then it is returned.  If not, the learning algorithm returns the default acceptor of type DRA for $T$.

\begin{proposition}
	\label{prop:learner-is-monotonic-and-consistent}
	Let $\aut{R}$ be an IRA in singleton normal form with acceptance condition $\alpha$ and assume that the input $\aut{M}$ is an automaton isomorphic to the automaton of $\aut{R}$. Assume the sample $T$ is consistent with $\aut{R}$.
	For each $k$, if $z_i$ is the example that causes the update of $\beta_k$, then $(\aut{M},\beta_k)$ accepts $z_j$ for all $j < i$.
	Moreover, $(\aut{M},\beta)$ is consistent with $T$.
\end{proposition}

\begin{proof}
	No pair $(q,B)$ is added to the acceptance condition if there is a negative example in $T$ that satisfies it, so $(\aut{M},\beta)$ is consistent with all the negative examples in $T$.
	
	Consider any positive example $(w,1)$ from $T$.
	There exists $i$ with $\infss{\aut{M}}(w) = \infss{\aut{M}}(z_i)$.
	When $z_i$ is processed by the algorithm, if it is already accepted by the current $(\aut{M},\beta_k)$, then it (and the word $w$) is also accepted by every subsequent hypothesis, including $(\aut{M},\beta)$, because pairs are not removed from $\beta_k$.
	
	Otherwise, $z_i$ causes the update to $\beta_k$, and among the pairs in $S_k$ that are added to $\beta_k$ there is at least one that $z_i$ satisfies.  To see this, note that $z_i$ must satisfy some pair in $\alpha$, say $(q,B_j)$.
	Thus, $q \in \inf(z_i)$ and $B_j \cap \inf(z_i) = \emptyset$.
	The pair $(q,B)$ where $B = Q \setminus \inf(z_i)$ has $B_j \subseteq B$.  Thus any $\omega$-word that satisfies $(q,B)$ will also satisfy $(q,B_j)$. Because $T$ is consistent with $\aut{R}$, there can be no negative example in $T$ satisfying $(q,B)$, so pair $(q,B)$ is part of $S_k$ and is added to $\beta_k$.
	The word $z_i$ (and the word $w$) satisfies $(q,B)$ and is therefore accepted by $(\aut{M},\beta_{k+1})$ and every subsequent hypothesis, including $(\aut{M},\beta)$.
\end{proof}

\subsection{Constructing \texorpdfstring{$T_{Acc}^{\mbox{\scriptsize{IRA}}}$}{T\_Acc\textasciicircum IRA}}
\label{ssec:T-Acc-IRA}

In this section we describe the construction of the sample $T_{Acc}^{\mbox{\scriptsize{IRA}}}$, which conveys the acceptance condition of an IRA.
Let $\aut{R} = \la \Sigma, Q, q_\iota, \delta, \alpha \ra$ be a deterministic complete IRA of $n$ states in singleton normal form, and let $\aut{M}$ be its automaton.
The construction of the sample $T_{Acc}^{\mbox{\scriptsize{IRA}}}$ proceeds in stages, simulating the learning algorithm $\alg{L}_{Acc}^{\mbox{\scriptsize{IRA}}}$ on the portion of the sample constructed so far to determine what examples still need to be added.

Initially, $\gamma_0 = \emptyset$ and $k = 0$.
The acceptance condition $\gamma_k$ tracks the learning algorithm's $\beta_k$. The set of words accepted by $(\aut{M}, \gamma_k)$ is always a subset of the set of words accepted by $(\aut{M}, \alpha)$.
The main loop is as follows.
If $(\aut{M},\gamma_k)$ is equivalent to $(\aut{M},\alpha)$ then the construction of $T_{Acc}^{\mbox{\scriptsize{IRA}}}$ is complete.

Otherwise, let $D_k$ be the set of $\omega$-words that satisfy $\alpha$ but not $\gamma_k$.  Let $C$ be the $\preceq$-largest set in $\{\inf(w) \mid w \in D_k\}$, and let $w_{k+1} = \alg{Witness}(C,\aut{M})$, an ultimately periodic word of length $O(n^2)$.
Then $w_{k+1}$ is added as a positive example to $T_{Acc}^{\mbox{\scriptsize{IRA}}}$.

Let $B = Q \setminus \inf(w_{k+1})$.
Define $P_k$ to be the set of all $(q,B)$ such that 
$q \in \inf(w_{k+1})$ and
there is no $\omega$-word $w'$ that satisfies $(q,B)$ but not $\alpha$.
Set $\gamma_{k+1} = \gamma_k \cup P_k$.

For each $q \in \inf(w_{k+1})$ such that there is some $\omega$-word $w'$ that satisfies $(q,B)$ but not $\alpha$, let $u(v)^{\omega} = \alg{Witness}(\inf(w'),\aut{M})$ for some such $w'$ and include $(u(v)^{\omega},0)$ as a negative example in $T_{Acc}$. The example $u(v)^{\omega}$ is of length $O(n^2)$.
Then set $k$ to $k+1$ and continue with the main loop.

We prove a polynomial bound on the number of examples added to the sample $T_{Acc}^{\mbox{\scriptsize{IRA}}}$, thus showing that its length is bounded by a  polynomial in the size of $\aut{R}$.

\begin{proposition}
	\label{prop:bound-on-T-Acc}
	If the acceptance condition $\alpha$ is in singleton normal form and has $m$ pairs, then at most $m$ positive examples and at most $m|Q|$ negative examples are added to $T_{Acc}$.
\end{proposition}

\begin{proof}
	We say an acceptance condition $\gamma$ \emph{covers} a pair $(q,B)$ iff every $\omega$-word $w$ that satisfies $(q,B)$ also satisfies $\gamma$.
	We will show that after each positive example $w_{k+1}$ is added to $T_{Acc}$, the condition $\gamma_{k+1}$ covers at least one pair in $\alpha$ that was not covered by $\gamma_k$.
	
	Suppose not, and let $k+1$ be the least index for which $\gamma_{k+1}$ does not cover a pair of $\alpha$ that was not covered by $\gamma_k$.
	Because $w_{k+1}$ is an example that satisfies $\alpha$ but not $\gamma_k$, there must be a pair $(q,B_j)$ of $\alpha$ that is satisfied by $w_{k+1}$.
	Note that $\gamma_k$ does not cover the pair $(q,B_j)$.
	Then $q \in \inf(w_{k+1})$ and letting $B = Q \setminus \inf(w_{k+1})$, $B_j \subseteq B$.
	The pair $(q,B)$ will be added to $\gamma_k$ in constructing $\gamma_{k+1}$ because every word that satisfies $(q,B)$ also satisfies $(q,B_j)$.
	
	If $\gamma_{k+1}$ does not cover $(q,B_j)$, there must be a word $w'$ that satisfies $(q,B_j)$ but not $(q,B)$.
	So $q \in \inf(w')$ and $B_j \cap \inf(w') = \emptyset$ but $B \cap \inf(w') \neq \emptyset$.
	Let $B' = Q \setminus \inf(w')$, so $B_j \subseteq B'$.
	We have $B \cap B' \subsetneq B$.
	Because $\inf(w_{k+1})$ and $\inf(w')$ are SCCs that overlap in $q$, their union is an SCC as well.
	Let $w''$ be an $\omega$-word such that $\inf(w'')$ is the union of $\inf(w_{k+1})$ and $\inf(w')$.
	Then $Q \setminus \inf(w'') = B \cap B'$.
	Note that $w''$ satisfies $(q,B_j)$ because $B_j \subseteq B \cap B'$ and thus is a positive example of $\alpha$.
	
	Because $B \cap B'$ is a proper subset of $B$, $\inf(w'')$ is a proper superset of $\inf(w_{k+1})$ and the positive example $w''$ would have been considered before $w_{k+1}$ in the construction of $T_{Acc}$.
	(We can imagine all the positive examples of $\alpha$ being considered in order to find a maximum positive counterexample at each stage.)
	At that time, it was either passed over because (1) the current $\gamma_r$ already covered it, or (2) it contributed a new pair to the current $\gamma_r$ to yield $\gamma_{r+1}$.
	
	In case (1), there is some pair $(q'', B'')$ in $\gamma_r$ that is satisfied by $w''$.  Then $q'' \in \inf(w'')$ and $B'' \cap \inf(w'') = \emptyset$.
	Recall $\inf(w'')$ is the union of $\inf(w_{k+1})$ and $\inf(w')$.
	Thus, $B'' \cap \inf(w_{k+1}) = \emptyset$ and $B'' \cap \inf(w') = \emptyset$.
	Note that $q'' \in \inf(w_{k+1})$ or $q'' \in \inf(w')$.
	If $q'' \in \inf(w_{k+1})$, $w_{k+1}$ satisfies the pair $(q'',B'')$ in $\gamma_j$, a contradiction, because $w_{k+1}$ is not accepted by $\gamma_k$ and $r \le k$.  And if $q'' \in \inf(w')$ then $w'$ satisfies the pair $(q'',B'')$ in $\gamma_r$, a contradiction, because $w'$ is not accepted by $\gamma_{k+1}$ and $r \le k$.
	
	In case (2), the positive example $w''$ contributes at least one term $(q'',B'')$ to $\gamma_{r+1}$.
	In this case $B'' = B \cap B'$ and $q'' \in \inf(w'')$.
	Thus, $q'' \in \inf(w_{k+1})$ or $q'' \in \inf(w')$, so $w_{k+1}$ or $w'$ satisfies the term $(q'',B'')$ of $\gamma_{r+1}$, a contradiction because $r+1 \le k$ and neither $w_{k+1}$ nor $w'$ is covered by $\gamma_k$.
	
	Thus, each positive example added to $T_{Acc}$ covers a new pair of $\alpha$, and at most $m$ positive examples can be added.  Each positive example added requires at most $|Q|$ negative examples to avoid adding incorrect pairs, so at most $m|Q|$ negative examples are added.
\end{proof}

\subsection{Correctness of \texorpdfstring{$\alg{L}_{Acc}^{\mbox{\scriptsize{IRA}}}$}{L\_Acc\textasciicircum IRA}}

We prove the correctness of the learning algorithm $\alg{L}_{Acc}^{\mbox{\scriptsize{IRA}}}$ and show that the classes $\class{IRA}$ and $\class{ISA}$ are identifiable in the limit using polynomial time and data.

\begin{theorem}
	\label{theorem:L-Acc-IRA-works}
	Algorithm $\alg{L}_{Acc}^{\mbox{\scriptsize{IRA}}}$ runs in time polynomial in the sizes of the inputs $\aut{M}$ and $T$.  Let $\aut{R}$ be an IRA. If the input automaton $\aut{M}$ is isomorphic to the automaton of $\aut{R}$, and the sample $T$ is consistent with $\aut{R}$ and subsumes $T_{Acc}^{\mbox{\scriptsize{IRA}}}$, then algorithm $\alg{L}_{Acc}^{\mbox{\scriptsize{IRA}}}$ returns an IRA $(\aut{M},\beta)$ equivalent to $\aut{R}$.
\end{theorem}

\begin{proof}
	By \Cref{prop:inf-to-SCC}, $\alg{L}_{Acc}^{\mbox{\scriptsize{IRA}}}$ can construct the sequence $z_1, z_2, \ldots, z_{\ell}$ and the successive acceptance conditions $\beta_k$ in time polynomial in the size of $\aut{M}$ and the length of $T$.
	
	Assume $\aut{R}$ is an IRA, that $\aut{M}$ is isomorphic to the automaton of $\aut{R}$, and that the sample $T$ is consistent with $\aut{R}$ and subsumes $T_{Acc}^{\mbox{\scriptsize{IRA}}}$.  For ease of notation, we assume that the isomorphism is the identity.
	
	We show by induction that for each $k$, the acceptance condition $\beta_k$ in the learning algorithm $\alg{L}_{Acc}^{\mbox{\scriptsize{IRA}}}$ is the same as the acceptance condition $\gamma_k$ in the construction of $T_{Acc}^{\mbox{\scriptsize{IRA}}}$. This is true for $k = 0$ because $\beta_0 = \gamma_0 = \emptyset$.
	
	Assume that $\beta_k = \gamma_k$ for some $k \ge 0$.  If $(\aut{M},\gamma_k)$ is equivalent to $\aut{R} = (\aut{M},\alpha)$, then also $(\aut{M}, \beta_k)$ is equivalent to $\aut{R}$, and none of the remaining positive examples cause any additions to $\beta_k$.  Thus this is the final value of $k$, so $\beta = \beta_k$ and $(\aut{M},\beta)$ is equivalent to $\aut{R}$.
	
	If $(\aut{M}, \gamma_k)$ accepts a proper subset of the language accepted by $(\aut{M},\alpha)$, then in the construction of sample $T_{Acc}^{\mbox{\scriptsize{IRA}}}$, $D_k$ is equal to the $\omega$-words accepted by $(\aut{M},\alpha)$ but not by $(\aut{M},\gamma_k)$. This causes the positive example $(w_{k+1},1)$ to be added to $T_{Acc}^{\mbox{\scriptsize{IRA}}}$, where $\infss{\aut{M}}(w_{k+1})$ is $\preceq$-largest in the set $\{\infss{\aut{M}}(w) \mid w \in D_k\}$.
	
	In the learning algorithm, because $(\aut{M},\beta_k)$ does not accept $w_{k+1}$, \Cref{prop:learner-is-monotonic-and-consistent} implies that there must be an example $z_i$ that causes the update to $\beta_k$, and all of the examples $z_1, \ldots, z_{i-1}$ are accepted by $(\aut{M},\beta_k)$.
	Because for every positive example $(w,1)$ in $T$ there exists $j$ such that $\infss{\aut{M}}(w) = \infss{\aut{M}}(z_j)$, there must be some $r$ such that $\infss{\aut{M}}(w_{k+1}) = \infss{\aut{M}}(z_r)$.  Moreover, $i \le r$.
	
	If $i < r$, then $\infss{\aut{M}}(w_i)$ is strictly $\preceq$-larger than $\infss{\aut{M}}(w_r)$, which contradicts the choice of $w_{k+1}$ by the sample construction procedure, because $z_i$ is accepted by $(\aut{M},\alpha)$ but not $(\aut{M},\gamma_k)$.
	Thus $i = r$ and the example $z_i = w_{k+1}$ is the element that causes the update to $\beta_k$.
	The negative examples included in $T_{Acc}^{\mbox{\scriptsize{IRA}}}$ for the positive example $w_{k+1}$ ensure that the update to $\beta_k$ is the same as the update to $\gamma_k$, and $\beta_{k+1} = \gamma_{k+1}$.
	
	Because $\beta_k$ and $\gamma_k$ are equal for all $k$, for the final value of $k$, $\beta = \beta_k = \gamma_k$, and therefore $(\aut{M},\beta)$ is equivalent to $\aut{R}$.
	Because the IRA $(\aut{M},\beta)$ is consistent with $T$, it is the acceptor returned by $\alg{L}_{Acc}^{\mbox{\scriptsize{IRA}}}$.
\end{proof}

\begin{theorem}
	\label{theorem:ir-is-poly-ident-in-limit}
	The classes $\IR$ and $\IS$ are identifiable in the limit using polynomial time and data.
\end{theorem}
\begin{proof}
	By \Cref{prop:ib-ic-ir-is-duality} it suffices to prove this for $\class{IRA}$.
	Let $\aut{R}$ be an IRA in singleton normal form accepting the language $L$.  The characteristic sample $T_L = T_{Aut} \cup T_{Acc}^{\mbox{\scriptsize{IRA}}}$ is of size polynomial in the size of $\aut{R}$.
	
	The combined learning algorithm $\alg{L}^{\mbox{\scriptsize{IRA}}}$ with a sample $T$ as input first runs $\alg{L}_{Aut}$ on $T$ to get a deterministic complete automaton $\aut{M}$ and then runs $\alg{L}_{Acc}^{\mbox{\scriptsize{IRA}}}$ on inputs $\aut{M}$ and $T$ and returns the resulting acceptor.  $\alg{L}^{\mbox{\scriptsize{IRA}}}$ runs in time polynomial in the length of $T$ and returns a DRA consistent with $T$.
	
	Now assume that the sample $T$ is consistent with $\aut{R}$ and subsumes $T_L$.   Then by \Cref{theorem:L-Aut-works}, the automaton $\aut{M}$ is isomorphic to the automaton of $\aut{R}$.
	By \Cref{theorem:L-Acc-IRA-works}, the acceptor returned by $\alg{T}^{\mbox{\scriptsize{IRA}}}$ is an IRA $(\aut{M},\beta)$ that is equivalent to $\aut{R}$, and this is the acceptor also returned by $\alg{L}^{\mbox{\scriptsize{IRA}}}$.
\end{proof}

\section{Constructing characteristic samples in polynomial time}
\label{sec:poly-time-char-samples}
The definition of identification in the limit using polynomial time and data requires that a characteristic sample exist and be of polynomial size, but says nothing about the cost of computing it.  An additional desirable property is that a characteristic sample be computable in polynomial time given an acceptor $\aut{A}$ as input. Recall that when this holds, we say that the class is efficiently teachable.  We now show that given an acceptor that is fully informative we can design efficient teachers,
i.e. algorithms that run in polynomial time and compute the characteristic samples we have defined. This is conditioned on having polynomial time algorithms for equivalence (that are given in Sections~\ref{sec:inclusion-algorithms}-\ref{sec:inclusion-DMAs}). To claim the class $\class{IXA}$ is efficiently teachable we also need to show that we can construct such sets when starting with an acceptor that is not, say an IBA, but has an equivalent IBA acceptor (and similarly for the other classes). This is done in Sections~\ref{sec:computing-right-congruence-automaton}-\ref{sec:testing-membership-in-IX}. 

\subsection{Computing \texorpdfstring{$T_{Aut}$}{T\_Aut}}
\label{ssec:computing-T-Aut}

For $T_{Aut}$, we need to be able to decide for two states $q_1$ and $q_2$ of an acceptor $\aut{A}$ whether there exists an $\omega$-word that distinguishes them, and if so, to return one such word.
We are thus led to consider the problems of inclusion and equivalence.

\paragraph{The problems of inclusion and equivalence}
\label{ssec:inclusion-equivalence-definitions}

The \emph{inclusion problem} is the following.
Given as input two $\omega$-acceptors  $\aut{A}_1$ and $\aut{A}_2$
over the same alphabet,
determine whether the language accepted by $\aut{A}_1$ is
a subset of the language accepted by $\aut{A}_2$, that is,
whether $\sema{\aut{A}_1} \subseteq \sema{\aut{A}_2}$.
If so, the answer should be ``yes''; if not, the answer should
be ``no'' and a \emph{witness}, that is, an ultimately periodic
$\omega$-word $u(v)^{\omega}$ accepted by $\aut{A}_1$ but
rejected by $\aut{A}_2$.

The \emph{equivalence problem} is similar: 
the input is two $\omega$-acceptors $\aut{A}_1$ and $\aut{A}_2$ 
over the same alphabet, and the problem is to determine 
whether they are equivalent,
that is, whether $\sema{\aut{A}_1} = \sema{\aut{A}_2}$.
If so, the answer should be ``yes''; if not, the answer should
be ``no'' and a witness, that is, an ultimately periodic
$\omega$-word $u(v)^{\omega}$ that is accepted by exactly one of $\aut{A}_1$ and $\aut{A}_2$.

If we have a procedure to solve the inclusion problem,
at most two calls to it will solve the equivalence problem.
We describe polynomial time algorithms to solve the
inclusion problem for DBAs, DCAs and DPAs in \Cref{sec:inclusion-algorithms} and~\ref{sec:inclusion-DPAs}, for DRAs and DSAs in \Cref{sec:inclusion-DRAs}, and for DMAs in \Cref{sec:inclusion-DMAs}.
Referring to those sections, we obtain polynomial time algorithms to solve the equivalence problem for DBAs, DCAs and DPAs from Theorem \ref{theorem:schewe} and \Cref{theorem:inclusion-equivalence-for-DPAs},
for DRAs and DSAs from
\Cref{theorem:inclusion-equivalence-for-DRAs-DSAs}, and for DMAs from \Cref{theorem:inclusion-equivalence-for-DMAs}. Thus, we have the following.
\begin{theorem}
	\label{theorem:T-Aut-in-poly-time}
	Given an acceptor $\aut{A}$ of type IBA, ICA, IPA, IRA, ISA, or IMA, the sample $T_{Aut}$ for the automaton portion of $\aut{A}$ can be computed in polynomial time.
\end{theorem}

\begin{proof}
	Given an acceptor $\aut{A}$ and two states $q_1$ and $q_2$, to determine whether there is an $\omega$-word that distinguishes them, we call the relevant polynomial time equivalence algorithm on the acceptors $\aut{A}^{q_1}$ and $\aut{A}^{q_2}$, which returns a distinguishing word $u(v)^{\omega}$ if they are not equivalent.
\end{proof}

\subsection{Computing \texorpdfstring{$T_{Acc}$}{T\_Acc}}
\label{ssec:computing-T-Acc}

For $T_{Acc}$, the requirements depend on the type of acceptor.

\begin{proposition}
	\label{prop:T-Acc-IBA-in-poly-time}
	Given an IBA $\aut{B}$, the sample $T_{Acc}^{\mbox{\scriptsize{IBA}}}$ can be computed in time polynomial in the size of $\aut{B}$.
\end{proposition}

\begin{proof}
	Given an IBA $\aut{B} =\la \Sigma, Q, q_{\iota}, \delta, F \ra$ with automaton $\aut{M}$, the sample $T_{Acc}^{\mbox{\scriptsize{IBA}}}$ described in \Cref{ssec:T-Acc-IBA} is computed as follows.
	For each SCC $C \in \maxSCCs(Q \setminus F)$, let $u(v)^{\omega} = \alg{Witness}(C,\aut{M})$ and include the negative example $(u(v)^{\omega},0)$ in $T_{Acc}^{\mbox{\scriptsize{IBA}}}$.
	
	To see that these examples are sufficient, suppose that $q \in Q$ and $w \in \Sigma^{\omega}$ are such that $\aut{B}$ rejects $w$ and $q \in \inf(w)$.
	Then $D = \inf(w)$ is an SCC of $\aut{B}$ contained in $Q \setminus F$, so it is contained in some $C \in \maxSCCs(Q \setminus F)$, and there is a negative example $(u(v)^{\omega},0)$ in $T_{Acc}^{\mbox{\scriptsize{IBA}}}$ such that $\inf(u(v)^{\omega}) = C$.  Because $D \subseteq C$, we have $q \in C$.
\end{proof}

\begin{proposition}
	\label{prop:T-Acc-IPA-in-poly-time}
	Given an IPA $\aut{P}$, the sample $T_{Acc}^{\mbox{\scriptsize{IPA}}}$ can be computed in time polynomial in the size of $\aut{P}$.
\end{proposition}

\begin{proof}
	Given an IPA $\aut{P}$ with automaton $\aut{M}$, the computation of $T_{Acc}^{\mbox{\scriptsize{IPA}}}$ proceeds as described in \Cref{ssec:T-Acc-IPA}.
	That is, the canonical forest $\forest{F}^*(\aut{P})$ is computed in polynomial time, and for each node $C$ in the forest, $u(v)^{\omega} = \alg{Witness}(C,\aut{M})$ is computed and $(u(v)^{\omega},l)$ is added to $T_{Acc}^{\mbox{\scriptsize{IPA}}}$, where $l$ is the label of node $C$ in the canonical forest.
\end{proof}

\begin{proposition}
	\label{prop:T-Acc-IRA-in-poly-time}
	Given an IRA $\aut{R}$, the sample $T_{Acc}^{\mbox{\scriptsize{IRA}}}$ can be computed in time polynomial in the size of $\aut{R}$.
\end{proposition}

\begin{proof}
	Given an IRA $\aut{R} = (\aut{M},\alpha)$, the computation of $T_{Acc}$ proceeds as described in \Cref{ssec:T-Acc-IRA}.  At each stage of the computation, it is necessary to find an $\omega$-word $u(v)^{\omega}$ with the $\preceq$-largest $\infss{\aut{M}}(u(v)^{\omega}$) that is accepted by $(\aut{M},\alpha)$ and rejected by $(\aut{M},\gamma_k)$.
	\Cref{theorem:inclusion-equivalence-for-DRAs-DSAs} gives a polynomial time algorithm that not only tests the inclusion of two DRAs, but returns a witness $u(v)^{\omega}$ with the $\preceq$-largest $\infss{\aut{M}}(u(v)^{\omega})$ in the case of non-inclusion, because $\aut{M}$ is isomorphic to $\aut{M} \times \aut{M}$.
\end{proof}

\begin{proposition}
	\label{prop:T-Acc-IMA-in-poly-time}
	Given an IMA $\aut{A}$, the sample $T_{Acc}^{\mbox{\scriptsize{IMA}}}$ can be computed in time polynomial in the size of $\aut{A}$.
\end{proposition}

\begin{proof}
	Given an IMA $\aut{A} = (\aut{M},{\mathcal  F})$, the sample $T_{Acc}^{\mbox{\scriptsize{IMA}}}$ described in \Cref{ssec:T-Acc-IMA} is computed as follows. For each $F \in {\mathcal  F}$ determine whether $F$ is a reachable SCC of $\aut{M}$, and if so, compute $u(v)^{\omega} = \alg{Witness}(F,\aut{M})$ and add $(u(v)^{\omega},1)$ to the sample $T_{Acc}^{\mbox{\scriptsize{IMA}}}$.
\end{proof}

\begin{theorem}
	\label{theorem:char-samples-in-poly-time}
	Let $\aut{A}$ be an IBA, IPA, IRA, or IMA accepting the $\omega$-language $L$.  Then the characteristic sample $T_L$ for $\aut{A}$ can be computed in polynomial time in the size of $\aut{A}$.
\end{theorem}

\begin{proof}
	By \Cref{theorem:T-Aut-in-poly-time}, the sample $T_{Aut}$ can be computed in polynomial time in the size of $\aut{A}$, and by \Cref{prop:T-Acc-IBA-in-poly-time}, \ref{prop:T-Acc-IPA-in-poly-time}, \ref{prop:T-Acc-IRA-in-poly-time}, or
	\ref{prop:T-Acc-IMA-in-poly-time} the sample $T_{Acc}^{\mbox{\scriptsize{IBA}}}$, $T_{Acc}^{\mbox{\scriptsize{IPA}}}$, $T_{Acc}^{\mbox{\scriptsize{IRA}}}$, or $T_{Acc}^{\mbox{\scriptsize{IMA}}}$ can also be computed in polynomial time in the size of $\aut{A}$.
\end{proof}

Note that \Cref{theorem:char-samples-in-poly-time} does not imply that the class $\class{IXA}$ for $\class{X}\in\{\class{B},\class{P},\class{R},\class{M}\}$ is efficiently teachable, since this class also has representations by non-isomorphic automata.

\section{Inclusion algorithms}
\label{sec:inclusion-algorithms}

We show that there are polynomial time algorithms for the inclusion problem for DBAs, DCAs, DPAs, DRAs, DSAs and DMAs.
Recall that two calls to an inclusion algorithm suffice to solve the equivalence problem.
By Claim~\ref{clm:basic-relations-between-omega-aut}~(\ref{claim:DBA-DCA-complement}), the inclusion and
equivalence problems for DCAs are efficiently reducible
to those for DBAs, and vice versa.
Also, by Claim \ref{clm:basic-relations-between-omega-aut}~(\ref{claim:NBA-to-NPA}), the inclusion
and equivalence problems for DBAs are efficiently
reducible to those for DPAs.
Thus it suffices to consider the inclusion problem for DPAs, DRAs and DMAs.

\paragraph{Remark.}
In the case of DFAs, a polynomial algorithm for the inclusion problem can be obtained
using polynomial algorithms for complementation, intersection and emptiness
(since for any two languages $L_1 \subseteq L_2$
if and only if $L_1 \cap \overline{L_2}=\emptyset$). However, a similar approach does not work in the case of DPAs; although complementation and emptiness for DPAs can be computed in polynomial time,  intersection cannot~\cite[Theorem 9]{Boker18}.

For the inclusion problem for DBAs, DCAs and DPAs, 
Schewe~\cite{Schewe10,Schewe2011} gives the following result.
\begin{thmC}[\cite{Schewe10}]
	\label{theorem:schewe}
	The inclusion problems for DBAs, DCAs and DPAs are in NL.
\end{thmC}
Because NL (nondeterministic logarithmic space) is contained in polynomial time,
this implies the existence of polynomial time inclusion and equivalence
algorithms for DBAs, DCAs and DPAs.
For the sake of completeness, and to address the problem of returning a witness we include a proof sketch.
\begin{proof}[Proof sketch]
	For $i = 1,2$, let
	$\aut{P}_i = \la \Sigma, Q_i, (q_\iota)_i, \delta_i, \kappa_i \ra$ be a DPA.
	It suffices to guess two states $q_1 \in Q_1$ and $q_2 \in Q_2$,
	and two words $u \in \Sigma^*$ and $v \in \Sigma^+$,
	and to check that for $i = 1,2$, $\delta_i((q_\iota)_i,u) = q_i$
	and $\delta_i(q_i,v) = q_i$, and also,
	that the smallest value of $\kappa_1(q)$ in the loop 
	in $\aut{P}_1$ from $q_1$ to $q_1$ on input
	$v$ is odd,
	while the smallest value of $\kappa_2(q)$ in the loop 
	in $\aut{P}_2$ from $q_2$ to $q_2$ on input 
	$v$ is even.
	If these checks succeed, then $\sema{\aut{P}_1}$ is not a subset of $\sema{\aut{P}_2}$, and the ultimately periodic word $u(v)^{\omega}$ is a witness.
	
	Logarithmic space is enough to record the two
	guessed states $q_1$ and $q_2$ as well as 
	the current minimum values of $\kappa_1$ and $\kappa_2$ 
	as the loops on $v$ are traversed in the two automata.
	The words $u$ and $v$ need only be guessed symbol-by-symbol, 
	using a pointer in each automaton to keep track of its current state.
\end{proof}

This approach does not seem to work in the case of testing DRA or DMA inclusion,
because the acceptance conditions would seem to require keeping track of more information
than would fit in logarithmic space.
To supplement the proof sketch for Schewe's theorem, in the next section (\Cref{sec:inclusion-DPAs}) we give an explicit polynomial time algorithm for testing DPA inclusion. 

For inclusion of DRAs and DMAs, \cite{ClarkeDK93} provides a reduction to the problem of  model checking a formula in the temporal logic \ctlstar.
While the complexity of model checking \ctlstar\ formulas is in general PSPACE-hard, for the \fairctl\ fragment, \cite{EmersonL87} provides a model checking algorithm that runs in polynomial time. It is further shown in~\cite{ClarkeDK93} that the \ctlstar\ formulas they reduce to can be modified to formulas in a fragment slightly extending \fairctl\ that can still be handled by the 
model checking algorithm of~\cite{EmersonL87} for \fairctl. Thus overall this gives a polynomial time algorithm for inclusion of DRAs and DMAs. 
Since our learning algorithm relies on the sample including shortlex examples and the algorithm above does not guarantee shortlex counterexamples, we give in \Cref{sec:inclusion-DRAs} and \Cref{sec:inclusion-DMAs} polynomial time automata-theoretic algorithms for testing
inclusion for DRAs and DMAs that provide shortlex counterexamples, which are novel results. 

\section{Inclusion and equivalence for DPAs, DBAs, DCAs}
\label{sec:inclusion-DPAs}

In this section we describe an explicit polynomial time algorithm for the inclusion problem for two DPAs, which yields algorithms for DBAs and DCAs.
If $\aut{P} = \la \Sigma, Q, q_{\iota}, \delta, \kappa \ra$ is a complete DPA and $w \in \Sigma^{\omega}$, we let $\aut{P}(w)$ denote the minimum color visited by $\aut{P}$ infinitely often on input $w$, that is, $\aut{P}(w) = \kappa(\inf(w))$.

\subsection{Searching for $w$ with given minimum colors in two acceptors}

We first describe an algorithm that searches for an $\omega$-word that yields specified
minimum colors in two different DPAs over the same alphabet.

For $i = 1,2$, let
${\aut{P}}_i = \la \Sigma, Q_i,(q_\iota)_i, \delta_i, \kappa_i \ra$
be a DPA, and let $\aut{M}_i$ be the automaton of $\aut{P}_i$.
Given inputs of $\aut{P}_1$ and $\aut{P}_2$ and two nonnegative integers $k_1$ and $k_2$,
the \alg{Colors} algorithm constructs
the product automaton $\M = {\M}_1 \times {\M}_2$
and the set $Q' = \{(q_1,q_2) \in Q_1 \times Q_2 \mid \kappa_1(q_1) \ge k_1 \wedge \kappa_2(q_2) \ge k_2\}$.

The algorithm then computes $S = \maxSCCs(Q')$ for the automaton $\aut{M}$, and loops through the SCCs $C \in S$
checking whether
$C$ is reachable in $\aut{M}$,
$\min(\kappa_1(\pi_1(C))) = k_1$, and
$\min(\kappa_2(\pi_2(C))) = k_2$.
If so, it returns the ultimately periodic word $u(v)^{\omega} = \alg{Witness}(C,\aut{M})$.
If none of the elements $C \in S$ satisfies
this condition, then the answer ``no'' is returned.

\begin{theorem}
	\label{theorem:dpa-k1-k2}
	The algorithm \alg{Colors} takes as input two
	DPAs $\aut{P}_1$ and $\aut{P}_2$ over the same alphabet 
	and
	two nonnegative integers $k_1$ and $k_2$, runs in polynomial time, and
	determines whether there exists an $\omega$-word $w$ such that
	$\aut{P}_1(w) = k_1$ and $\aut{P}_2(w) = k_2$.
	If not, it returns the answer ``no''. If so, it returns an ultimately periodic $\omega$-word $u(v)^{\omega}$ such that $\aut{P}_1(u(v)^{\omega}) = k_1$ and $\aut{P}_2(u(v)^{\omega}) = k_2$. 
\end{theorem}

\begin{proof}
	
	The polynomial running time of the algorithm follows from Props.~\ref{prop:SCC-to-inf} and \ref{prop:maxSCCs-in-linear-time}.
	To see the correctness of the algorithm, suppose first that it returns an ultimately periodic word $u(v)^\omega$.
	This occurs only if it finds an SCC $C$ of $\aut{M}$ such that $C$ is reachable in $\aut{M}$, $\kappa_1(\pi_1(C)) = k_1$, and $\kappa_2(\pi_2(C)) = k_2$.
	Then for $i = 1,2$,
	$\pi_i(C)$ is the set of states visited infinitely often by $\aut{M}_i$ on
	the input $u(v)^{\omega}$, which has minimum color $k_i$.
	
	To see that the algorithm does not incorrectly answer ``no'',
	suppose $w$ is an $\omega$-word such that for $i = 1,2$, $\aut{P}_i(w) = k_i$.
	Let $D_i = \infss{\aut{M}_i}(w)$ be an SCC of $\aut{M}_i$.
	No state in $D_i$ has a color less than $k_i$, so if
	$D = \infss{\aut{M}}(w)$, then $D \subseteq Q'$.
	Also, $D$ is a reachable SCC in $\aut{M}$.
	
	Then $D$ is contained in some element $C$ of $\maxSCCs(Q')$.
	Because there are no states $(q_1,q_2)$ in $C$ with $\kappa_1(q_1) < k_1$ or $\kappa_2(q_2) < k_2$, we must have $\kappa_i(\pi_i(C)) = k_i$ for $i = 1,2$.
	Also, $C$ is reachable in $\aut{M}$ because $D$ is.
	Thus, the algorithm will find at least one such $C$
	and return $u(v)^{\omega}$ such that $\infss{\aut{M}}(u(v)^{\omega}) = C$.
\end{proof}

\subsection{An inclusion algorithm for DPAs}

The inclusion problem for DPAs $\aut{P}_1$ and $\aut{P}_2$ 
over the same alphabet can be solved by
looping over all odd $k_1$ in the range of $\kappa_1$ and
all even $k_2$ in the range of $\kappa_2$,
calling the \alg{Colors} algorithm with inputs
$\aut{P}_1$, $\aut{P}_2$, $k_1$, and $k_2$.
If the \alg{Colors} algorithm returns 
any witness $u(v)^\omega$, then
$u(v)^{\omega} \in \sema{\aut{P}_1} \setminus \sema{\aut{P}_2}$,
and $u(v)^\omega$ is returned as a witness of non-inclusion.
Otherwise, by \Cref{theorem:dpa-k1-k2}, there is no $\omega$-word $w$
accepted by $\aut{P}_1$ and not accepted by $\aut{P}_2$, and
the answer ``yes'' is returned for the inclusion problem.
Note that for $i = 1,2$,
the range of $\kappa_i$ has at most $|Q_i|$ distinct elements.
Thus we have the following.

\begin{theorem}
	\label{theorem:inclusion-equivalence-for-DPAs}
	There are polynomial time algorithms for the inclusion and equivalence problems
	for two DPAs over the same alphabet.
\end{theorem}

From Claim~\ref{clm:basic-relations-between-omega-aut}~(\ref{claim:NBA-to-NPA}~and~\ref{claim:DBA-DCA-complement}),
we have the following.

\begin{theorem}
	\label{theorem:inclusion-equivalence-for-DBAs-DCAs}
	There are polynomial time algorithms for the inclusion and equivalence problems
	for two DBAs (or DCAs) over the same
	alphabet.
\end{theorem}

\section{An inclusion algorithm for DRAs}
\label{sec:inclusion-DRAs}

In this section we describe a polynomial time algorithm to solve the inclusion problem for two DRAs.  The algorithm returns a $\preceq$-largest witness in the case of non-inclusion.

\begin{algorithm}
{\small{
	\caption{$\alg{SubInc}^{\mathit{DRA}}$}\label{alg:SubInc-DRA}
	\begin{algorithmic}
		\Require{Two DRAs $\aut{R}_1 = (\aut{M}_1,\alpha_1)$ and $\aut{R}_2 = (\aut{M}_2,\alpha_2)$ in singleton normal form, where $\alpha_1 = \{(q',B')\}$ and $\alpha_2 = \{(q_1'',B_1''), \ldots, (q_k'', B_k'')\}$, and a set $S$ of states of $\aut{M} = \aut{M}_1 \times \aut{M}_2$.}
		\Ensure{$u(v)^{\omega} \in \sema{\aut{R}_1} \setminus \sema{\aut{R}_2}$ with $\infss{\aut{M}}(u(v)^{\omega}) \subseteq S$ if such exists, else ``none''.}
		
		\State{$\aut{M} = \aut{M}_1 \times \aut{M}_2$}
		\State{$W \leftarrow \emptyset$}
		\State{$S' \leftarrow S \setminus \{(q_1,q_2) \in S \mid q_1 \in B'\}$}
		\State $\mathcal{C} \leftarrow \maxSCCs(S')$
		\For{each reachable $C \in \mathcal{C}$ such that $q' \in \pi_1(C)$}
		\If{for no $j \in [1..k]$ is $q_j'' \in \pi_2(C)$ and $B_j'' \cap \pi_2(C) = \emptyset$}
		\State $W \leftarrow W \cup \{\alg{Witness}(C,\aut{M})\}$ \Comment{A new candidate witness}
		\Else
		\State{$J = \{q_j'' \mid j \in [1..k], B_j'' \cap \pi_2(C) = \emptyset\}$}
		\State{$S'' \leftarrow C \setminus \{(q_1,q_2) \in C \mid q_2 \in J \}$}
		\State{Call $\alg{SubInc}^{\mathit{DRA}}$ recursively with $\aut{R}_1$, $\aut{R}_2$, and $S''$}
		\If{the returned value is $u(v)^{\omega}$}
		\State{$W \leftarrow W \cup \{u(v)^{\omega}\}$}
		\EndIf
		\EndIf
		\EndFor
		\If{$W$ is $\emptyset$}
		\State{\Return{``none''}}
		\Else
		\State{Let $u(v)^{\omega} \in W$ have the $\preceq$-largest value of $\infss{\aut{M}}(u(v)^{\omega})$}
		\State{\Return{$u(v)^{\omega}$}}
		\EndIf
	\end{algorithmic}
}}\end{algorithm}

The algorithm $\alg{SubInc}^{\mathit{DRA}}$ takes as input two DRAs $\aut{R}_1 = (\aut{M}_1,\alpha_1)$ and $\aut{R}_2 = (\aut{M}_2,\alpha_2)$ in singleton normal form, where $\alpha_1$ consists of a single pair $(q',B')$. It also takes as input a subset $S$ of the state set of the product automaton $\aut{M} = \aut{M}_1 \times \aut{M}_2$. The problem it solves is to determine whether there exists an $\omega$-word $u(v)^{\omega}$ with $\infss{\aut{M}}(u(v)^{\omega}) \subseteq S$ such that $u(v)^{\omega} \in \sema{\aut{R}_1} \setminus \sema{\aut{R}_2}$.
If there is such a word, the algorithm returns one with the $\preceq$-largest value of $\infss{\aut{M}}(u(v)^{\omega})$, and otherwise, it returns ``none''.

\begin{proposition}
	\label{prop:correctness-of-Inc-DRA}
	For $i = 1,2$ let $\aut{R}_i = (\aut{M}_i,\alpha_i)$ be a DRA in singleton normal form.  Assume $\alpha_1 = \{(q',B')\}$ and $\alpha_2 = \{(q_1'',B_1''), \ldots, (q_k'',B_k'')\}$.  Let $\aut{M} = \aut{M}_1 \times \aut{M}_2$, and let $S$ be a subset of the states of $\aut{M}$.  Then with inputs $\aut{R}_1$, $\aut{R}_2$, and $S$, the algorithm $\alg{SubInc}^{\mathit{DRA}}$ runs in polynomial time and returns $u(v)^{\omega} \in \sema{\aut{R}_1} \setminus \sema{\aut{R}_2}$ with the $\preceq$-largest value of $\infss{\aut{M}}(u(v)^{\omega})$ contained in $S$, if such exists, else it returns ``none''.
\end{proposition}

\begin{proof}
	When the element $u(v)^{\omega} = \alg{Witness}(\aut{M},C)$ is added to $W$, we have that $C$ is a reachable SCC of $\aut{M}$ contained in $S$, $q' \in \pi_1(C)$, and $B' \cap \pi_1(C) = \emptyset$ (because $S'$ contains no elements $(q_1,q_2)$ with $q_1 \in B'$), so $(\aut{M}_1,\{(q',B')\})$ accepts $u(v)^{\omega}$.  Also, we have that for no $j \in [1..k]$ do we have $q_j'' \in \pi_2(C)$ and $B_j'' \cap \pi_2(C) = \emptyset$, so $\aut{R}_2$ rejects $u(v)^{\omega}$.  Thus, any returned $u(v)^{\omega}$ is a witness to the non-inclusion of $\sema{(\aut{M}_1,\{(q',B')\})}$ in $\sema{\aut{R}_2}$ with $\infss{\aut{M}}(u(v)^{\omega}) \subseteq S$.
	
	We now show by induction on the recursive calls that if $w$ is any $\omega$-word such that $\infss{\aut{M}}(w) \subseteq S$, $(\aut{M}_1, \{(q',B')\})$ accepts $w$, and $\aut{R}_2$ rejects $w$, then $\alg{SubInc}^{\mathit{DRA}}$ returns a witness $u(v)^{\omega}$ such that $\infss{\aut{M}}(u(v)^{\omega})$ is at least as large as $\infss{\aut{M}}(w)$ in the $\preceq$-ordering.  Let $D = \infss{\aut{M}}(w)$. Then $D$ is a reachable SCC of $\aut{M}$ such that $q' \in \pi_1(D)$, $B' \cap \pi_1(D) = \emptyset$, and for no $j \in [1..k]$ do we have $q_j'' \in \pi_2(D)$ and $B_j'' \cap \pi_2(D) = \emptyset$.  Then $D \subseteq S'$ because $B' \cap \pi_1(D) = \emptyset$.  Thus, $D$ must be a subset of exactly one of the elements $C$ of $\maxSCCs(S')$.  Then $C$ is reachable, $q' \in \pi_1(C)$, and $B' \cap \pi_1(C) = \emptyset$ (because $C$ is a subset of $S'$).  
	
	If $C$ is such that for no $j \in [1..k]$ do we have $q_j'' \in \pi_2(C)$ and $B_j'' \cap \pi_2(C) = \emptyset$, then a witness $u(v)^{\omega} = \alg{Witness}(C,\aut{M})$ is added to $W$, and we have that $C = \infss{\aut{M}}(u(v)^{\omega})$ is at least as large in the $\preceq$-ordering as $D = \infss{\aut{M}}(w)$, because $D \subseteq C$.
	
	Otherwise, the set $J = \{q_j'' \mid j \in [1..k], B_j'' \cap \pi_2(C) = \emptyset\}$ is non-empty, and the algorithm removes from $C$ all the states $(q_1,q_2)$ such that $q_2 \in J$ to form the set $S''$. Because $D \subseteq C$, if $B_j'' \cap \pi_2(C) = \emptyset$, then also $B_j'' \cap \pi_2(D) = \emptyset$.  Thus, if for any $q_j'' \in J$ we have $q_j'' \in \pi_2(D)$, this would violate the assumption that $\aut{R}_2$ rejects $w$.  Hence, $D \subseteq S''$, and by the inductive assumption on the recursive calls, the recursive call to $\alg{SubInc}^{\mathit{DRA}}$ returns a witness $u(v)^{\omega}$ such that $\infss{\aut{M}}(u(v)^{\omega})$ is at least as large in the $\preceq$-ordering as $\infss{\aut{M}}(w)$.  Because the top-level algorithm returns $u(v)^{\omega}$ to maximize $\infss{\aut{M}}(u(v)^{\omega})$ with respect to $\preceq$, it will be at least as large as $\infss{\aut{M}}(w)$.
	
	For the polynomial running time, we note that all the SCCs $C$ considered are distinct elements of a decreasing forest of SCCs for the automaton $\aut{M}$, and so there can be at most as many as the number of states of $\aut{M}$.
\end{proof}

\begin{theorem}
	\label{theorem:inclusion-equivalence-for-DRAs-DSAs}
	There are polynomial time algorithms to solve the inclusion and equivalence problems for two DRAs (resp. DSAs) $\aut{R}_1$ and $\aut{R}_2$.  In the case of non-inclusion or non-equivalence, these algorithms return a witness $u(v)^{\omega}$ with the $\preceq$-largest value of $\infss{\aut{M}}(u(v)^{\omega})$, where $\aut{M}$ is the product of the automata $\aut{R}_1$ and $\aut{R}_2$.
\end{theorem}

\begin{proof}
	It suffices to consider just DRAs, by Claim~\ref{clm:basic-relations-between-omega-aut}~(\ref{claim:DRA-DSA-complement}).  Given two DRAs $\aut{R}_1 = (\aut{M}_1,\alpha_1)$ and $\aut{R}_2 = (\aut{M}_2,\alpha_2)$, we may assume they are in singleton normal form. Then for each pair $(q_i,B_i)$ in $\alpha_1$, we call $\alg{SubInc}^{\mathit{DRA}}$ with inputs $(\aut{M}_1, \{(q_i,B_i)\})$, $\aut{R}_2$, with $S$ equal to the whole state set of $\aut{M}$. If all of these calls return ``none'', then $\sema{\aut{R}_1}$ is a subset of $\sema{\aut{R}_2}$, and the answer returned is ``yes''.  Otherwise, one or more calls return a witness, and $u(v)^{\omega}$ is returned such that $\infss{\aut{M}}(u(v)^{\omega})$ is $\preceq$-largest among the witnesses returned by the calls.  The running time and correctness follow from the running time and correctness guarantees of $\alg{SubInc}^{\mathit{DRA}}$.
\end{proof}

\section{An inclusion algorithm for DMAs}
\label{sec:inclusion-DMAs}

In this section we develop a polynomial time algorithm
to solve the inclusion problem for two DMAs over the same alphabet.
The proof proceeds in two parts: (1) a polynomial time reduction
of the inclusion problem for two DMAs to the inclusion problem
for a DBA and a DMA, and (2) a polynomial time algorithm for
the inclusion problem for a DBA and a DMA.

\subsection{Reduction of DMA inclusion to DBA/DMA inclusion}

We first reduce the problem of inclusion for two arbitrary DMAs
to the inclusion problem for two DMAs 
where the first one has just a single final state set.
For $i = 1,2$, define the DMA
$\aut{U}_i = \la Q_i, \Sigma, (q_\iota)_i, \delta_i, \mathcal{F}_i \ra$,
where $\aut{F}_i$ is the set of final state sets for $\aut{U}_i$.
Let the elements of $\mathcal{F}_1$ be $\{F_1, \ldots, F_k\}$,
and for each $j \in [1..k]$, let
\[\aut{U}_{1,j} = \la Q_1, \Sigma, (q_\iota)_1, \delta_1, \{F_j\} \ra,\]
that is, 
$\aut{U}_{1,j}$ is $\aut{U}_1$ with $F_j$ as its only final state set.
Then by the definition of DMA acceptance,
\[\sema{\aut{U}_1} = \bigcup_{j=1}^k \sema{\aut{U}_{1,j}},\]
which implies that to test whether
$\sema{\aut{U}_1} \subseteq \sema{\aut{U}_2}$,
it suffices to test for all $j \in [1..k]$ that
$\sema{\aut{U}_{1,j}} \subseteq \sema{\aut{U}_2}$.

\begin{proposition}
	\label{prop:arbitrary-to-one-final-state-set}
	Suppose $\alg{L}$ is a procedure that solves the
	inclusion problem for two DMAs over the same alphabet,
	assuming that the first DMA has a single final state set.
	Then there is an algorithm that solves the inclusion
	problem for two arbitrary DMAs over the same alphabet,
	say $\aut{U}_1$ and $\aut{U}_2$,
	which simply makes $|\mathcal{F}_1|$ calls to $\alg{L}$, where $\mathcal{F}_1$
	is the family of final state sets of $\aut{U}_1$.
\end{proposition}

Next we describe a procedure $\alg{SCCtoDBA}$ that
takes as inputs a deterministic automaton $\aut{M}$, an SCC $F$ of $\aut{M}$,
and a state $q \in F$, and returns
a DBA $B(\aut{M},F,q)$ that accepts exactly $L(\aut{M},F,q)$,
where $L(\aut{M},F,q)$ is the set of
$\omega$-words $w$ that visit only the
states of $F$ when processed by $\aut{M}$ starting at state $q$,
and visits each of them infinitely many times.

Assume the states in $F$ are $\{q_0, q_1, \ldots, q_{m-1}\}$,
where $q_0 = q$.
The DBA $B(\aut{M}, F, q)$ is $\la Q', \Sigma, q_0, \delta', \{q_0\} \ra$,
where we define $Q'$ and $\delta'$ as follows.
We create new states $r_{i,j}$ for $i,j \in [0..m-1]$ such that
$i \neq j$, and denote the set of these by $R$.
We also create a new dead state $d_0$.
Then the set of states $Q'$ is $Q \cup R \cup \{d_0\}$.

For $\delta'$,
the dead state $d_0$ behaves as expected:
for all $\sigma \in \Sigma$, $\delta'(d_0,\sigma) = d_0$.
For the other states in $Q'$, let
$\sigma \in \Sigma$ and $i \in [0..m-1]$.
If $\delta(q_i, \sigma)$ is not in $F$, then
in order to deal with runs that would visit
states outside of $F$, we define
$\delta'(q_i,\sigma) = d_0$ and, for all $j \neq i$,
$\delta'(r_{i,j},\sigma) = d_0$.

Otherwise, for some $k \in [0..m-1]$ we have
$q_k = \delta(q_i,\sigma)$.
If $k = (i+1) \bmod m$, then we define
$\delta'(q_i,\sigma) = q_k$, and otherwise we
define
$\delta'(q_i,\sigma) = r_{k,(i+1) \bmod m}$.
For all $j \in [0..m-1]$ with $j \neq i$,
if $k = j$, we define $\delta'(r_{i,j},\sigma) = q_k$,
and otherwise we define
$\delta'(r_{i,j},\sigma) = r_{k,j}$.

Intuitively, for an input from $L(\aut{M},F,q)$, in $B(\aut{M},F,q)$
the states $q_i$ are visited in a repeating cyclic
order: $q_0, q_1, \ldots, q_{m-1}$, and the meaning of the
state $r_{i,j}$ is that at this point in the input,
$\aut{M}$ would be in state $q_i$, and the machine $B(\aut{M},F,q)$
is waiting for a transition that would arrive at state $q_j$ in $\aut{M}$,
in order to proceed to state $q_j$ in $B(\aut{M},F,q)$.\footnote{This construction is reminiscent of the construction transforming a generalized B\"uchi into a B\"uchi automaton~\cite{Vardi08,Choueka74}, by considering each state in $F$ as a singleton set of a generalized B\"uchi, but here we need to send transitions to states outside $F$ to a sink state.}
An example of the construction is shown in Figure~\ref{fig:example-for-B-M-F-q};
the dead state and unreachable states are omitted for clarity.

\begin{figure}[t]
	\begin{center}
		\scalebox{0.7}{
			\begin{tikzpicture}[->,>=stealth',shorten >=1pt,auto,node
			distance=1.8cm,semithick,initial text=,initial where=left]
			
			\node (start) {$\M:$};
			\node[state,initial] (q0) [right of=start] {$q_0$};
			\node[state] (q1) [right of=q0] {$q_1$};
			\node[state] (q2) [below of=q1] {$q_2$};
			\node[state] (q3) [below of=q0] {$q_3$};
			
			\path (q0) edge [loop above]        node {$b$} (q0);
			\path (q0) edge [above, near start] node {$a$} (q2);
			\path (q1) edge [above]             node {$a$} (q0);
			\path (q1) edge [bend right]        node {$b$} (q2);
			\path (q2) edge [bend right]        node {$a$} (q1);
			\path (q2) edge                     node {$b$} (q3);
			\path (q3) edge [below, near start] node {$a$} (q1);
			\path (q3) edge                     node {$b$} (q0);
			
			\node (next) [right of=start, node distance=8cm] {$B(\M,F,q):$}; 
			
			\node[state,initial,accepting] (bq0) [right of=next, node distance=2.5cm] {$q_0$};
			\node[state] (br01) [above right of=bq0] {$r_{0,1}$};
			\node[state] (br21) [right of=bq0] {$r_{2,1}$};
			\node[state] (bq1)  [right of=br21] {$q_1$};
			\node[state] (br02) [below right of=bq1] {$r_{0,2}$};
			\node[state] (bq2) [below left of=br02] {$q_2$};
			\node[state] (br10) [left of=bq2] {$r_{1,0}$};
			\node[state] (br20) [left of=br10] {$r_{2,0}$};
			
			\path (bq0) edge [above]        node {$b$} (br01);
			\path (bq0) edge [below]        node {$a$} (br21);
			\path (br01) edge [loop above]  node {$b$} (br01);
			\path (br01) edge [right]       node {$a$} (br21);
			\path (br21) edge [above]       node {$a$} (bq1);
			\path (bq1) edge [right]        node {$a$} (br02);
			\path (bq1) edge [right]        node {$b$} (bq2);
			\path (br02) edge [loop right]  node {$b$} (br02);
			\path (br02) edge [below]       node {$a$} (bq2);
			\path (bq2) edge [above]        node {$a$} (br10);
			\path (br10) edge [right]       node {$a$} (bq0);
			\path (br10) edge [below, bend left] node {$b$} (br20);
			\path (br20) edge [above, bend left] node {$a$} (br10);
			
			\end{tikzpicture} }   
		\caption{Example of the construction of $B(\M,F,q)$ with $F = \{q_0,q_1,q_2\}$ and $q = q_0$.}
		\label{fig:example-for-B-M-F-q}
	\end{center}
\end{figure}

\begin{lemma}
	Let $\aut{M}$ be a deterministic automaton with alphabet $\Sigma$ and
	states $Q$, and let $F$ be an SCC of $\aut{M}$ and $q \in F$.
	With these inputs, the procedure $\alg{SCCtoDBA}$ runs in polynomial time and
	returns the DBA $B(\aut{M},F,q)$, which accepts the language $L(\aut{M},F,q)$
	and has $|F|^2+1$ states.
\end{lemma}

\begin{proof}
	Suppose $w$ is in $L(\M,F,q)$.
	Let $q = s_0, s_1, s_2, \ldots$ be the sequence of states in the
	run of $\M$ from state $q$ on input $w$.
	This run visits only states in $F$ and visits each one of them
	infinitely many times.
	We next define a particular increasing sequence 
	$i_{k,\ell}$ of indices in $s$,
	where $k$ is a positive integer and $\ell \in [0,m-1]$.
	These indices mark particular visits to the states 
	$q_0, q_1, \ldots, q_{m-1}$
	in repeating cyclic order.
	The initial value is $i_{1,0} = 0$, marking the initial visit to $q_0$.
	If $i_{k,\ell}$ has been defined and $\ell < m-1$, then $i_{k,\ell+1}$
	is defined as the least natural number $j$ such that $j > i_{k,\ell}$
	and $s_j = q_{\ell+1}$, marking the next visit to $q_{\ell+1}$.
	If $\ell = m-1$, then $i_{k+1,0}$ is defined
	as the least natural number $j$ such that $j > i_{k,\ell}$ and
	$s_j = q_0$, marking the next visit to $q_0$.
	
	There is a corresponding division of $w$ into a concatenation
	of finite segments $w_{1,1},\allowbreak w_{1,2}, \ldots,  w_{1,m-1}, w_{2,0}, \ldots$
	between consecutive elements in the increasing sequence of indices.
	An inductive argument shows that in $B(\M,F,q)$, the
	prefix of $w$ up through $w_{k,\ell}$ arrives at the state $q_\ell$,
	so that $w$ visits $q_0$ infinitely often and is therefore
	accepted by $B(\M,F,q)$.
	
	Conversely, suppose $B(\M,F,q)$ accepts the $\omega$-word $w$.
	Let $s_0, s_1, s_2, \ldots$ be the run of $B(\M,F,q)$ on $w$,
	and let $t_0, t_1, t_2, \ldots$ be the run of $\M$ starting
	from $q$ on input $w$.
	An inductive argument shows that if $s_n = q_i$ then $t_n = q_i$, and
	if $s_n = r_{i,j}$ then $t_n = q_i$.
	Because the only way the run $s_0, s_1, \ldots$ can visit the final
	state $q_0$
	infinitely often is to progress through the states
	$q_0, q_1, \ldots q_{m-1}$ in repeating cyclic order, the run
	$t_0, t_1, \ldots$ must visit only states in $F$ and visit
	each of them infinitely often, so $w \in L(\M,F,q)$.
	
	The DBA $B(\M,F,q)$ has a dead state, and $|F|$ states for each
	element of $F$, for a total of $|F|^2 + 1$ states.
	The running time of the procedure $\alg{SCCtoDBA}$ is linear in the size of $\M$ and
	the size of the resulting DBA, which is polynomial in the size of $\M$.
\end{proof}

We now show that this construction may be used to reduce the inclusion of
two DMAs to the inclusion of a DBA and a DMA.
Recall that if $\A$ is an acceptor and $q$ is a state of $\A$,
then $\A^q$ denotes the acceptor $\A$ with the initial state changed
to $q$.

\begin{lemma}
	Let $\aut{U}_1$ be a DMA with automaton $\M_1$ and a single final
	state set $F_1$.
	Let $\aut{U}_2$ be an arbitrary DMA over the same alphabet as $\aut{U}_1$,
	with automaton $\M_2$ and family of final state sets $\mathcal{F}_2$.
	Let $\M$ denote the product automaton $\M_1 \times \M_2$ 
	with unreachable states removed.
	Then $\sema{\aut{U}_1} \subseteq \sema{\aut{U}_2}$ iff for
	every state $(q_1,q_2)$ of $\M$ with $q_1 \in F_1$ we have
	$\sema{B(\M_1,F_1,q_1)} \subseteq \sema{\aut{U}_2^{q_2}}$.
\end{lemma}

\begin{proof}
	Suppose that for some state $(q_1,q_2)$ of $\M$ with $q_1 \in F_1$, we have
	$w \in \sema{B(\M_1,F_1,q_1)} \setminus \sema{\aut{U}_2^{q_2}}$.
	Let $C_1$ be the set of states visited infinitely often in $B(\M_1,F_1,q_1)$
	on input $w$, and let $C_2$ be the set of states visited infinitely
	often in $\aut{U}_2^{q_2}$ on input $w$.
	Then $Q_1 \cap C_1 = F_1$ and $C_2 \not\in \mathcal{F}_2$.
	Let $u$ be a finite word such that $\M(u) = (q_1,q_2)$.
	Then $\infss{\M_1}(uw) = Q_1 \cap C_1 = F_1$ and $\infss{\M_2}(uw) = C_2$,
	so $uw \in \sema{\aut{U}_1} \setminus \sema{\aut{U}_2}$.
	
	Conversely,
	suppose that $w \in \sema{\aut{U}_1} \setminus \sema{\aut{U}_2}$.
	For $i = 1,2$ let $C_i = \infss{\M_i}(w)$.
	Note that $C_1 = F_1$ and $C_2 \not\in \mathcal{F}_2$.
	Let $w = xw'$, where $x$ is a finite prefix of $w$ that is
	sufficiently long that
	the run of $\M_1$ on $w$ does not visit any state outside $C_1$ after
	$x$ has been processed, and for $i = 1,2$ let $q_i = \M_i(x)$.
	Then $(q_1,q_2)$ is a (reachable) state of $\M$, $q_1 \in F_1$,
	and the $\omega$-word
	$w'$, when processed by $\M_1$ starting at state $q_1$ visits only
	states of $C_1 = F_1$ and visits each of them infinitely many times,
	that is, $w' \in \sema{B(\M_1,F_1,q_1)}$.
	Moreover, when $w'$ is processed by $\M_2$ starting at state $q_2$, the
	set of states visited infinitely often is $C_2$, 
	which is not in $\mathcal{F}_2$.
	Thus, $w' \in \sema{B(\M_1,F_1,q_1)} \setminus \sema{\aut{U}_2^{q_2}}$.
\end{proof}

To turn this into an algorithm to test inclusion for two DMAs,
$\aut{U}_1$ with automaton $\M_1$ and a single final state set $F_1$ that is
an SCC of $\M_1$ and
$\aut{U}_2$ with automaton $\M_2$,
we proceed as follows.
Construct the product automaton 
$\M = \M_1 \times \M_2$ with unreachable states removed,
and for each state $(q_1,q_2)$ of $\M$, if $q_1 \in F_1$,
construct the DBA $B(\M_1,F_1,q_1)$ and the DMA $\aut{U}_2^{q_2}$ and
test the inclusion of language accepted by the DBA in the language
accepted by the DMA.
If all of these tests return ``yes'', 
then the algorithm returns ``yes'' for
the inclusion question for $\aut{U}_1$ and $\aut{U}_2$.
Otherwise, for the first test 
that returns ``no'' and a witness $u(v)^\omega$,
the algorithm finds by breadth-first search a minimum length
finite word $u'$ such that $\M(u') = (q_1,q_2)$, 
and returns the witness $u'u(v)^\omega$.

Combining this with \Cref{prop:arbitrary-to-one-final-state-set}, 
we have the following.

\begin{theorem}
	\label{theorem:DMA-DMA-reduced-to-DBA-DMA}
	Let $\alg{L}$ be an algorithm to test inclusion for an arbitrary DBA
	and an arbitrary DMA over the same alphabet.
	There is an algorithm to test inclusion for an arbitrary pair
	of DMAs $\aut{U}_1$ and $\aut{U}_2$ over the same alphabet
	whose running time is linear in the sizes of $\aut{U}_1$ and
	$\aut{U}_2$ plus the time for at most $k \cdot |Q_1| \cdot |Q_2|$
	calls to the procedure $\alg{L}$, where $k$ is the number of final state
	sets in $\aut{U}_1$, and $Q_i$ is the state set of $\aut{U}_i$
	for $i = 1,2$.
\end{theorem}

\subsection{A DBA/DMA inclusion algorithm}
\label{ssec:dba-dma-inclusion}

In this section, we give a polynomial time algorithm 
\alg{DBAinDMA} to test inclusion
for an arbitrary DBA and an arbitrary DMA over the same alphabet.

Assume the inputs are
a DBA $\aut{B} = (\aut{M}_1,F_1)$ and a DMA $\aut{U} = (\aut{M}_2,{\mathcal  F})$.
The overall strategy of the algorithm 
is to seek an SCC $C$ of $\aut{M} = \aut{M}_1 \times \aut{M}_2$ such that 
$\pi_1(C) \cap F_1 \neq \emptyset$
and $\pi_2(C) \not\in \mathcal{F}$.
If such a $C$ is found, the algorithm calls 
$\alg{Witness}(C,\aut{M})$, which returns $u(v)^\omega$ such that
$\infss{\aut{M}}(u(v)^\omega) = C$.
Because $\infss{\M_1}(u(v)^\omega) = \pi_1(C)$ and 
$\pi_1(C) \cap F_1 \neq \emptyset$, $u(v)^{\omega} \in \sema{\aut{B}}$, and
because $\infss{\M_2}(u(v)^\omega) = \pi_2(C)$ 
and $\pi_2(C) \not\in \mathcal{F}$, $u(v)^{\omega} \notin \sema{\aut{U}}$.
The details are given in Algorithm~\ref{alg:DBAinDMA}.

\begin{algorithm}
{\small{
	\caption{$\alg{DBAinDMA}$}\label{alg:DBAinDMA}
	\begin{algorithmic}
		\Require{A DBA $\aut{B} = (\aut{M}_1,F_1)$ and a DMA $\aut{U} = (\aut{M}_2,{\mathcal  F})$}, where $Q_i$ is the state set of $\aut{M}_i$ for $i = 1,2$.
		\Ensure{$u(v)^{\omega} \in \sema{\aut{B}} \setminus \sema{\aut{U}}$ if such exists, else ``yes''.}
		
		\State{$\aut{M} = \aut{M}_1 \times \aut{M}_2$}
		\State ${\mathcal  C} \leftarrow \maxSCCs(Q_1 \times Q_2)$
		\For{each reachable $C \in \mathcal{C}$ such that $\pi_1(C) \cap F_1 \neq \emptyset$}
		\If{$\pi_2(C) \not\in {\mathcal  F}$}
		\State{\Return{$\alg{Witness}(C,\aut{M})$}}
		\Else
		\For{each $F \in {\mathcal  F}$ such that $F \subseteq \pi_2(C)$ and each $q \in F$}
		\State{$S \leftarrow \{(q_1,q_2) \in Q_1 \times Q_2 \mid q_1 \in \pi_1(C) \wedge q_2 \in F \setminus \{q\}\}$}
		\State{${\mathcal  D} \leftarrow \maxSCCs(S)$}
		\For{each $D \in {\mathcal  D}$}
		\If {$\pi_1(D) \cap F_1 \neq \emptyset$ and $\pi_2(D) \notin {\mathcal  F}$}
		\State{\Return{$\alg{Witness}(D,\aut{M})$}}
		\EndIf
		\EndFor
		\EndFor
		\EndIf
		\EndFor
		\State{\Return{``yes''}}		
	\end{algorithmic}
 }}
\end{algorithm}

\begin{theorem}
	\label{theorem:DBA-DMA-inclusion}
	The $\alg{DBAinDMA}$ algorithm runs in polynomial time and solves the inclusion problem
	for an arbitrary DBA $\aut{B}$ and an arbitrary DMA $\aut{U}$
	over the same alphabet.
\end{theorem}

\begin{proof}
	Suppose the returned value is a witness $u(v)^\omega$.
	Then the algorithm found
	an SCC $E$ with $\pi_1(E) \cap F_1 \neq \emptyset$ and
	$\pi_2(E) \not\in \mathcal{F}$ and returned $\alg{Witness}(E,\aut{M})$.
	In this case, the returned value is correct.
	
	Suppose for the sake of contradiction that the algorithm incorrectly returns the answer ``yes'', that is, 
	there exists an $\omega$-word $w$ such that $w \in \sema{\aut{B}}$ and
	$w \not\in \sema{\aut{U}}$.
	Let $C'$ denote $\infss{\aut{M}}(w)$.
	Then because $w \in \sema{\aut{B}}$, $\pi_1(C') \cap F_1 \neq \emptyset$,
	and because $w \not\in \sema{\aut{U}}$, $\pi_2(C') \not\in \mathcal{F}$.
	
	Then $C'$ is a subset of a unique SCC $C \in \maxSCCs(Q_1 \times Q_2)$ and
	$\pi_1(C) \cap F_1 \neq \emptyset$.
	It must be that $\pi_2(C) \in \mathcal{F}$, because otherwise
	the algorithm would have returned $\alg{Witness}(C,\aut{M})$.
	Consider the collection
	\[R = \{F \in {\mathcal  F} \mid  \pi_2(C') \subseteq F \subseteq \pi_2(C)\},\]
	of all the $F \in {\mathcal  F}$ contained in $\pi_2(C)$ that contain $\pi_2(C')$.
	The collection $R$ is nonempty because $C' \subseteq C$, and
	therefore $\pi_2(C') \subseteq \pi_2(C)$, and $\pi_2(C) \in \mathcal{F}$,
	so at least $\pi_2(C)$ is in $R$.
	Let $F'$ denote a minimal element of $R$ in the subset ordering.
	
	Then $\pi_2(C') \subseteq F'$ but because $\pi_2(C') \not\in \mathcal{F}$,
	it must be that $\pi_2(C') \neq F'$.
	Thus, there exists some $q \in F'$ that is not in $\pi_2(C')$.
	When the algorithm considers this $F'$ and $q$, then
	because $\pi_2(C') \subseteq F' \setminus \{q\}$, $C'$ is
	contained in $R$ and therefore is a subset of a unique SCC $D$ in $\maxSCCs(R)$.
	
	Because $C' \subseteq D$, and $\pi_1(C') \cap F_1 \neq \emptyset$, we
	have $\pi_1(D) \cap F_1 \neq \emptyset$.
	Also, $\pi_2(C') \subseteq \pi_2(D) \subseteq F'$, but because
	$q \not\in \pi_2(D)$, $\pi_2(D)$ is a proper subset of $F'$.
	When the algorithm considers this $D$, 
	because $\pi_1(D) \cap F_1 \neq \emptyset$,
	it must find that $\pi_2(D) \in \mathcal{F}$, or else it
	will return $\alg{Witness}(D,\aut{M})$.
	But then $\pi_2(D)$ is in $R$ and is a proper subset of $F'$, 
	contradicting our choice of $F'$ as a minimal element of $R$.
	Thus, if the algorithm outputs ``yes'', this is a correct answer.
\end{proof}

Combining \Cref{theorem:DMA-DMA-reduced-to-DBA-DMA},
\Cref{theorem:DBA-DMA-inclusion}, 
and the reduction of equivalence to inclusion, we have the following.

\begin{theorem}
	\label{theorem:inclusion-equivalence-for-DMAs}
	There are polynomial time algorithms to solve the inclusion and
	equivalence problems for two arbitrary DMAs
	over the same alphabet.
\end{theorem}

\section{Computing the automaton \texorpdfstring{$\aut{M}_{\sim_L}$}{M\_\textasciitilde L}}
\label{sec:computing-right-congruence-automaton}

In this section we use polynomial time algorithms  to construct the automaton $\aut{M}_{\sim_L}$ of the right congruence relation $\sim_L$ of the language $L$ accepted by an acceptor $\aut{A}$ of one of the types DBA, DCA, DPA, DRA, DSA, or DMA.
This gives a polynomial time algorithm to test whether a given DBA (resp., DCA, DPA, DRA, DSA, DMA) is of type IBA (resp., ICA, IPA, IRA, ISA, IMA).

Recall that $\aut{A}^q$ is the acceptor $\aut{A}$
with the initial state changed to $q$.
If $q_1$ and $q_2$ are two states of $\aut{A}$,
testing the equivalence of $\aut{A}^{q_1}$ to $\aut{A}^{q_2}$
determines whether these two states have the same
right congruence class, and, if not, 
returns a witness $u(v)^\omega$ that is accepted 
from exactly one of the two states.
The following is a consequence of
Theorems~\ref{theorem:inclusion-equivalence-for-DPAs}, 
\ref{theorem:inclusion-equivalence-for-DRAs-DSAs}, and
\ref{theorem:inclusion-equivalence-for-DMAs}.

\begin{proposition}
	\label{prop:right-congruence-test}
	There is a polynomial time procedure to test 
	whether two states of an arbitrary
	DBA, DCA, DPA, DRA, DSA or DMA $\aut{A}$ have the same right congruence class, 
	returning the answer ``yes'' if they do, and
	returning ``no'' and a witness $u(v)^{\omega}$ accepted from
	exactly one of the states if they do not.
\end{proposition}

We now describe an algorithm
$\alg{RightCon}$ that takes as input a DBA (or DCA, DPA, DRA, DSA, or DMA) $\aut{A}$ accepting a language $L$ and returns a deterministic automaton $\aut{M}$ isomorphic to the right congruence automaton of $L$, i.e., $\aut{M}_{\sim_L}$.

\begin{algorithm} {\small{
	\caption{$\alg{RightCon}$}\label{alg:RightCon}
	\begin{algorithmic}
		\Require{An acceptor $\aut{A} = \la \Sigma, Q, q_{\iota}, \delta, \alpha \ra$ of type  DBA, DCA, DPA, DRA, DSA, or DMA.}
		\Ensure{A deterministic automaton $\aut{M}$ isomorphic to $\aut{M}_{\sim_L}$, where $L = \sema{\aut{A}}$.}
		
		\State{$Q' \leftarrow \{\varepsilon\}$}
		\State{$q_{\iota}' \leftarrow \varepsilon$}
		\State{$\delta'$ is initially undefined}
		\While{there exists $x \in Q'$ and $\sigma \in \Sigma$ such that $\delta'(x,\sigma)$ is undefined}
		\State{$q_1 \leftarrow \delta(q_{\iota},x\sigma)$}
		\If{there exists $y \in Q'$ such that $\sema{\aut{A}^{q_1}} = \sema{\aut{A}^{q_2}}$ for $q_2 = \delta(q_{\iota},y)$}
		\State{Define $\delta'(x,\sigma) = y$}
		\Else
		\State{$Q' \leftarrow Q' \cup \{x\sigma\}$}
		\State{Define $\delta'(x,\sigma) = x\sigma$}
		\EndIf
		\EndWhile
		\State{\Return{$\aut{M} = \la \Sigma, Q', q_{\iota}', \delta' \ra$}}        
		
	\end{algorithmic}
}}\end{algorithm}

Assume the input acceptor is $\aut{A} = \la \Sigma, Q, q_\iota, \delta, \alpha \ra$.
The $\alg{RightCon}$ algorithm constructs a deterministic automaton $\aut{M} = \la \Sigma, Q', q_\iota', \delta' \ra$ in which the states are elements of $\Sigma^*$ and $q_\iota' = \varepsilon$.
The set $Q'$ initially contains just $\varepsilon$, and $\delta'$ is completely undefined.

While there exists a word $x \in Q'$ and a symbol $\sigma \in \Sigma$
such that $\delta'(x,\sigma)$ has not yet been defined,
loop through the words $y \in Q'$ and ask whether the states
$\delta(q_\iota,x\sigma)$
and $\delta(q_\iota,y)$ have the same right congruence class
in $\aut{A}$.
If so, then define $\delta'(x,\sigma)$ to be $y$.
If no such $y$ is found, then the word $x\sigma$ is added as
a new state to $Q'$, and the transition $\delta'(x,\sigma)$ is defined to be $x\sigma$.

This process must terminate because the elements of $Q'$ represent distinct right congruence classes of $L$, and $\aut{M}_{\sim_L}$ cannot have more than $|Q|$ states.
When it terminates, the automaton $\aut{M} = \la \Sigma, Q', q_{\iota}', \delta' \ra$ is isomorphic to the right congruence automaton of $\aut{A}$, $\aut{M}_{\sim_L}$.

\begin{theorem}
	\label{theorem:right-con-algorithm}
	The $\alg{RightCon}$ algorithm with input an acceptor $\aut{A}$ (a DBA, DCA, DPA, DRA, DSA, or DMA) accepting $L$, runs in polynomial time and returns $\aut{M}$, a deterministic automaton isomorphic to $\aut{M}_{\sim_L}$,
\end{theorem}

To test whether a given DBA (resp., DCA, DPA, DRA, DSA, DMA) $\aut{A}$ is an IBA (resp., ICA, IPA, IRA, ISA, IMA), we run the $\alg{RightCon}$ algorithm on $\aut{A}$ and test the returned automaton $\aut{M}$ for isomorphism with the automaton of $\aut{A}$.  
(Note that isomorphism can be checked in polynomial time by gradually constructing a map $h$ between states of $\aut{A}$ to states of $\aut{M}$. Initially $h(q_\iota)=q'_\iota$. Assume $h(q)=q'$ and $\delta(q,\sigma)=p$, $\delta(q',\sigma)=p'$. If $h(p)$ is defined and is different from $p'$ return ``non-isomorphic'', otherwise set $h(p)=p'$. This is repeated until all states are mapped and have been tested with respect to each letter of the alphabet.)
If they are isomorphic, then $\aut{A}$ is an IBA (resp. ICA, IPA, IRA, ISA, IMA), otherwise it is not.
This proves the following.

\begin{theorem}
	\label{theorem:testing-acceptor-informativeness}
	There is a polynomial time algorithm to test whether a given DBA (resp., DCA, DPA, DRA, DSA, DMA) is an IBA (resp., ICA, IPA, IRA, ISA, IMA).
\end{theorem}

\section{Testing membership in $\class{IXA}$}
\label{sec:testing-membership-in-IX}

In the previous section we showed that there is a polynomial time algorithm to test whether a given DBA $\aut{B}$ is an IBA.
However, we can also ask the following harder question.  Given a DBA $\aut{B}$ that is not an IBA, is $\sema{\aut{B}} \in \class{IBA}$, that is, does there exist an IBA $\aut{B}'$ such that $\sema{\aut{B}'} = \sema{\aut{B}}$?
This section shows that there are such polynomial time algorithms for DBAs, DCAs, DPAs, DRAs, DSAs and DMAs.
The algorithms first compute the right congruence automaton $\aut{M} = \aut{M}_{\sim_L}$, where $L = \sema{\aut{A}}$, and then attempt to construct an acceptance condition $\alpha$ of the appropriate type such that $\sema{(\aut{M},\alpha)} = L$.

\subsection{Testing membership in $\class{IBA}$}
\label{ssec:testing-membership-in-IB}

We describe the algorithm $\alg{TestInIB}$ that takes as input a DBA $\aut{B}$ and returns an IBA accepting $\sema{\aut{B}}$ if $\sema{\aut{B}} \in \class{IBA}$, and otherwise returns ``no''.
By Claim~\ref{clm:basic-relations-between-omega-aut}~(\ref{claim:DBA-DCA-complement}), the case of a DCA is reduced to that of a DBA.

\begin{algorithm}
{\small{
	\caption{$\alg{TestInIB}$}\label{alg:TestInIB}
	\begin{algorithmic}
		\Require{A DBA $\aut{B}$.}
		\Ensure{If $\sema{\aut{B}} \in \class{IBA}$ then return an IBA accepting $\sema{\aut{B}}$, else return ``no''.}
		
		\State{$\aut{M} \leftarrow \alg{RightCon}(\aut{B})$}
		\State{$F \leftarrow \emptyset$}
		\For{each state $q$ of $\aut{M}$}
		\If{$\sema{(\aut{M},\{q\})} \subseteq \sema{\aut{B}}$}
		\State{$F \leftarrow F \cup \{q\}$}
		\EndIf
		\EndFor
		\If{$\sema{(\aut{M},F)} = \sema{\aut{B}}$}
		\State{\Return{$(\aut{M},F)$}}
		\Else
		\State{\Return{``no''}}
		\EndIf
		
	\end{algorithmic}
}}\end{algorithm}

\begin{theorem}
	\label{theorem:membership-in-IB}
	The algorithm $\alg{TestInIB}$ takes a DBA $\aut{B}$ as input, runs in polynomial time, and returns an IBA accepting $\sema{\aut{B}}$ if $\sema{\aut{B}} \in \class{IBA}$, and otherwise returns ``no''.
\end{theorem}

\begin{proof}
	The algorithm calls the $\alg{RightCon}$ algorithm, and also the inclusion and equivalence algorithms from \Cref{theorem:inclusion-equivalence-for-DBAs-DCAs}, which run in polynomial time in the size of $\aut{B}$.
	If the algorithm returns an acceptor, it is an IBA accepting $\sema{\aut{B}}$.  
	
	To see that the algorithm does not incorrectly return the answer ``no'', suppose $\aut{B}'$ is an IBA accepting $\sema{\aut{B}}$.  Then because $\aut{M}$ is isomorphic to $\aut{M}_{\sim_L}$, we may assume that $\aut{B}' = (\aut{M},F')$.
	For every state $q \in F'$, the inclusion query with $(\aut{M}, \{q\})$ will answer ``yes'', so $q$ will be added to $F$.  Thus, $F' \subseteq F$, and $\sema{(\aut{M},F)}$ subsumes $\sema{(\aut{M},F')}$.  Every state $q$ added to $F$ preserves the condition that $\sema{(\aut{M},F)}$ is a subset of $\sema{\aut{B}}$, so the final equivalence check will pass, and $(\aut{M},F)$ will be returned.
\end{proof}

\subsection{Testing membership in $\class{IPA}$}
\label{ssec:testing-membership-in-IP}

We describe the algorithm $\alg{TestInIP}$ that takes as input a DPA $\aut{P}$ and returns an IPA accepting $\sema{\aut{P}}$ if $\sema{\aut{P}} \in \class{IPA}$, and otherwise returns ``no''.

\begin{algorithm}
{\small{
	\caption{$\alg{TestInIP}$}\label{alg:TestInIP}
	\begin{algorithmic}
		\Require{A DPA $\aut{P}$.}
		\Ensure{If $\sema{\aut{P}} \in \class{IPA}$ then return an IPA accepting $\sema{\aut{P}}$, else return ``no''.}
		
		\State{$\aut{M} = \la \Sigma, Q, q_\iota, \delta \ra \leftarrow \alg{RightCon}(\aut{P})$}
		\State{Define $\kappa(q) = 0$ for all states $q \in Q$}
		\For{$k = 1$ to $|Q|$}
		\If{$\sema{(\aut{M},\kappa)} = \sema{\aut{P}}$}
		\State{\Return{$(\aut{M},\kappa)$}}
		\ElsIf{$k$ is odd}
		\While{$\sema{\aut{P}}$ is not a subset of $\sema{(\aut{M},\kappa)}$}
		\State{Let $u(v)^{\omega}$ be the returned witness}
		\State{Define $\kappa(q) = k$ for all $q \in \infss{\aut{M}}(u(v)^{\omega})$}
		\EndWhile
		\Else
		\While{$\sema{\aut{P}}$ is not a superset of $\sema{(\aut{M},\kappa)}$}
		\State{Let $u(v)^{\omega}$ be the returned witness}
		\State{Define $\kappa(q) = k$ for all $q \in \infss{\aut{M}}(u(v)^{\omega})$}
		\EndWhile
		\EndIf
		\EndFor
		\State{\Return{``no''}}
		
	\end{algorithmic}
}}\end{algorithm}

\begin{theorem}
	\label{theorem:membership-in-IP}
	The algorithm $\alg{TestInIP}$ takes a DPA $\aut{P}$ as input, runs in polynomial time, and returns an IPA accepting $\sema{\aut{P}}$ if $\sema{\aut{P}} \in \class{IPA}$, and otherwise returns ``no''.
\end{theorem}

\begin{proof}
	The algorithm calls the $\alg{RightCon}$ algorithm 
	and the inclusion and equivalence algorithms for DPAs from \Cref{theorem:inclusion-equivalence-for-DPAs}, which run in polynomial time in the size of $\aut{P}$.  Below we show that each \textbf{while} loop terminates after at most $|Q|$ iterations.
	If the algorithm returns an acceptor, then the acceptor is an IPA accepting $\sema{\aut{P}}$.  
	
	To see that the algorithm does not incorrectly return the answer ``no'', suppose $\aut{P}'$ is an IPA accepting $\sema{\aut{P}}$.
	We may assume that $\aut{P}' = (\aut{M}, \kappa^*)$, where $\kappa^*$ is the canonical coloring of $\aut{P}'$.
	We prove inductively that the final coloring $\kappa$ is equal to $\kappa^*$.  To do so, we consider the conditions after the \textbf{for} loop has been completed $\ell$ times: (1) if $\ell$ is even then $\sema{(\aut{M},\kappa)} \subseteq \sema{\aut{P}}$, and if $\ell$ is odd, then $\sema{\aut{P}} \subseteq \sema{(\aut{M}, \kappa)}$, and (2) for all $q \in Q$, if $\kappa^*(q) \le \ell$ then $\kappa(q) = \kappa^*(q)$, and if $\kappa^*(q) > \ell$ then $\kappa(q) = \ell$.
	
	The initialization of $\kappa(q) = 0$ for all $q \in Q$ implies that these two conditions hold for $\ell = 0$.  Suppose the conditions hold for some $\ell \ge 0$.  If the equivalence check at the start of the next iteration returns ``yes'' then the correct IPA $(\aut{M},\kappa)$ is returned.  Otherwise, $k = \ell+1$; we consider the cases of odd and even $k$.
	
	If $k$ is odd, then by condition (1), $\sema{(\aut{M},\kappa)} \subsetneq \sema{\aut{P}}$ and at least one witness $u(v)^{\omega}$ accepted by $\aut{P}$ and rejected by $\sema{(\aut{M},\kappa)}$ will be processed in the \textbf{while} loop.  Consider such a witness $u(v)^{\omega}$ and let $C = \infss{\aut{M}}(u(v)^{\omega})$.  Then because $\kappa(q) = \kappa^*(q)$ if $\kappa(q) \le \ell$, it must be that $\kappa^*(C) > \ell$ and $\kappa(C) = \ell$.  For all $q \in C$, $\kappa(q)$ is set to $k = \ell+1$, so at least one state changes $\kappa$-color from $\ell$ to $\ell$+1.  This can happen at most $|Q|$ times, so the \textbf{while} loop for this $k$ must terminate after at most $|Q|$ iterations.  No state $q$ with $\kappa^*(q) \le \ell$ has its $\kappa$-value changed, so when the \textbf{while} loop is terminated, we have that $\kappa^*(q) \le \ell$ implies $\kappa(q) = \kappa^*(q)$.
	
	Consider any state $q$ with $\kappa^*(q) \ge \ell+1$. By property (2), at the start of this iteration of the \textbf{for} loop, $\kappa(q) = \ell$. Referring to the canonical forest $\forest{F}^*$ for $\aut{P}'$, the state $q$ is in $\Delta(D)$ for some node $D$ of $\forest{F}^*$.  The node $D$ is a descendant (or possibly equal to) some node $C$ for which the states $q \in \Delta(C)$ all have $\kappa^*(q) = \ell+1$.  Thus, as long as the value of $\kappa(q)$ remains $\ell$, the SCC $C$ will have $\kappa(C)$ even and $\kappa^*(C)$ odd, and the \textbf{while} loop cannot terminate.  But we have shown that it does terminate, so after termination we must have $\kappa(q) = k = \ell+1$.  Thus, after this iteration of the \textbf{for} loop, property (2) holds for $\ell+1$.
	
	The case of even $k$ is dual to the case of odd $k$.  Because the range of $\kappa^*$ is $[0..j]$ for some $j \le |Q|$, the equivalence test must return ``yes'' before the \textbf{for} loop completes, at which point the IPA $(\aut{M},\kappa)$ is returned.
\end{proof}

\subsection{Testing membership in $\class{IRA}$}
\label{ssec:testing-membership-in-IR}

We describe the algorithm $\alg{TestInIR}$ that takes as input a DRA $\aut{R}$ and returns an IRA accepting $\sema{\aut{R}}$ if $\sema{\aut{R}} \in \class{IRA}$, and otherwise returns ``no''.
By Claim~\ref{clm:basic-relations-between-omega-aut}~(\ref{claim:DRA-DSA-complement}), the case of a DSA is reduced to that of a DRA.

We first show that given a DRA $\aut{R}$ such that $\sema{\aut{R}} \in \class{IRA}$, there is an IRA equivalent to $\aut{R}$ whose size is bounded by a polynomial in the size of $\aut{R}$.
\begin{lemma}
	\label{lemma:IRA-bound-from-DRA}
	Let $\aut{R}$ be a DRA in singleton normal form whose acceptance condition has $m$ pairs, and assume $\sema{\aut{R}} \in \class{IRA}$.  Let $\aut{M}$ be the right congruence automaton of $\sema{\aut{R}}$ with state set $Q$ and assume $|Q| = n$.  Then there exists an acceptance condition $\alpha$ in singleton normal form with at most $mn$ pairs such that $(\aut{M},\alpha)$ accepts $\sema{\aut{R}}$.
\end{lemma}

\begin{proof}
	Let $\aut{R} = (\aut{M}_1,\alpha_1)$, where all the states of $\aut{M}_1$ are reachable, and let the function $f$ map each state of $\aut{M}_1$ to the state of its right congruence class in $\aut{M}$.  It suffices to show that for each $(q,B) \in \alpha_1$ there exists an acceptance condition $\alpha'$ of $\aut{M}$ containing at most $n$ pairs such that $\sema{(\aut{M}_1,\{(q,B)\})} \subseteq \sema{(\aut{M},\alpha')} \subseteq \sema{\aut{R}}$.  Taking the union of these $\alpha'$ conditions for all $m$ pairs $(q,B) \in \alpha_1$ yields the desired acceptance condition $\alpha$ for $\aut{M}$.
	
	Because we assume $\sema{\aut{R}} \in \class{IRA}$, there exists an IRA $(\aut{M},\alpha_2)$ in singleton normal form that accepts $\sema{\aut{R}}$.  Given any $u(v)^{\omega}$ in $\sema{\aut{R}}$, let $C = \infss{\aut{M}_1}(u(v)^{\omega})$.  Then $f(C) = \infss{\aut{M}}(u(v)^{\omega})$ and there exists $(q',B') \in \alpha_2$ such that $q' \in f(C)$, and $f(C) \cap B' = \emptyset$.  Then also $\sema{(\aut{M},\{(q',Q \setminus f(C))\})} \subseteq \sema{\aut{R}}$.  To see this, consider any $u'(v')^{\omega}$ with $D = \infss{\aut{M}}(u'(v')^{\omega})$ and $q' \in D$ and $D \cap (Q \setminus f(C)) = \emptyset$.  Then $D$ is a subset of $f(C)$, $D \cap B' = \emptyset$, $u'(v')^{\omega}$ satisfies $(q',B')$, and $u'(v')^{\omega} \in \sema{\aut{R}}$.
	
	Given a pair $(q,B) \in \alpha_1$ the construction of the initially empty acceptance condition $\alpha'$ proceeds as follows.  Let $C_0$ be the maximum SCC of $\aut{M}_1$ that contains $q$ and contains no element of $B$.  If $C_0$ is empty, then $(\aut{M}_1,\{(q,B)\})$ does not accept any words, and the empty condition $\alpha'$ suffices.  If $C_0$ is nonempty, then there is an element $u(v)^{\omega}$ of $\sema{(\aut{M}_1,\{(q,B)\}}$ such that $C_0 = \infss{\aut{M}_1}(u(v)^{\omega})$ and there is a pair $(q_0,B_0)$ in $\alpha_2$ such that $q_0 \in f(C_0)$ and $B_0 \cap f(C_0) = \emptyset$.  We add the pair $(q_0, Q \setminus f(C_0))$ to $\alpha'$ and note that by the argument in the preceding paragraph, $\sema{(\aut{M},\alpha')} \subseteq \sema{\aut{R}}$.
	
	If $\sema{(\aut{M}_1,\{(q,B)\}} \subseteq \sema{(\aut{M},\alpha')}$ then $\alpha'$ is the desired acceptance condition.  If not, there exists a word $u(v)^{\omega}$ such that for $C = \infss{\aut{M}_1}(u(v)^{\omega})$ we have $q \in C$ and $C \cap B = \emptyset$, but either $q_0 \not\in f(C)$ or $f(C) \cap (Q \setminus f(C_0)) \neq \emptyset$.  Because $C_0$ is the maximum SCC of $\aut{M}_1$ containing $q$ and containing no element of $B$, we have $C \subseteq C_0$, so  $f(C) \subseteq f(C_0)$ and therefore $q_0 \not\in f(C)$.  Let $C_1$ be the maximum SCC $C$ of $\aut{M}_1$ such that $C \subseteq C_0$, $q \in C$, and $q_0 \not\in f(C)$.  This is not empty, so there is a word $u'(v')^{\omega}$ such that $C_1 = \infss{\aut{M}_1}(u'(v')^{\omega})$, which is in $\sema{\aut{R}}$ because it satisfies $(q,B)$.  Thus there exists a pair $(q_1,B_1)$ in $\alpha_2$ that is satisfied by $u'(v')^{\omega}$, and we add the pair $(q_1,Q \setminus f(C_1))$ to the acceptance condition $\alpha'$.  As above, we have $\sema{(\aut{M},\alpha')} \subseteq \sema{\aut{R}}$.
	
	If now $\sema{(\aut{M}_1,\{(q,B)\})} \subseteq \sema{(\aut{M},\alpha')}$, then $\alpha'$ is the desired acceptance condition.  If not, we repeat this step again.  In general, after $k$ steps of this kind, $\alpha'$ consists of $k$ pairs of the form $(q_i,Q \setminus f(C_i))$ for $i \in [0..k-1]$, where all of the states $q_i$ are distinct and $C_{i+1} \subseteq C_i$ for $i \in [0,k-2]$.  Because $Q$ has $n$ states, there can only be $n$ repetitions of this step before $\alpha'$ satisfies the required condition, and thus $\alpha'$ has at most $n$ pairs.
\end{proof}

The algorithm $\alg{TestInIR}$ is based on the algorithm to learn Horn sentences by Angluin, Frazier, and Pitt~\cite{AngluinFP92}, using the analogy between singleton normal form for Rabin automata and propositional Horn clauses.  $\alg{TestInIR}$ maintains for each state $q$ of the right congruence automaton $\aut{M}$ an ordered sequence $S_q$ of SCCs of $\aut{M}$, each of which corresponds to a positive example of $\sema{\aut{R}}$.  At each iteration, the algorithm uses these sequences and inclusion queries with $\sema{\aut{R}}$ to construct an acceptance condition $\alpha$ for a hypothesis $(\aut{M},\alpha)$, which it tests for equivalence to $\aut{R}$.  In the case of non-equivalence, the witness is a positive example of $\sema{\aut{R}}$ that is used to update the sequences $S_q$.

In $\alg{TestInIR}$ the test of whether $C \cup C_i$ is positive is implemented by calling $\alg{Witness}(C \cup C_i)$ and testing the resulting word $u(v)^{\omega}$ for membership in $\sema{\aut{R}}$.

\begin{algorithm}
{\small{
	\caption{$\alg{TestInIR}$}\label{alg:TestInIR}
	\begin{algorithmic}
		\Require{A DRA $\aut{R} = (\aut{M}_1,\alpha_1)$ in singleton normal form with $|\alpha_1| = m$.}
		\Ensure{If $\sema{\aut{R}} \in \class{IRA}$ then return an IRA accepting $\sema{\aut{R}}$, else return ``no''.}
		
		\State{$\aut{M} \leftarrow \alg{RightCon}(\aut{R})$}
		\State{Let $Q$ be the states of $\aut{M}$ and $n = |Q|$}
		\State{For each $q \in Q$ initialize a sequence $S_q$ to be empty}
		\For{$k = 1$ to $mn^3$}
		\For{all $q \in Q$}
		\State{$\alpha_q = \emptyset$}
		\For{all $C \in S_q$}
		\For{all $q' \in C$}
		\If{$\sema{(\aut{M},\{(q',Q \setminus C)\}} \subseteq \sema{\aut{R}}$}
		\State{$\alpha_q = \alpha_q \cup \{(q', Q \setminus C)\}$}
		\EndIf
		\EndFor
		\EndFor
		\EndFor
		\State{$\alpha \leftarrow \bigcup_{q \in Q} \alpha_q$}
		\If{$\sema{(\aut{M},\alpha)} = \sema{\aut{R}}$}
		\State{\Return{$(\aut{M},\alpha)$}}
		\Else
		\State{Let $u(v)^{\omega}$ be the witness returned}
		\State{Let $C = \infss{\aut{M}}(u(v)^{\omega})$}
		\For{all $q \in C$}
		\If{there is some $C_i \in S_q$ such that $C  \not\subseteq C_i$ and $C \cup C_i$ is positive}
		\State{Let $i$ be the least such $i$ and replace $C_i$ by $C_i \cup C$}
		\Else
		\State{Add $C$ to the end of the sequence $S_q$}
		\EndIf
		\EndFor
		\EndIf
		\EndFor
		\State{\Return{``no''}}
		
	\end{algorithmic}
}}\end{algorithm}

Because pairs are only added to $\alpha$ that preserve inclusion in $\sema{\aut{R}}$, it is clear that any witness $u(v)^{\omega}$ returned in response to the test of equivalence of $(\aut{M},\alpha)$ and $\aut{R}$ is a positive example of $\sema{\aut{R}}$.
Note also that all elements of $S_q$ are SCCs of $\aut{M}$ that contain $q$.
The proof of correctness and running time of $\alg{TestInIR}$ depends on the following two lemmas.

\begin{lemma}
	\label{lemma:IRA-technical-lemma-1}
	Assume that $\aut{R}'$ is an IRA equivalent to the target DRA $\aut{R}$.
	Consider a positive example $u(v)^{\omega}$ of $\sema{\aut{R}}$ returned in response to the test of equivalence of $(\aut{M},\alpha)$ and $\aut{R}$, and let $C = \infss{\aut{M}}(u(v)^{\omega})$.  Let $(q',B)$ be a pair of $\aut{R}'$ such that $q' \in C$ and $C \cap B = \emptyset$. If for some $q \in C$ and some $C_i$ in $S_q$ we have $C_i \cap B = \emptyset$ then for some $j \le i$, the element $C_j$ of $S_q$ will be replaced by $C_j \cup C$.
\end{lemma}

\begin{proof}
	Assume that there is no such replacement for $j < i$.  When $i$ is considered, we have $q \in C_i$ and $q \in C$, so $C_i \cup C$ is the union of overlapping SCCs and therefore an SCC.  Then $C_i \cup C$ is positive because $q' \in C$, so $q' \in C_i \cup C$, and $C \cap B = \emptyset$ and $C_i \cap B = \emptyset$ by hypothesis, so $(C_i \cup C) \cap B = \emptyset$ and the word $\alg{Witness}(C_i \cup C)$ is a positive example of $\sema{\aut{R}}$ because it satisfies $(q',B)$.
	
	To see that $C \not\subseteq C_i$, we assume to the contrary.  Then $(q',Q \setminus C_i)$ is an element of $\alpha$.  To see this, we show that $\sema{(\aut{M},\{(q',Q \setminus C_i)\}} \subseteq \sema{\aut{R}}$.  Let $u'(v')^{\omega}$ with $D = \infss{\aut{M}}(u'(v')^{\omega})$ satisfy $(q',Q \setminus C_i)$.  Then $q' \in D$ and $D \cap (Q \setminus C_i) = \emptyset$, and therefore $D \subseteq C_i$ and $C_i \cap B = \emptyset$ by hypothesis.  Thus $u'(v')^{\omega}$ satisfies $(q',B)$ and is in $\sema{\aut{R}}$.  Because $(q',Q \setminus C_i)$ is an element of $\alpha$, $u(v)^{\omega}$ is accepted by $(\aut{M},\alpha)$ because $q' \in C$ and $C \cap (Q \setminus C_i) = \emptyset$ (because we assume $C \subseteq C_i$).  But this means that $u(v)^{\omega}$ cannot be a witness to the non-equivalence of $(\aut{M},\alpha)$ and $\aut{R}$, a contradiction.
	
	Thus, the conditions for $C_i$ to be replaced by $C_i \cup C$ are satisfied.
\end{proof}

\begin{lemma}
	\label{lemma:IR-technical-lemma-2}
	Assume that $\aut{R}' = (\aut{M},\alpha')$ is an IRA in singleton normal form equivalent to the target DRA $\aut{R}$.
	The following two conditions hold throughout the algorithm $\alg{TestInIR}$.
	\begin{enumerate}
		\item For all $q \in Q$, elements $C_i$ of $S_q$, and $(q',B) \in \alpha'$, if $C_i$ satisfies $(q',B)$ then for no $j < i$ do we have $C_j \cap B = \emptyset$.
		\item For all $q \in Q$, elements $C_i$ and $C_j$ of $S_q$ with $j < i$, and $(q',B) \in \alpha'$, if $C_i$ satisfies $(q',B)$ then $C_j$ does not satisfy $(q',B)$.
	\end{enumerate}
\end{lemma}

\begin{proof}
	We first show that condition (1) implies condition (2).  Let $q \in Q$, $C_i$ and $C_j$ be in $S_q$ with $j < i$ and $(q',B) \in \alpha'$.  If $C_i$ satisfies $(q',B)$ then by condition (1), we have $C_j \cap B \neq \emptyset$, so $C_j$ does not satisfy $(q',B)$.
	
	We now prove condition (1) by induction on the number of witnesses to non-equivalence.  The condition holds of the empty sequences $S_q$.  Suppose conditions (1) and (2) hold of the sequences $S_q$, and the witness to non-equivalence is $u(v)^{\omega}$ with $C = \infss{\aut{M}}(u(v)^{\omega})$.  If $q \not\in C$ then $S_q$ is not modified, so assume $q \in C$.  We consider two cases, depending on whether $C$ is added to the end of $S_q$ or causes some $C_\ell$ to be replaced by $C_\ell \cup C$.
	
	Assume that $C$ is added to the end of $S_q$ and property (1) fails to hold.  Then it must be that for $C_i = C$, some $C_j$ in $S_q$ with $j < i$ and some $(q',B) \in \alpha'$, $C_i$ satisfies $(q',B)$ and $C_j \cap B = \emptyset$.  By \Cref{lemma:IRA-technical-lemma-1}, because $C_j \cap B = \emptyset$, $C$ should cause $C_\ell$ for some $\ell \le j$ to be replaced by $C \cup C_\ell$ rather than being added to the end of $S_q$, a contradiction.
	
	Assume that $C$ causes $C_\ell$ in $S_q$ to be replaced by $C_\ell \cup C$ and property (1) fails to hold.  Then it must be for some $C_i$ in $S_q$ and pair $(q',B)$ in $\alpha'$, either (i) $i > \ell$ and $C_i$ satisfies $(q',B)$ and $(C_\ell \cup C) \cap B = \emptyset$, or (ii) $i < \ell$ and $C_\ell \cup C$ satisfies $(q',B)$ and $C_i \cap B = \emptyset$.
	In case (i), it must be that $C_i$ satisfies $(q',B)$ and $C_\ell \cap B = \emptyset$, which contradicts the assumption that property (1) holds before $C$ is processed, because $\ell < i$.
	In case (ii), $q' \in C_\ell \cup C$ and $(C_\ell \cup C) \cap B = \emptyset$.  Thus $C \cap B = \emptyset$ and $C_\ell \cap B = \emptyset$.  If $q' \in C_\ell$, then $C_\ell$ satisfies $(q',B)$, violating the assumption that property (2) holds before $C$ is processed.  If $q' \in C$ then $C$ satisfies $(q',B)$ and because $C_i \cap B = \emptyset$, by \Cref{lemma:IRA-technical-lemma-1}, for some $j \le i$ we have $C_j$ replaced by $C_j \cup C$, a contradiction because $\ell > i$.  Thus, in either case property (1) holds after $C_\ell$ is replaced by $C_\ell \cup C$.
\end{proof}

\begin{theorem}
	\label{theorem:membership-in-IR}
	The algorithm $\alg{TestInIR}$ takes as input a DRA $\aut{R}$, runs in polynomial time, and returns an IRA accepting $\sema{\aut{R}}$ if $\sema{\aut{R}} \in \class{IRA}$, and otherwise returns ``no''.
\end{theorem}

\begin{proof}
	Assume the input is a DRA $\aut{R} = (\aut{M}_1,\alpha_1)$ in singleton normal form with $|\alpha_1| = m$.  The algorithm $\alg{TestInIR}$ computes the right congruence automaton $\aut{M}$ of $\sema{\aut{R}}$, which has $n$ states, at most the number of states of $\aut{R}$.  The main loop of the algorithm is executed at most $mn^3$ times, and each execution makes calls to the inclusion and equivalence algorithms for DRAs, and runs in time polynomial in the size of $\aut{R}$, so the overall running time of $\alg{TestInIR}$ is polynomial in the size of $\aut{R}$.
	
	Clearly, if $\sema{\aut{R}}\not\in \class{IRA}$, then the test of equivalence between $(\aut{M},\alpha)$ and $\aut{R}$ will not succeed, and the value returned will be ``no''.
	Assume that $\sema{\aut{R}} \in \class{IRA}$.  Then by \Cref{lemma:IRA-bound-from-DRA}, there is an IRA $\aut{R}' = (\aut{M},\alpha')$ in singleton normal form equivalent to $\aut{R}$ such that $|\alpha'| \le mn$.  For every state $q \in Q$, each member of the sequence $S_q$ satisfies some pair $(q',B)$ in $\alpha'$, and by \Cref{lemma:IR-technical-lemma-2}, no two members of $S_q$ can satisfy the same pair, so the length of each $S_q$ is bounded by $mn$. Each positive counterexample must either add another member to at least one sequence $S_q$ or cause at least one member of some sequence $S_q$ to increase in cardinality by $1$.  The maximum cardinality of any member of any $S_q$ is $n$, and the total number of sequences $S_q$ is $n$, so no more than $mn^3$ positive counterexamples can be processed before the test of equivalence between $(\aut{M},\alpha)$ and $\aut{R}$ succeeds and $(\aut{M},\alpha)$ is returned.
\end{proof}

\subsection{Testing membership in $\class{IMA}$}
\label{ssec:testing-membership-in-IM}

We describe the algorithm $\alg{TestInIM}$ that takes as input a DMA $\aut{U}$ and returns an IMA accepting $\sema{\aut{U}}$ if $\sema{\aut{U}} \in \class{IMA}$, and otherwise returns ``no''.

\begin{algorithm}
{\small{
	\caption{$\alg{TestInIM}$}\label{alg:TestInIM}
	\begin{algorithmic}
		\Require{A DMA $\aut{U} = \la \Sigma, Q, q_{\iota}, \delta, {\mathcal  F} \ra$.}
		\Ensure{If $\sema{\aut{U}} \in \class{IMA}$ then return an IMA accepting $\sema{\aut{U}}$, else return ``no''.}
		
		\State{$\aut{M} \leftarrow \alg{RightCon}(\aut{U})$}
		\State{${\mathcal  F}' \leftarrow \emptyset$}
		\While{$\sema{(\aut{M},{\mathcal  F}')} \neq \sema{\aut{U}}$}
		\State{Let $u(v)^{\omega}$ be the witness returned}
		\State{Let $C = \infss{\aut{M}}(u(v)^{\omega})$}
		\If{$u(v)^{\omega} \in \sema{\aut{U}}$}
		\State{${\mathcal  F}' \leftarrow {\mathcal  F}' \cup \{C\}$}
		\Else
		\State{\Return{``no''}}
		\EndIf
		\EndWhile
		\State{\Return{$(\aut{M},{\mathcal  F}')$}}
		
	\end{algorithmic}
}}\end{algorithm}

\begin{theorem}
	\label{theorem:membership-in-IM}
	The algorithm $\alg{TestInIM}$ takes as input a DMA $\aut{U}$, runs in polynomial time, and returns an IMA accepting $\sema{\aut{U}}$ if $\sema{\aut{U}} \in \class{IMA}$, and otherwise returns ``no''.
\end{theorem}

\begin{proof}
	The algorithm calls the $\alg{RightCon}$ algorithm and also the DMA equivalence algorithm from \Cref{theorem:inclusion-equivalence-for-DMAs}, which run in polynomial time.  
	When there is a witness $u(v)^{\omega}$ accepted by $\aut{U}$, there is a set $F \in {\mathcal  F}$ whose image in $\aut{M}$ is added to ${\mathcal  F}'$, so there can be no more such witnesses than the number of sets in ${\mathcal  F}$.  After this, there must be a successful equivalence test or a witness rejected by $\aut{U}$, either of which terminates the while loop.  Thus, the overall running time is polynomial in the size of $\aut{U}$.
	
	If the algorithm returns an acceptor $(\aut{M},{\mathcal  F}')$, then the acceptor is an IBA that accepts $\sema{\aut{U}}$.  To see that the algorithm does not incorrectly return the answer ``no'', assume that $\aut{U}'$ is an IBA accepting $\sema{\aut{U}}$.  We may assume that $\aut{U}' = (\aut{M}, {\mathcal  F}'')$, where ${\mathcal  F}''$ contains no redundant sets. Then the first witness will be a word accepted by $\aut{U}$ that will add an element of ${\mathcal  F}''$ to ${\mathcal  F}'$.  This continues until all the elements of ${\mathcal  F}''$ have been added to ${\mathcal  F}'$, at which point the while loop terminates with equivalence.
\end{proof}

\subsection{Variants of the testing algorithms}
\label{ssec:variants-of-testing-algorithms}

A variant of the task considered above is the following.
Given the right congruence automaton $\aut{M}$ of a language $\sema{\aut{B}}$ in $\class{IBA}$, and access to information from certain queries about $\sema{\aut{B}}$, learn an acceptance condition $\alpha$ such that $(\aut{M},\alpha)$ accepts $\sema{\aut{B}}$.
In the case of $\class{IBA}$, the algorithm $\alg{TestInIB}$ could be modified to perform this task using just equivalence queries with respect to $\sema{\aut{B}}$.
Similarly, equivalence queries would suffice in the case of $\class{IMA}$.
For $\class{IPA}$, subset and superset queries with respect to the target language would suffice.
And for $\class{IRA}$, $\alg{TestInIR}$ could be modified to use membership and equivalence queries with respect to the target language, relying on negative examples to remove incorrect pairs rather than using subset queries.

\subsection{Efficient teachability of the informative classes}
We can finally claim that the the informative classes are efficiently teachable.
\begin{theorem}
\label{theorem:informative-classes-efficiently-teachable-learnable}
    The classes \class{IBA}, \class{ICA}, \class{IPA}, \class{IMA}, \class{IRA} and \class{ISA}
    are efficiently teachable.
\end{theorem}

\begin{proof}
By Theorems \ref{thm:im:itptd}, \ref{theorem:limit-id-of-DBA-DCA}, \ref{theorem:ip-poly-id-lim}, and \ref{theorem:ir-is-poly-ident-in-limit} the classes \class{IMA},  \class{IBA}, \class{ICA}, \class{IPA}, \class{IRA} and \class{ISA} are identifiable in the limit using polynomial time and data. It remains to show that the characteristic samples can be computed in polynomial time for any acceptor in the class.

Let $\aut{A}$ be an acceptor of type DXA for $X\in\{B,C,P,R,S,M\}$. 
By Theorems \ref{theorem:membership-in-IB}, \ref{theorem:membership-in-IP}, \ref{theorem:membership-in-IR}, \ref{theorem:membership-in-IM}, there are polynomial time algorithms that return an equivalent acceptor $\aut{A'}$ in IXA if such exists and ``no'' otherwise.
By \Cref{theorem:char-samples-in-poly-time},  given an acceptor $\aut{A}'$ of type IBA, ICA, IPA, IRA, ISA, or IMA, the characteristic sample $T_L$ for $\aut{A}$ may be computed in polynomial time in the size of $\aut{A}$.
It follows that the classes \class{IMA},  \class{IBA}, \class{ICA}, \class{IPA}, \class{IRA} and \class{ISA} are efficiently teachable.
\end{proof}

\section{Discussion}
\label{sec:discussion}
We have provided general definitions and comparisons for characteristic samples, efficient learnability, efficient teachability, efficient teachability/learnability, and identifiability in the limit using polynomial time and data.
While efficient teachability implies identifiability in the limit using polynomial time and data, we have shown that the converse is not true if there is no polynomial time algorithm for integer factorization.
We have shown that if a class is efficiently teachable, then it is it efficiently teachable/learnable.

We then asked which classes of representations of regular $\omega$-languages are efficiently teachable. 
The non-deterministic acceptors $\class{NBA}$, $\class{NCA}$, $\class{NPA}$, $\class{NRA}$, $\class{NSA}$, and $\class{NMA}$ do not have polynomial size characteristic sets, and thus are neither efficiently identifiable in the limit with polynomial time and data nor efficiently teachable.
We have shown that the classes \class{M}2\class{MA} and  \class{SUBA} are efficiently teachable. 

Focusing on the classes of informative languages,  $\IB$, $\IC$, $\IP$, $\IR$, $\IS$ and $\IM$, we have shown that they are efficiently teachable.
To obtain these results we have given new polynomial time algorithms to test inclusion and equivalence for DBAs, DCAs, DPAs, DRAs, DSAs and DMAs.
We have given a polynomial time algorithm to compute the right congruence automaton $\aut{M}_{\sim_L}$ for a language $L$ specified by a DBA, DCA, DPA, DRA, DSA, or DMA.
This yields a polynomial time algorithm to test whether an acceptor $\aut{A}$ of type DBA is of type IBA, and similarly for acceptors of types DCA, DPA, DRA, DSA and DMA.
Moreover, we have given a polynomial time algorithm to test whether an acceptor $\aut{A}$ of type DBA accepts a language in the class $\class{IBA}$, and similarly for acceptors of types DCA, DPA, DRA, DSA and DMA.

The questions of whether the full deterministic classes $\class{DBA}$, $\class{DCA}$, $\class{DPA}$, $\class{DRA}$, $\class{DSA}$ and $\class{DMA}$ are efficiently teachable or identifiable in the limit using polynomial time and data remain open.
We note that Bohn and L\"{o}ding~\cite{BohnL22,BohnL23} have recently obtained interesting results for passive learning of \class{DBA} and \class{DPA}, but from characteristic samples of cardinality  that may be exponential in the size of the minimal representations in the worst case.
Thus these results do not settle the question of whether these classes are efficiently teachable or identifiable in the limit using polynomial time and data. 
Another intriguing open question is whether the classes $\class{IBA}$, $\class{ICA}$, $\class{IPA}$, $\class{IRA}$, $\class{ISA}$ and $\class{IMA}$ can be learned by polynomial time algorithms using membership and equivalence queries. However, as shown by Bohn and L{\"{o}}ding~\cite{BohnL21}, this question is not easier than whether the corresponding deterministic classes can be learned by polynomial time algorithms using membership and equivalence queries.

\section*{Acknowledgments}
We thank the anonymous reviewers for very constructive comments that greatly improved the quality of this paper. This work was partially supported by ISF grant 2507/21.

\bibliographystyle{alphaurl}
\bibliography{bib.bib}

\end{document}